\documentclass[acmsmall]{acmart}
\PassOptionsToPackage{table,xcdraw}{xcolor}

\usepackage[table]{xcolor}
\usepackage{xspace}
\usepackage{graphicx}
\usepackage{amsmath}
\usepackage{amsthm}
\usepackage{IEEEtrantools}
\usepackage[position=b,font=small]{subcaption}
\usepackage{balance}
\usepackage{tabularx}
\usepackage{tikz}
\usepackage{enumitem}
\usetikzlibrary{arrows,positioning,automata} 
\usepackage{hyperref}
\usepackage{wrapfig}
\usepackage{kbordermatrix}
\usepackage{wrapfig}
\usepackage{blkarray}

\theoremstyle{plain}

\newtheorem*{theorem*}{Theorem}
\newtheorem{theorem}{Theorem}

\newtheorem{innercustomthm}{Theorem}
\newenvironment{customthm}[1]
  {\renewcommand\theinnercustomthm{#1}\innercustomthm}
  {\endinnercustomthm}

\tikzset{
    >=stealth',
    punkt/.style={
           rectangle,
           rounded corners,
           draw=black, very thick,
           text width=6.5em,
           minimum height=2em,
           text centered},
    pil/.style={
           ->,
           thick,
           shorten <=2pt,
           shorten >=2pt,}
}

\usepackage[vlined, linesnumbered]{algorithm2e}
\SetArgSty{textrm}
\SetCommentSty{textrm}
\SetFuncSty{textrm}
\SetProcArgSty{textrm}
\SetProcNameSty{textrm}
\SetFuncArgSty{textrm}
\SetNlSty{}{}{:}
\SetVlineSkip{3pt}    
\SetInd{0.5em}{0.7em}
\SetAlCapNameFnt{\small}
\SetAlCapFnt{\small}

\AtBeginDocument{%
  \providecommand\BibTeX{{%
    \normalfont B\kern-0.5em{\scshape i\kern-0.25em b}\kern-0.8em\TeX}}}

\setcopyright{acmcopyright}
\copyrightyear{2018}
\acmYear{2018}

\definecolor{success}{rgb}{0.6,0.9,0.6}
\definecolor{failure}{rgb}{0.9,0.6,0.6}
 \newcommand{\pn}{Recorp\xspace}

\newcommand{\pullop}{\textsf{pull}\xspace}
\newcommand{\sleepop}{\textsf{sleep}\xspace}
\newcommand{\pullops}{\textsf{pulls}\xspace}

\newcommand{\srv}[1]{$srv(#1)$\xspace}

\newcommand{\servicelist}{\textsf{service list}\xspace}

\newcommand{\activelist}{\textsf{active list}\xspace}
\newcommand{\activeflows}{\relax\ifmmode \mathcal{A}\else $\mathcal{A}$\xspace\fi}
\newcommand{\tmat}[1]{\relax\ifmmode \mathcal{M}_{#1}\else $\mathcal{M}_{#1}$\xspace\fi}
\newcommand{\tmatlb}[1]{\relax\ifmmode \mathcal{\widehat{M}}_{#1}\else $\mathcal{\widehat{M}}_{#1}$\xspace\fi}

\newcommand{\allstates}{\relax\ifmmode \Psi \else $\Psi$\xspace\fi}

\newcommand{\reliability}[1]{\relax\ifmmode R_{#1} \else $R_{#1}$\xspace\fi}
\newcommand{\reliabilitylb}[1]{\relax\ifmmode \widehat{R}_{#1} \else $\widehat{R}_{#1}$\xspace\fi}
\newcommand{\linkq}[1]{\relax\ifmmode LQ_{#1} \else $LQ_{#1}$\xspace\fi}
\newcommand{\linkqt}[1]{\relax\ifmmode LQ_{#1}(t) \else $LQ_{#1}(t)$\xspace\fi}
\newcommand{\snapshot}[1]{\relax\ifmmode \boldsymbol{LQ}_{#1} \else $\boldsymbol{LQ}_{#1}$\xspace\fi}
\newcommand{\pstate}[1]{\relax\ifmmode \boldsymbol{P}_{#1} \else $\boldsymbol{P}_{#1}$\xspace\fi}
\newcommand{\pstatelb}[1]{\relax\ifmmode \boldsymbol{\widehat{P}}_{#1} \else $\boldsymbol{\widehat{P}}_{#1}$\xspace\fi}

\newcommand{\pull}[2]{PL$_{#1}$(#2)}
\newcommand{\tx}[2]{TX$_{#1}$(#2)}
\newcommand{\flowpath}[1]{\relax\ifmmode \Gamma_{#1} \else $\Gamma_{#1}$\xspace\fi}

\newcommand{\builder}{\textsf{builder}\xspace}
\newcommand{\evaluator}{\textsf{evaluator}\xspace}

\newcommand{\policy}{\relax\ifmmode \pi \else$\pi$\xspace\fi}

\newcommand{\success}{\textsf{S}\xspace}
\newcommand{\failure}{\textsf{F}\xspace}

\newcommand{\psuccess}{\relax\ifmmode P_{s}\else $P_{s}$\xspace\fi}
\newcommand{\pfailure}{\relax\ifmmode P_{f} \else $P_{f}$\xspace\fi}
\newcommand{\pmin}{\relax\ifmmode m \else$m$\xspace\fi}
\newcommand{\pminactual}{\relax\ifmmode \bar{P}_{m} \else$\bar{P}_{m}$\xspace\fi}

\newcommand{\nodes}{\relax\ifmmode \mathcal{N} \else$\mathcal{N}$\xspace\fi}
\newcommand{\edges}{\relax\ifmmode \mathcal{E} \else$\mathcal{E}$\xspace\fi}
\newcommand{\flows}{\relax\ifmmode \mathcal{F} \else$\mathcal{F}$\xspace\fi}

\begin{document}

\title{\pn: Receiver-Oriented Policies for Industrial Wireless Networks}

\author{Ryan Brummet}
\email{ryan-brummet@uiowa.edu}
\orcid{1234-5678-9012}
\author{Md Kowsar Hossain}
\email{mdkowsar-hossain@uiowa.edu}
\author{Octav Chipara}
\email{octav-chipara@uiowa.edu}
\author{Ted Herman}
\email{ted-herman@uiowa.edu}
\author{David E.~Stewart}
\email{david-e-stewart@uiowa.edu}
\affiliation{%
  \institution{University of Iowa}
}

%
%
%
%
%
%

\renewcommand{\shortauthors}{Brummet, et al.}

\begin{abstract}
      Future Industrial Internet-of-Things (IIoT) systems will require wireless solutions to connect sensors, actuators, and controllers as part of high data rate feedback-control loops over real-time flows. A key challenge in such networks is to provide predictable performance and adaptability in response to  link quality variations. We address this challenge by developing RECeiver ORiented Policies (Recorp), which leverages the stability of IIoT workloads by combining offline policy synthesis and run-time adaptation. Compared to schedules that service a single flow in a slot, Recorp policies share slots among multiple flows by assigning a coordinator and a list of flows that may be serviced in the same slot. At run-time, the coordinator will execute one of the flows depending on which flows the coordinator has already received. A salient feature of Recorp is that it provides predictable performance: a policy meets the end-to-end reliability and deadline  of flows when the link quality exceeds a user-specified threshold. Experiments show that across IIoT workloads, policies provided a median increase of 50\% to 142\% in real-time capacity and a median decrease of 27\% to 70\% in worst-case latency when schedules and policies are configured to meet an end-to-end reliability of 99\%.
\end{abstract}



\maketitle

\section{Introduction}
\label{sec:introduction}

Industrial Internet-of-Things (IIoT) systems are gaining rapid adoption in process control industries such as oil refineries, chemical plants, and factories. 
In contrast to prior work that has focused primarily on low-data rate or energy-efficient applications, we are interested in the next generation of smart factories expected to use sophisticated powered sensors such as cameras, microphones, and accelerometers (e.g., \cite{lynch2008implementation, correll2017wireless, hayat2016survey}). 
Since such applications will require higher data rates,  we need to develop a versatile wireless solution to connect them with actuators and controllers as part of feedback-control loops over multihop \emph{real-time flows}.
A practical solution must meet the following two requirements: (1) must support high data rates,  and (2) support real-time communication over multiple hops.
Both requirements must be met, notwithstanding significant variations in the quality of wireless links common in harsh industrial environments~\cite{candell2017industrial, emInterference2}.

\begin{wrapfigure}{r}{0.45\textwidth}
    \centering
    \includegraphics[width=\linewidth]{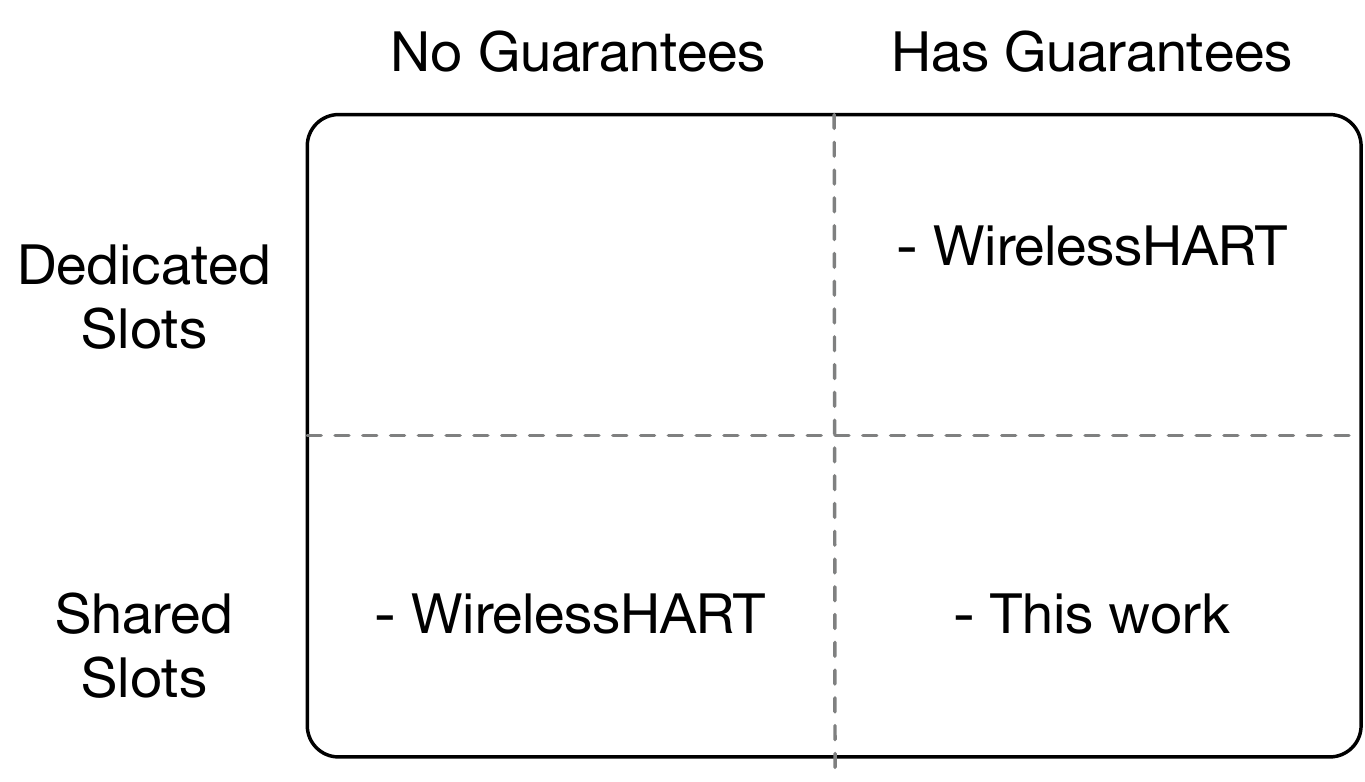}
    \caption{Design space of wireless control solutions.}
    \label{fig:related-work}
\end{wrapfigure}


Let's examine whether WirelessHART can meet the requirements of the next generation of IIoT applications. WirelessHART  is the state-of-the-art standard for industrial wireless communication and has successfully provided high reliability in a broad range of industrial settings. 
At the heart of WirelessHART is Time Slotted Channel Hopping (TSCH) ---  a MAC layer that combines Time Division Multiple Access (TDMA) and channel-hopping in a mesh network.
The TSCH data plane relies on a centralized network manager to generate routes and a transmission schedule for all the flows in the network.
The schedule is represented as a two-dimensional scheduling matrix that specifies the time and frequency of each transmission.
TSCH supports both real-time and best-effort traffic by using two scheduling strategies whose trade-offs are shown in Figure \ref{fig:related-work}.

To support real-time traffic, a transmission is assigned to a \emph{dedicated entry} in the scheduling matrix and, at run-time, the transmission is performed without contention.
The reliability of real-time traffic is ensured by using retransmissions and channel hopping.
The number of retransmissions allocated is usually determined based on the worst-case quality of a link as required to tolerate significant variations in link quality. 
Since the scheduling matrix cannot be updated at the rate at which the link quality varies, the only run-time adaption available is to cancel  a link's retransmissions when an acknowledgment is received.
As a result, it is common for a significant number of slots to remain unused when a packet is relayed successfully to the next hop before exhausting a link's allocated retransmissions.
As demonstrated by the experiments in Section \ref{sec:experiments}, protocols that use dedicated entries cannot handle higher data rates efficiently.

In contrast, best-effort traffic is supported by having multiple transmissions assigned to a \emph{shared entry} in the scheduling matrix.
At run-time, contention-based techniques are used to arbitrate which transmissions will be performed.
Shared entries provide more opportunities for locally adapting what transmissions may be performed, resulting in more efficient use of network resources.
Unfortunately, there are no current techniques to effectively analyze the network's performance when shared entries are used.
The open research question, and the focus of this paper, 
is whether \emph{it is possible to use shared entries to support higher throughput and respond more effectively to changes induced by link quality variations while providing performance guarantees.}

To answer this question, we propose \emph{\textbf{REC}eiver \textbf{OR}iented \textbf{P}olicies (\pn)} -- a new data plane that provides higher performance and agility than traditional scheduling solutions that do not use shared entries.
We exploit the typical characteristics of the industrial setting to obtain improvements in network capacity and latency while providing predictability under prescribed link variability. 
Specifically, our approach has the following features:
\begin{itemize}
    \item Since IIoT workloads consist of sets of real-time flows that are stable for long periods of time, we compute offline \pn policies and disseminate them to all nodes.
    \pn policies assign a coordinator and list of candidate flows for each entry in the scheduling matrix.
   Only one of the candidate flows will be executed at run-time, depending on which flows the coordinator has already received.
    The benefit of this approach is that it enables flows to be dynamically executed in an entry depending on the successes and failures of transmissions observed at run-time.
    As a consequence, \pn policies can handle variations in link quality more effectively than schedules.
    \item We propose a novel link model in which the quality of the links can vary \emph{arbitrarily} within an interval from slot to slot.
    Our model is motivated by current guidelines for deploying wireless IIoT networks (e.g., \cite{emerson-deployment-guide}), focusing on ensuring that communication links have a minimum link quality.
    The proposed model is well-suited to industrial settings where link quality may vary widely over short time scales.
    \item
    In contrast to best-effort entry sharing approaches that provide no performance guarantees, we ensure that a constructed \pn policy will meet a user-specified reliability and deadline constraint for each flow as long as the quality of all (used) links exceeds a minimum link quality as specified by our model. 
\end{itemize}

We demonstrate the effectiveness of \pn through testbed measurements and simulations.
When schedules and \pn policies are configured to meet the same target end-to-end reliability of 99\%,
empirical results show that \pn policies increased the median real-time capacity by 96\% for a data collection workload.
Furthermore, the performance bounds derived analytically were safe: \pn policies met all end-to-end reliability and deadline constraints when the minimum link quality exceeded a user-specified level of 70\%.  
Larger-scale multihop simulations that consider two topologies indicate that across typical IIoT workloads, policies provided a median increase of 50\% to 142\% in the real-time capacity as well as a median decrease of 27\% to 70\% in worst-case latency.

The remainder of the paper is organized as follows. Section \ref{sec:problem-formulation} describes the problem formulation and informally introduces \pn policies and the challenges associated with their synthesis. Section \ref{sec:system-model} describes \pn's system and network models, while Section \ref{sec:reliability-model} introduces \pn's reliability model. \pn is described in Section \ref{sec:design}. Simulation and testbed experiments evaluating \pn's performance against representative protocols using both dedicated and shared entries are included  in Section \ref{sec:experiments}. Section \ref{sec:discussion} describes how \pn handles network dynamics, aperiodic traffic, and energy efficiency. \pn is placed in the context of related work in Section \ref{sec:related-work}. We conclude the paper in Section \ref{sec:conclusions}.

\section{Problem Formulation}
\label{sec:problem-formulation}
In this section, we start by considering the problem of building real-time protocols from a fresh perspective, discuss how this perspective opens new opportunities for optimization, and then informally introduce \pn policies while highlighting the challenges of their synthesis.

\textbf{Optimization Problem:}
We consider supporting real-time and reliable communication as a sequential decision problem.
In each slot, the offline policy synthesis procedure uses the current estimate of the network state to select the actions performed in the current slot.
%
Then, the estimated network state is updated to reflect the impact of those actions.
In this paper, we limit our attention to myopic (or greedy) policies that maximize the number of flows that may be executed in a slot while providing prioritization based on the flows' statically assigned priorities.
A myopic policy selects the optimal actions over a time horizon of one slot, but those decisions may be suboptimal over longer horizons.
Our choice is motivated by the simplicity of myopic policies which can be synthesized efficiently. 
The unique aspects of \pn policies are what actions may be performed in a slot and how the network state is represented.

\textbf{Intuition:}
Schedules and policies differ in the information they use as part of the offline scheduling and synthesis process.
Consider a star topology with three nodes where the base station is the receiver of two incoming flows $F_0$ and $F_1$.
Both flows are released at the beginning of slot 0 with
flow $F_0$ having a higher priority than flow $F_1$.
Since $F_0$ and $F_1$ share the same receiver, only one of them can transmit in the first slot without conflict. 
In slot 0, both schedules and policies assign and execute (at run-time) $F_0$ to enforce prioritization.

Schedules and policies differ in how they account for the outcome of $F_0$'s transmission.
At run-time, the network is in one of two states, depending on the outcome of $F_0$'s transmission: 
either $F_0$'s data was relayed successfully to the base station, or it was not.
Scheduling approaches ignore this information and assign a fixed number of retransmissions for $F_0$, regardless of whether these retransmissions are successful or not at run-time.
However, when we capture both possible outcomes, there are new opportunities for optimization.
Ideally, we would like to transmit $F_1$ if $F_0$ has succeeded or otherwise retransmit $F_0$.
Surprisingly, we can achieve this behavior (which is impossible for scheduling approaches):
Offline, both flows $F_0$ and $F_1$ are assigned to be \emph{candidates} to be executed in slot $1$.
At run-time, the base station will track the flows from which it has received packets in the previous slots.
As a result, it will know whether $F_0$ was successful or not at the end of slot 0, i.e., the star network's precise state.
Using this information, in slot 1, the base station can request $F_1$'s packet if it has already received $F_0$ or, otherwise, it can request a retransmission for $F_0$.
We say that $F_0$ and $F_1$ \emph{share} slot 1 as either flow may execute at run-time depending on the observed successes and failures.

\textbf{Actions:}
This approach can be generalized to multi-hop scenarios by observing that any node with multiple flows routed through it can act as a coordinator for those flows, not just a base station in a star topology.
A \pn policy is represented as a matrix whose rows indicate channels and columns indicate slots. 
In each entry of the matrix, a \pn policy may include, at most, one \pullop.
A \emph{\pullop} has two arguments: a \emph{coordinator} and a \emph{service list}.
A \pullop is executed by a coordinator that can dynamically request data (i.e., a \texttt{pull}, henceforth) from a service list of flows depending on the outcome of previous transmissions.
The synthesis procedure determines the nodes that will be coordinators and the composition of the service list, both of which can change from slot to slot.
At run-time, a coordinator executing a \pullop requests the packet of the first flow in the service list from which it has not yet received the packet.
The adaptation mechanism is localized, lightweight, and does not require carrier sense.

\textbf{State Estimation:}
A challenge to synthesizing policies is to estimate the network's state as \pullops are performed.
Specifically, we need to know the likelihood that a flow's packet is located at a specific node in a given slot.
Knowing this information \emph{offline} is challenging because the quality of a link is probabilistic, and the likelihood of a successful transmission varies from slot to slot.
To address this challenge, we propose a Threshold Link Reliability (TLR) model.
We model the quality of a link \linkqt{i} in slot $t$ used by flow $i$ as a Bernoulli variable.
TLR allows the quality of the link to change arbitrarily from slot-to-slot as long as it exceeds a minimum value \pmin (i.e., $\linkq{i}(t) \ge \pmin~~~ \forall~t$).
We will show it is possible to provide guarantees on the performance of \pn when all links follow the TLR model.


\textbf{Scalability:}
Another significant challenge in synthesizing policies is avoiding the state explosion problem.
The critical decision is how to balance the trade-off between the expressiveness of policies, the performance improvements they provide, and the scalability of the synthesis procedure.
Cognizant of these trade-offs, we make two important design choices:
(1) We limit nodes to operating on their local states such that their decisions are independent of the state of other nodes.
As a consequence, the probabilities of packets being forwarded across links of a multi-hop flow are independent. 
This property reduces the number of states maintained during synthesis since it is sufficient to capture the interactions of flows locally at each node rather than globally across the network.
(2) The synthesis procedure incrementally constructs policies in a slot-by-slot manner using a \builder and an \evaluator.
The \builder casts the problem of determining the \pullops performed in the current slot as an Integer Linear Program (ILP).
In turn, the \evaluator applies each \pullop selected by the \builder to the system state and tracks the state as it evolves from slot to slot.
The iterative nature of the synthesis procedure improves its scalability as it suffices to maintain only the states associated with the current slot.

\section{System Model}
\label{sec:system-model}
We base our network model on WirelessHART as it is an open standard developed specifically for IIoT systems with stringent real-time and reliability requirements~\cite{wirelesshart}.
A network consists of a \emph{base station} and tens of \emph{field devices}. 
\pn is best-suited for applications that require high data rates and have a backbone of grid-powered nodes to carry this traffic (e.g., \cite{burns2018airtight,agarwal2011duty,chipara2010reliable}).

A centralized \emph{network manager} is responsible for synthesizing policies, evaluating their performance, and distributing them across the network.  
The field devices form a time-synchronized wireless mesh network that we model as a graph $G(\nodes, \edges)$, where \nodes and \edges represent the devices (including the base station) and wireless links.
The network can be   synchronized using a high-accuracy time synchronization protocol designed for wireless sensor networks (see \cite{timeSyncSurvey} for a survey).
We will initially assume that the communication graph remains fixed while each link has a minimum link quality.
In Section \ref{sec:discussion} we will discuss how \pn can handle node failures and topology changes, such as adding and removing nodes, by distributing new policies.
The network maintains two trees, an \emph{upstream tree} and \emph{downstream tree}, for packet routing to and from the base station, respectively. 
We assume that  both upstream and downstream trees are spanning trees consistent to source routing in WirelessHART.

At the physical layer, WirelessHART adopts the 802.15.4 standard with up to 16 channels.
This paper focuses on receiver-initiated communication, where a node requests data from a neighbor and receives a response within the same 10~$ms$ slot.

We use \emph{real-time flows} as a communication primitive.
The following parameters characterize a real-time flow $F_i$:
phase $\sigma_i$,  period $P_i$,  deadline $D_i$, end-to-end target reliability requirement $T_i$, and static priority $i$ where lower values have higher priority.
The k$^{th}$ instance of flow $F_i$, $J_{i, k}$, is released at time $r_{i,k} = \phi_i + k * P_i$ and has an absolute deadline $d_{i,k} = r_{i, k} + D_i$.
We assume $D_i \leq P_i$, which implies only one instance of a flow is released at a time. Consequently, to simplify the notation, we will use $J_i$ to refer to the instance of flow $F_i$ that is currently released.
The variable \flows denotes the set of flows in the network.
A flow $i$ has a forwarding path $\flowpath{i}$ that is used by all of its instances.
During the execution of an instance, only one of the links on the \flowpath{i} is active and considered for scheduling.
We will use the notation $\linkq{i}(t)$ to refer to the link quality of the currently active link at time $t$. 

A \pn policy \policy is a scheduling matrix whose number of slots is equal to the hyperperiod of the flow's periods.
The policy may be represented as a two-dimensional matrix such that the rows indicate channels, the columns indicate slots, and the entries that represent actions.
An action may be either a \pullop or a sleep.
A policy is well-formed if it satisfies the following constraints: 
(1) Each node transmits or receives at most once in an entry to avoid intra-network interference. 
(2) The hop-by-hop packet forwarding precedence constraints are maintained such that senders receive packets before forwarding them.
(3) Nodes do not perform consecutive transmissions using the same channel.
(4) Each flow instance is delivered to its destination before its absolute deadline and meets its reliability constraint.


\begin{table}[t!]
    \centering
    \begin{tabular}{c|c}
         Description &  Symbol \\
         \hline
         Set of nodes &  $\mathcal{N}$ \\
         \hline 
         Set of flows & $\mathcal{F}$ \\
         Flow i & $F_i$ \\ 
         Period of flow $i$ & $P_i$ \\
         Deadline of flow $i$ & $D_i$ \\
         Phase of flow $i$  & $\phi_i$ \\
         Target end-to-end reliability of flow $i$ & $T_i$ \\
         Path of flow $i$ & $\Gamma_i$ \\
         Quality of the active link of flow $i$ & \linkq{i} \\
         \hline
         Instance $k$ of flow $i$ & $J_{i,k}$  (or simply $J_i$) \\
         Release time of $J_{i,k}$ & $r_{i,k}$ \\
         Absolute deadline of $J_{i,k}$ & $d_{i,k}$ \\
         Link quality of the active link of $J_{i,k}$ & $\linkq{i}$ \\
         \hline 
         Policy & $\pi$ \\
         Service list of the \pullop in slot $t$ (and channel $c$)& \srv{t} (or \srv{t,c}) \\
         \hline
         Set of all possible states & \allstates \\
         Transition matrix & $\tmat{}$ \\ 
         Reliability of instance $J_{i}$ & \reliability{i} \\
         Lower-bound on reliability of $J_{i}$ & \reliabilitylb{i}
    \end{tabular}
    \caption{Summary of key notations.}
    \label{tab:notations}
\end{table}

\section{Reliability Model}
\label{sec:reliability-model}

The wireless communications community has developed a wide range of probabilistic models to model link quality (e.g., \cite{gilbert1960capacity,kamthe2009m}).
Examples span the complexity-accuracy trade-off from simple models such as Gilbert-Elliott \cite{gilbert1960capacity} to more complex models that use multi-level Markov Chains (e.g., \cite{kamthe2009m}) to distinguish between the short-term and long-term behavior of wireless links.
However, these models usually focus on the ``average-case'' behavior of links.
Guarantees on the end-to-end reliability of flows should hold even as links deviate from their average-case behavior.
Furthermore, a practical model must require little tuning, preferably having reasonable default values for its parameters that fit the rules-of-thumb engineers use to deploy real wireless networks. 

To address the above challenges, we propose the Threshold Link Reliability (TLR) model.
We model the likelihood that a single \pullop for flow $i$  is successful (including both the pull request and the response containing the data) as a Bernoulli variable \linkqt{i}.
We assume that consecutive pulls performed over the same or different links are independent.
Empirical studies suggest that this property holds when channel hopping is used~\cite{channelHopping1,channelHopping2}.
A minimum Packet Delivery Rate (PDR) \pmin lower bounds the values of \linkqt{i} such that $\pmin \le 
\linkqt{i}~~\forall i \in \flows, t \in \mathbb{N}$.
A strength of TLR is that aside from the lower bound $\pmin$ on link quality, we make \emph{no assumptions regarding how the quality of a link varies from slot to slot}.
This characteristic makes TLR widely applicable to networks experiencing significant link quality variations.
TLR can be integrated with existing guidelines for deploying IIoT wireless networks.
For example, Emerson engineers suggest that WirelessHART networks should be deployed to provide a minimum link quality between 60--70\%~\cite{emerson-deployment-guide}.
Accordingly, in this paper, we set \pmin to either 60\% or 70\%.


On a more technical note, it is important to note that TLR does not require the transmissions in an actual network deployment to be independent --
we only require that there is a TLR model that lower bounds the behavior of the deployed network.
Specifically, we require that a Bernoulli distribution \emph{lower bounds} the distribution of consecutive packet losses in the network.
Thus, by selecting an appropriate value for $\pmin$, it is possible to find a model for which the assumption of independence holds, albeit at the cost of increased pessimism regarding the quality of links.

The end-to-end reliability \reliability{i} of a flow $i$ depends on both the likelihood of successfully relaying a packet over the links of its path as well as the links of other flows it shares entries with.
For instance, returning to our running example, the probability the packet released by $F_1$ reaches its destination is dependent not only on the quality of its link but also $F_0$'s link since $F_1$ is conditionally attempted depending on the success of $F_0$.
One might assume that finding a lower bound on \reliability{i} under the TLR model only requires considering the case when all links exhibit their worst link quality in all slots (i.e., $\linkqt{i} = \pmin~~\forall i \in \flows, ~t \in \mathbb{N}$).
While we will show that this approach provides a safe lower bound for \pn policies, this property does not hold for \emph{all} policies that use shared slots.
Consider, for example, the two flows $F_0$ and $F_1$.
Suppose these flows are scheduled using the following simple (non-\pn) policy.
In the first slot, $F_0$ will be executed.
In the next slot, $F_1$ will be executed only if $F_0$ failed in the first slot; otherwise, the base station sleeps.
Under this policy, the probability that $F_1$ is attempted will decrease as the link quality increases since increasing the link's quality will increase the probability that $F_0$ is successful in the first slot.
As a consequence, the end-to-end reliability of $F_1$ will drop as the link becomes more reliable.
Therefore, for policies such as \pn that share slots, it is essential to prove that they do not exhibit such pathological behavior. 
Theorem \ref{th:monotonic} demonstrates that \pn policies do not exhibit this behavior.

\section{Design}
\label{sec:design}

 \begin{figure}[t!]
     \centering
     \includegraphics[width=.5\linewidth]{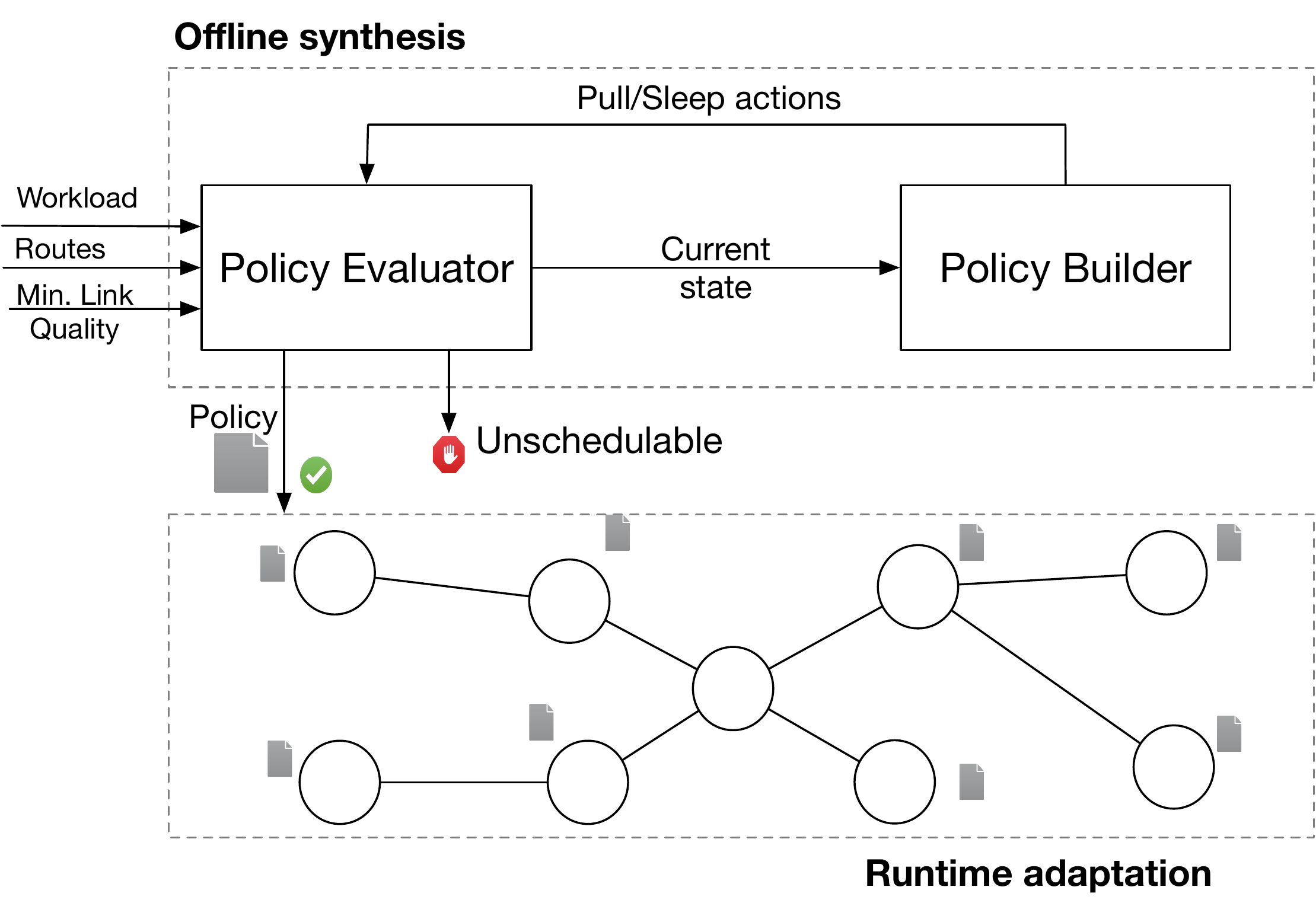}
     \caption{Design of \pn.}
     \label{fig:system-design}
 \end{figure}

\pn is a practical and effective solution for IIoT applications that require predictable, real-time, and reliable communication in dynamic wireless environments (see Figure~\ref{fig:system-design}).
Central to our approach is \pn policies.
The policy synthesis procedure runs on the network manager and has as inputs the workload, routing information, and a user-specified minimum link quality threshold \pmin.
If the synthesis procedure is successful, the constructed policy guarantees probabilistically that all flows will meet their real-time and reliability constraints as long as the quality of all links meets or exceeds \pmin.
The synthesis procedure fails when the workload is unschedulable, i.e., when a policy that meets both the real-time and reliability constraints of all flows cannot be found.
Note that this case is unlikely to arise in practice since an application's workload specification is known a priori, and the designer can validate that the workload remains unschedulable during the system's deployment.
If the synthesis procedure is successful, the manager disseminates the generated policy to all nodes.
During the operation of the network, some links may fall below the minimum link quality threshold \pmin. 
Since \pn provides no guarantees under this regime, a new policy should be constructed after either changing the flows' routes to avoid low-quality links or by lowering \pmin.

The separation between offline synthesis and run-time adaptation is essential to building agile networks.
The run-time adaptation is lightweight:
when a node is the coordinator of a \pullop, it can execute any of the flows included in its service list without requiring global consensus.
In contrast, policy synthesis is computationally expensive and ensures the global invariant that no transmission conflicts occur regardless of coordinators' local decisions.

We will formalize the semantics of \pn policies and discuss their run-time adaptation mechanism in Section \ref{sec:policies}.
After, we will consider synthesizing \pn policies in a scalable manner.
We will start by considering the problem of synthesizing policies for a data collection workload in a star topology in Section \ref{sec:single-hop}.
In Section \ref{sec:multi-hop}, we will extend our approach to handle general workloads and topologies.

\subsection{\pn Policies and Their Run-time Adaptation}
\label{sec:policies}

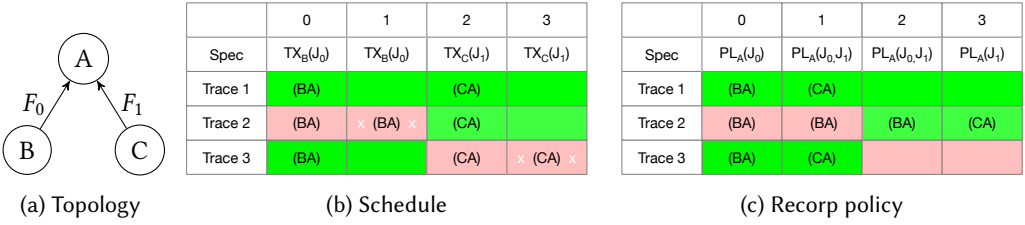
\begin{figure*}[ht!]
\centering
\begin{subfigure}[b]{0.15\linewidth}
\centering
\begin {tikzpicture}
\node[circle,draw=black] (A) at (.75, 1.2) {A};
\node[circle,draw=black] (B) at (0,  0) {B};
\node[circle,draw=black] (C) at (1.5, 0) {C};
\path[->]
(B) edge node[left] {$F_0$} (A);
\path[->]
(C) edge node[right] {$F_1$} (A);
\end{tikzpicture}
\caption{Topology}
\label{fig:topology-example}
\end{subfigure}%
\quad
\begin{subfigure}[b]{0.385\textwidth}
\centering
\includegraphics[width=\linewidth]{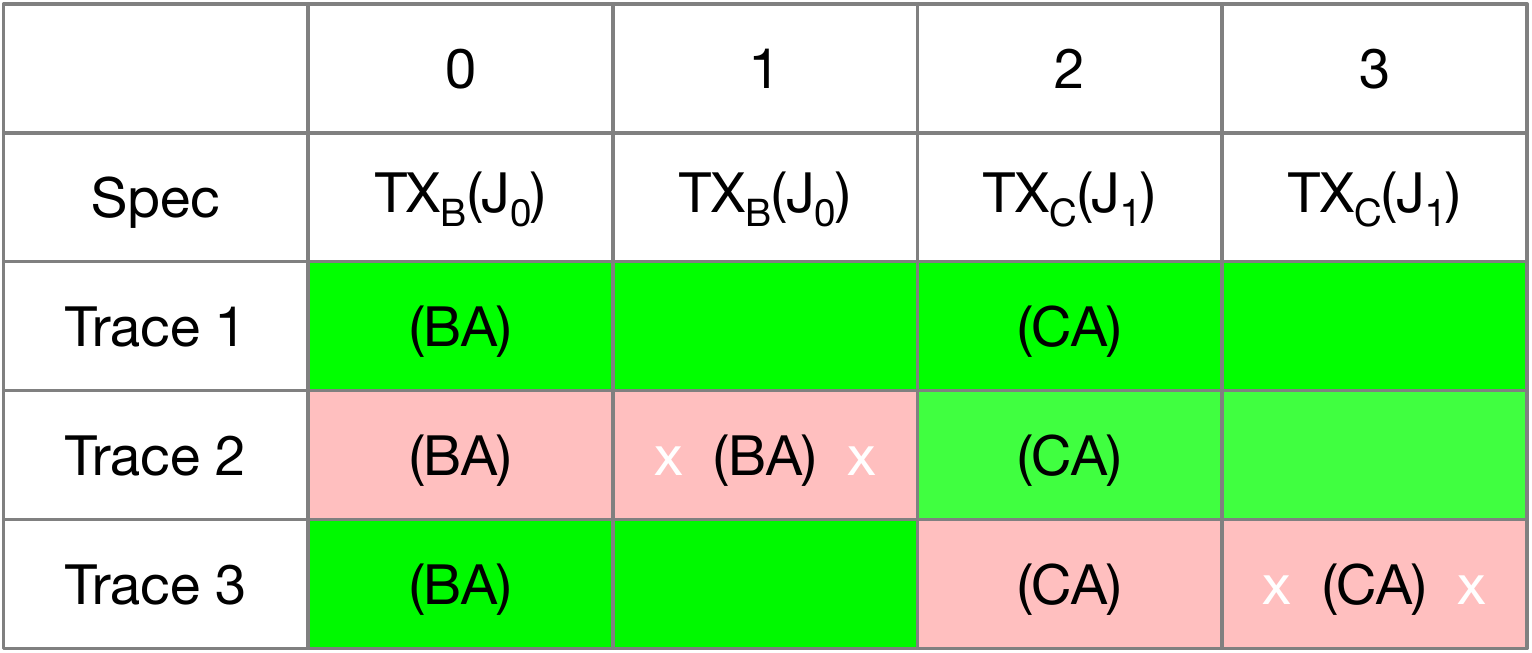}
\caption{Schedule}
\label{fig:example:schedule}
\end{subfigure}
\quad
\begin{subfigure}[b]{0.385\textwidth}
\includegraphics[width=\linewidth]{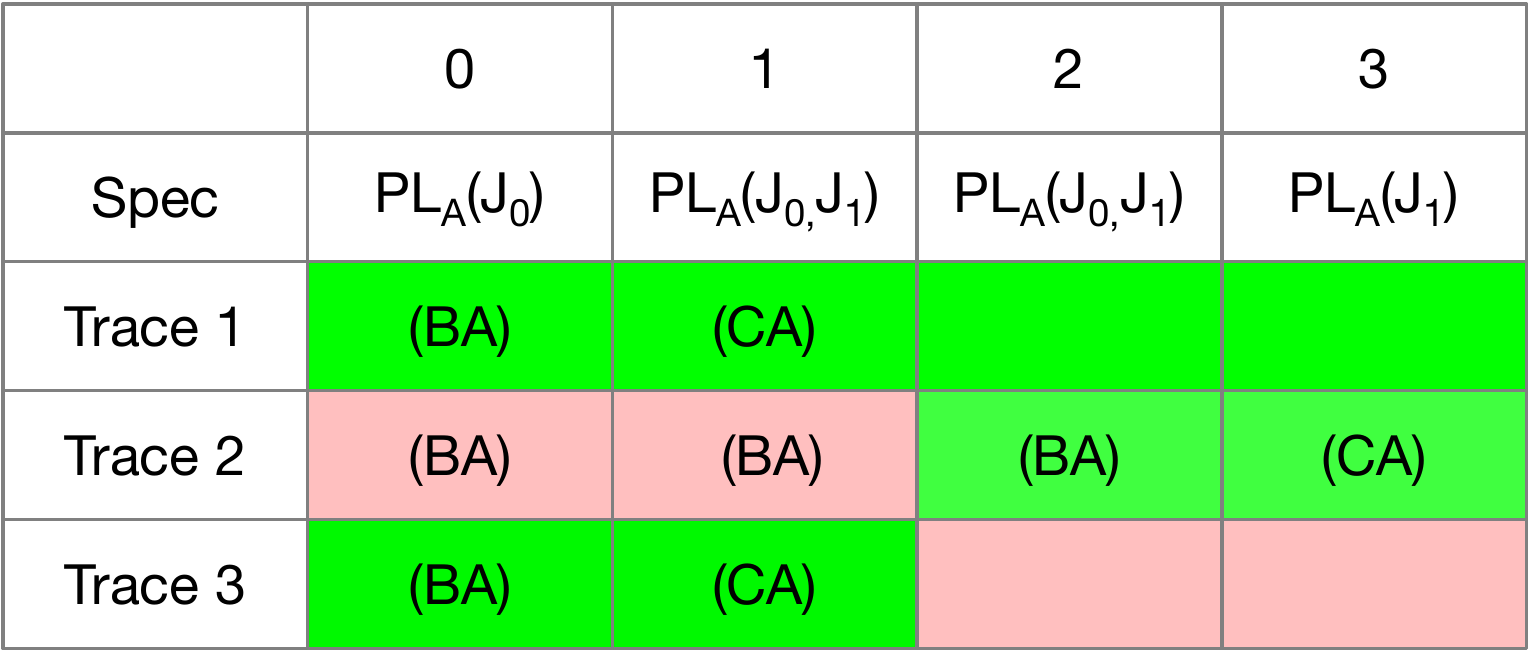}
\caption{\pn policy}
\label{fig:example:policy}
\end{subfigure}
\caption{A schedule and policy for the topology shown in Figure \ref{fig:topology-example} are constructed.
At run-time, schedules and policies behave differently depending on observed successes (green background) or failures (red background).
The traces show how schedules and policies adapt run-time behavior in response to successes and failures.
Notably, the schedule drops packets in traces 2 and 3 (indicated by white ``x''-es) while the policy drops no packets. 
%
%
%
}
\label{fig:example}
\end{figure*}

A \pn policy is represented as a scheduling matrix with a \sleepop or a \pullop action in each entry. 
A \sleepop action indicates that no action is taken in a slot and channel.
A \pullop has two arguments: a \emph{coordinator} and a \emph{service list}.
The coordinator is the node that executes the action at run-time, and the service list includes the instances that \emph{may} be executed in that slot and channel.
The instances in the service list are ordered according to the priority of their flows.
At run-time, only one of the candidate flows in the service list will be executed.
Any node can become a coordinator, and the coordinators can change from slot to slot. 
The execution of a policy is cyclic, with nodes returning to the policy's beginning upon reaching its end.


A coordinator executes a \pullop at run-time by considering the instances in the service list in priority order.
For each instance, the coordinator checks whether it has received the instance's packet.
If the coordinator has already received the instance, it will consider the next instance in the service list.
Otherwise, it will request the instance's packet
from the coordinator's neighbor through which the instance is routed.
Upon receiving a request, the neighbor may or may not have the packet (the latter case can happen when the packet was dropped at a previous hop).
If the neighbor has the packet, it includes it in its response to the coordinator. 
Otherwise, the neighbor marks the packet as dropped in its response.
In response to receiving either response, the coordinator marks the instance as successfully executed.
The invariant maintained by the execution of a \pullop is, \emph{at most, one instance from the service list is executed in a slot}.
Note that the request or the response may not be delivered since links are unreliable.
We account for this by having an instance be included in the service list of several \pullops performed by the coordinator.
As discussed in Section \ref{sec:single-hop}, an instance is included in the service list of sufficient \pullops to meet the flow's target end-to-end reliability, given the TLR's minimum reliability threshold.

The proposed adaptation mechanism is sufficiently lightweight to run within 10 ms slots, as specified by WirelessHART.
The memory usage is proportional to the number of flows routed through a node, which is small.
Equally important, the adaptation mechanism does not employ carrier sensing and, instead, relies on receiver-initiated pulls.  

To illustrate the differences between \pn policies and schedules, consider a star topology (see Figure \ref{fig:topology-example}).
In this example, two flows --  $F_0$ and $F_1$ -- relay data from $B$ and $C$ to the sink $A$.
In slot $0$, instances $J_0$ and $J_1$ are released from flows $F_0$ and $F_1$, respectively.
WirelessHART requires the construction of a schedule with two transmissions for each instance (see Figure \ref{fig:example:schedule}).
Three traces that differ in the pattern of packet losses observed at run-time are also included in the figure.
The only run-time adaptation mechanism
available in schedules is to cancel scheduled transmissions whose data has already been delivered.
The notation \tx{B}{$J_0$} indicates that $B$ transmits $J_0$'s packet to $A$.
The synthesized \pn policy is shown in Figure \ref{fig:example:policy} and uses the notation \pull{A}{$J_0, J_1$} to indicate a \pullop with $A$ as the coordinator and $\{J_0, J_1\}$ as the service list.


To highlight several differences between policies and schedules, consider trace 2, where there are failures in slots 0 and 1.
For this trace, the schedule included in Figure \ref{fig:example:schedule} cannot successfully deliver $J_0$'s packet because it is allocated only a single retransmission.
In contrast, the \pn policy included in Figure \ref{fig:example:policy} can successfully deliver $J_0$'s packet.
The policy includes $J_0$ in the service list of the \pullops in slots 0, 1, and 2.
At run-time, $J_0$'s transmission in slots 0 and 1 fails, but $J_0$ will be delivered in slot 2.
In slot 3, the policy successfully executes $J_1$.
A similar scenario is included in trace 3, where $J_1$'s packets cannot be delivered by schedules but are successfully delivered using a \pn policy.
Traces 2 and 3 highlight the flexibility of \pn policies to improve reliability by dynamically reallocating retransmissions based on the successes and failures observed at run-time.

A key property of the run-time adaptation mechanism that we will leverage during policy synthesis is the following:

\begin{theorem}
The execution of \pn actions on a node is not affected by the actions of other nodes.
\label{th:node-independence}
\end{theorem}

\begin{proof}
Consider the execution of a \pullop by a node $R$.
A \pullop's behavior depends on what instances are included in the service list and the local state of the node.
Since the service list is fixed once the policy is constructed, the only way another node may affect $R$'s state is by directly modifying its state, which does not happen.
\end{proof}

\subsection{Synthesizing \pn Policies for Data Collection on Star Topologies}
\label{sec:single-hop}

As a starting point, let us consider the problem of constructing \pn policies for a star topology where all flows have the base station as the destination (see Figure \ref{fig:topology-example} for an example).
This setup simplifies the synthesis of policies in two regards:
(1) The base station will be the coordinator of all \pullops. Therefore, we only have to focus on determining the \servicelist of each \pullop.
(2) Since all flows have the base station as the destination, there will be no transmission conflicts, and a (different) single channel can be used in each slot.
We will generalize our approach to general multi-hop topologies and workloads in the next section.

The policy synthesis procedure involves two key components -- an \evaluator and a \builder (see Figure \ref{fig:system-design}).
The policy is synthesized incrementally by alternating the execution of the \builder and \evaluator in each slot.
\begin{itemize}
    \item  The \builder determines the \pullops that will be executed in each slot. The \builder maintains an \activelist that contains all of the instances that have been released but have \emph{not} yet met their end-to-end reliability. In a slot $t$, the builder checks whether an instance $J_{i,k}$ is released (i.e., when $r_{i,k} = t$) and, if this is the case, $J_{i,k}$ is added to the \activelist . If the \activelist is not empty in $t$, a \pullop having the base station as coordinator and the instances in the \activelist as its \servicelist is assigned in the entry $t$ of the matrix. 
    \item The \evaluator maintains the likelihood that each instance in the \activelist has been delivered to the base station.
    The probabilities are updated incrementally to reflect the execution of the \pullop provided by the \builder in slot $t$.
    \item At the end of slot $t$, the \builder removes all instances whose reliability exceeds their end-to-end reliability targets from the \activelist.
\end{itemize}


In the remainder of the section, we will answer the question of how to estimate the reliability of flows given the sequence of \pullops determined by the \builder.
This problem can be modeled at a high level as a Markov Decision Process (MDP) whose transitions depend on the likelihood of successfully executing \pullops. 
Let \allstates be the set of all possible states.
A state $s$ ($s \in \allstates$) is represented as a vector of size $|\flows|$, where the $i^{th}$ entry represents the state of instance $J_{i}$.
The state of an instance $J_i$ may be \success or \failure, indicating whether the base station requested $J_i$'s data and received a reply successfully.
The reply may either include a flow's packet or an indication that it has been dropped on a previous hop.
A direct encoding of this information would require $O(2^{|\flows|})$ states, which is not practical when there are numerous flows.
To avoid state explosion, we propose the following mechanism.
We bound the length of the \activelist maintained by a coordinator.
This requires a simple modification to the \builder:
an instance is added to the \activelist until it reaches the user-specified maximum size.
The additional instances that are released when the active list is full are added to an \texttt{inactive} list.
The \texttt{inactive} list includes instances that are released but not yet active.
When an instance completes, the size of the \activelist decreases by one, and the highest priority instance from the \texttt{inactive} list is moved to the \activelist.

With this modification, the maximum number of states a coordinator maintains is reduced to $O(2^{|\activelist|})$.
Additionally, we observe that the likelihood an instance is executed depends on its index in the \servicelist.
If the index of an instance in the \servicelist exceeds 3 or 4, then the instance is unlikely to be executed.
Accordingly, we also cap the maximum size of the \servicelist.
The \servicelist of a \pullop is then a subset of the \activelist.
In our experiments, we constrain $|\activelist|\le10$ and $|\servicelist|\le 4$ except where otherwise stated.

\textbf{End-to-end Reliability Using Instantaneous Link Quality:} 
Let us start by deriving a method for computing the end-to-end reliability of flows under the assumption that there is an omniscient oracle that can provide the instantaneous probability of a successful \pullop for all links in a slot $t$.
We will use the notation \snapshot{t} to represent the link quality of all links in slot $t$.
Later, we will relax this requirement by constraining links to follow the TLR model i.e., their link quality is lower bounded by $\pmin$ (i.e., $LQ_i(t) \ge \pmin$).
Under this assumption, we will show that the worst-case end-to-end reliability of a flow occurs when the quality of all links is equal to $\pmin$ in all slots.

\begin{algorithm}[t]
\footnotesize	
\LinesNumbered
\DontPrintSemicolon

\SetKwProg{buildmatrixproc}{Procedure}{}{}
\SetKwFunction{buildmatrix}{BuildTransitionMatrix}
\SetKwFunction{onSuccess}{onSuccess}

\buildmatrixproc{\buildmatrix{$srv(t)$, \snapshot{t}}}{
    \tmat{srv(t)} = $\boldsymbol{I}$ \\
    \For{current in \allstates} {
        \For{$J_i$ in $srv$} {
            Let $i$ be the flow id of $J_i$ \\
            \If{current[$i$] = \failure}{
                \tcc{the execution fails} 
                \tmat{srv}[current, current] = 1 - \linkqt{i} \label{algo:matrix-fail} \\
                \BlankLine
                \tcc{the execution is successful} 
                next = \onSuccess{current, $i$} \\
                \tmat{srv(t)}[current, next] = \linkqt{i} \\
                \textbf{break}
            }
        } 
    }
    \Return{\tmat{srv(t)}}

}

\buildmatrixproc{\onSuccess{state, $i$}} {
    \tcc{the next\_state is the same as current except for the entry for $J_i$ becomes \success \label{algo:matrix-succcess}}
    next\_state[$j$] = state[$j$] $~~~~~~~\forall j \neq i$ \\
    next\_state[$i$] = \success \\
    \Return{next\_state}
}

\caption{Computes the transition matrix \tmat{srv(t)} given the \servicelist $srv$ of a \pullop and a snapshot of current link \snapshot{t}}
\label{algo:matrix}
\end{algorithm}

The actions of the MDP are the \pullops that the \builder assigns in each slot.
Initially, the system is in a state $s_0$, in which the base station has not received the data from any of the flows.
Consider the execution of a \pullop with \servicelist $srv$ in slot $t$.
To account for the impact of executing the \pullop on the state of the system, we construct a transition matrix $\tmat{srv(t)}$ of size $2^{|\activelist|} \times 2^{|\activelist|}$ using Algorithm \ref{algo:matrix}.
Let $J_{i}$ be an instance included in the \servicelist $srv$ (not necessarily as the head of the list).
According to the semantics of \pullops, $J_{i}$ will be executed in any $current$ state where the $i^{th}$ entry of the vector is a failure (i.e., $current[i] = \failure$) and the execution of all instances $J_j$ with higher priority than $J_i$ in the \servicelist $srv$ have already succeeded (i.e., $current[j] = \success$).
From such a $current$ state, there are two possible outgoing transitions depending on whether the \pullop is successful or not.
If the execution of $J_{i}$ fails, then the system remains in the same state (see line \ref{algo:matrix-fail}, Algorithm \ref{algo:matrix}).
Accordingly, the entry $\tmat{srv(t)}[current, current]$ is set to 1 - $\linkqt{i}$, where $\linkqt{i}$ is the probability of performing a successful \pullop over the link used by flow $i$ in slot $t$.
Conversely, if the execution of $J_{i}$ succeeds, the system transitions from the $current$ state to a $next$ state.
The entries of the $current$ and the $next$ states are the same, except for the entry associated with the $J_i$ element for which $next[i] = \success$ (see line \ref{algo:matrix-succcess}, Algorithm \ref{algo:matrix}).
In this case, we set $\tmat{srv(t)}[current, next] = \linkqt{i}$.
If a \sleepop is assigned slot $t$, then the state of the system does not change.

After executing $t$ \pullops, the probability of each state is given by the vector $\pstate{t}$:
\begin{equation}
    \pstate{t} = \boldsymbol{s}_0^{T} \tmat{srv(0)} \tmat{srv(1)} \cdots \tmat{srv(t)} 
    \label{eq:state-transition}
\end{equation}
\noindent where $\boldsymbol{s}_0$ is the initial state of the system and $\tmat{srv(t')}$ is the transition matrix associated with the \pullop that has $srv(t')$ as its service list and is executed in slot $t'$ ($0 \le t' \le t$).
Equation \ref{eq:state-transition} describes the evolution of the system as a discrete-time Markov Chain (MC) that is parametric and time inhomogeneous. 
The structure of $\tmat{srv(t')}$ depends on the service list and its values depends on the quality of all links in slot $t'$.

The end-to-end reliability $\reliability{i}$ of  instance $J_i$ after executing $t$ \pullops is computed by summing up the probability of each state $s$ ($s \in \allstates$) such that $s[i]$ is \success. 
Leveraging the properties of matrix multiplication, $R_i$ may be written as:
\begin{equation}
    \reliability{i,t}= \pstate{t} \boldsymbol{\chi}_{i}
\end{equation}
\noindent where, $\boldsymbol{\chi}_{i}$ is a vector such that $\boldsymbol{\chi}_{i}[k] = 1$ for any state $s$ such that $s[k] = \success$ and $\boldsymbol{\chi}_{i}[k] = 0$ otherwise.

\textbf{End-to-end reliability under TLR:}
Computing $\reliability{i,t}$ requires that we know the instantaneous quality of all links in any slot $t$. 
It is infeasible to have access to this information during the synthesis of a policy.
In the following, we will derive a lower bound $\reliabilitylb{i,t}$ on $\reliability{i,t}$.
To this end, we will construct a new MC with transition matrix $\tmatlb{srv(t)}$ that is computed by considering each transition matrix $\tmat{srv(t)}$ and replacing each link quality variable $\linkqt{i}$ with its lower-bound \pmin.
We claim that a lower bound on the end-to-end reliability of a flow \reliability{i,t} is:
\begin{equation}
    \reliability{i,t} \geq
    \reliabilitylb{i, t} = 
                \pstatelb{t} \boldsymbol{\chi}_{i} = \boldsymbol{s}_0^{T} \tmatlb{srv(0)} \tmatlb{srv(1)} \cdots \tmatlb{srv(t)} \boldsymbol{\chi}_{i}
\end{equation}
\noindent 
The following theorem implies that to compute a lower-bound on the reliability of a flow, 
it is sufficient to consider only the case when all links perform their worst.

\begin{theorem}
Consider a star topology that has node $A$ as a base station and a set of flows $\flows = \{F_0, F_1, \dots F_N\}$ that have $A$ as destination.
Let \linkqt{0}, \linkqt{1}, \dots \linkqt{N} be the quality of the links used by each flow in slot $t$ such that $\pmin \le LQ_i(t) \le 1$ for all flows $F_i$ ($F_i \in \flows$) and all slots $t$ ($t \in \mathbb{N}$).
Under these assumptions, the reliability $\reliability{i,t}$ of an instance $J_i$ after executing $t$ \pullops of the \pn policy $\pi$ is lower bounded by $\reliabilitylb{i,t}.$
\label{th:monotonic}
\end{theorem}

\begin{proof}
See Section \ref{sec:proof}.
\end{proof}


\begin{figure}[ht]
    \centering
    \includegraphics[width=\textwidth]{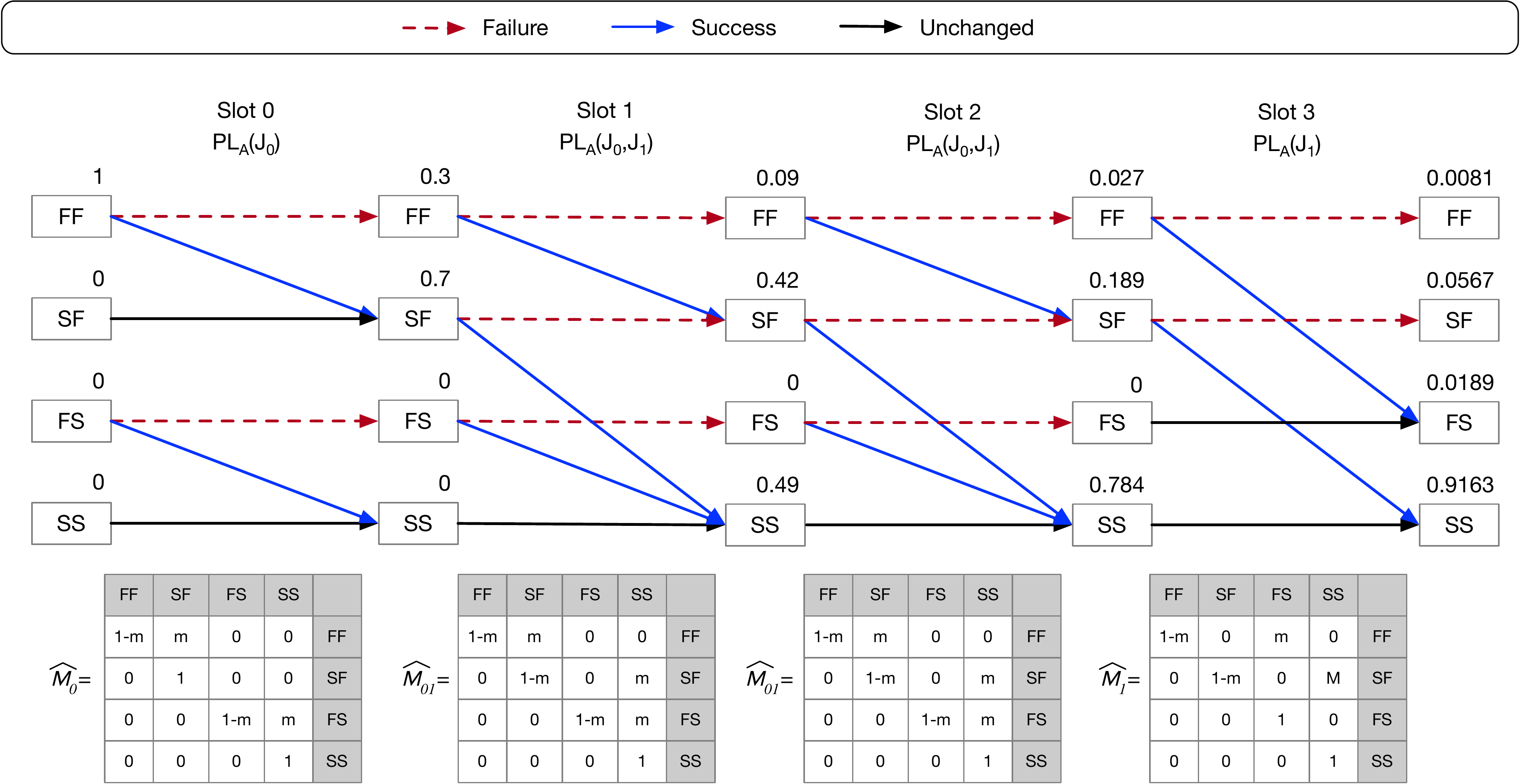}
    \caption{Estimating the state of the network and lower-bounds on the end-to-end reliability.}
    \label{fig:policy-mc}
\end{figure}

Let us return to our running example of the construction and execution of the policy shown in Figure \ref{fig:example:policy}.
In Figure \ref{fig:policy-mc}, we will illustrate how the end-to-end reliability of flows will be estimated for this example.
The workload includes two flows -- $F_0$ and $F_1$ -- with phases $\phi_{0} = 0$ and $\phi_{1} = 1$.
Accordingly, instances $J_0$ and $J_1$ are released in slots $0$ and $1$.
We will evaluate the estimated state of the network $\pstatelb{t}$ and the lower-bounds on the reliability of each flow as the policy is executed. 
Given that the workload involves only two flows,  the possible states of the systems are $\allstates = \{ \texttt{FF}, \texttt{SF}, \texttt{FS}, \texttt{SS}\}$.
Each state encodes whether the base station $A$ has received the data of $J_0$ and $J_1$.
In any slot $t$, the probability vector $\pstatelb{t}$ is the likelihood that the network is in a state \texttt{FF}, \texttt{SF}, \texttt{FS}, and \texttt{SS} (in that order).
The lower bound on the reliability of instance $J_0$ is $\reliabilitylb{0,t} = \pstatelb{t}[\texttt{SF}] + \pstatelb{t}[\texttt{SS}] = \pstatelb{t}\chi_0 $, where $\boldsymbol{\chi}_0 = [0,1,0,1]$.
Similarly, $\reliabilitylb{1,t} = \pstatelb{t}[\texttt{FS}] + \pstatelb{t}[\texttt{SS}] = \pstatelb{t}\boldsymbol{\chi}_1 $, where $\boldsymbol{\chi}_1 = [0,0,1,1]$.

Initially, the system is in state $\boldsymbol{s}_0 = [1,0,0,0]^{T}$ i.e., $\boldsymbol{s}_{0}[\texttt{FF}] = 1$ and the likelihood of the remaining states is zero. 
The action \pull{A}{$J_0$} is executed in slot $0$.
The \evaluator constructs the matrix $\tmatlb{0}$ to account for the impact of executing the \pullop on the state of the system.
After executing the \pullop, the state of the network is 
$\pstatelb{0} = \boldsymbol{s}_0^T\tmatlb{0}$.
The reliability of $J_0$ after executing \pull{A}{$J_0$} is $\reliabilitylb{0,0} = \pstatelb{0} \boldsymbol{\chi}_0 = \pstatelb{0}[\texttt{SF}] +  \pstatelb{0}[\texttt{SS}] = 0.7$.
Figure \ref{fig:policy-mc} shows the states of the MC after the execution of each \pullop.
The transition matrices associated with each \pullop are included at the bottom of the figure.
The reliability of flows is evaluated in a similar manner in the remaining slots.

\subsection{Synthesizing \pn Policies for General Topologies}
\label{sec:multi-hop}

In this section, we extend the results from the previous subsection to general workloads and topologies.
Doing so requires that we determine both a coordinator and a service list for each \pullop.
The \builder must assign coordinators and service lists such that no transmission and channel conflicts occur.
The \evaluator must provide lower bounds on the reliability of the flows as they interact across multiple hops.
A naive evaluation that simply keeps track of when a coordinator received packets from all combinations of flows does not scale.
We will start by discussing how a scalable \evaluator may be built and then extend the \builder.

\subsubsection{The Multi-hop Evaluator}

\begin{wrapfigure}{r}{0.3\textwidth}
    \centering
    \includegraphics[width=\linewidth]{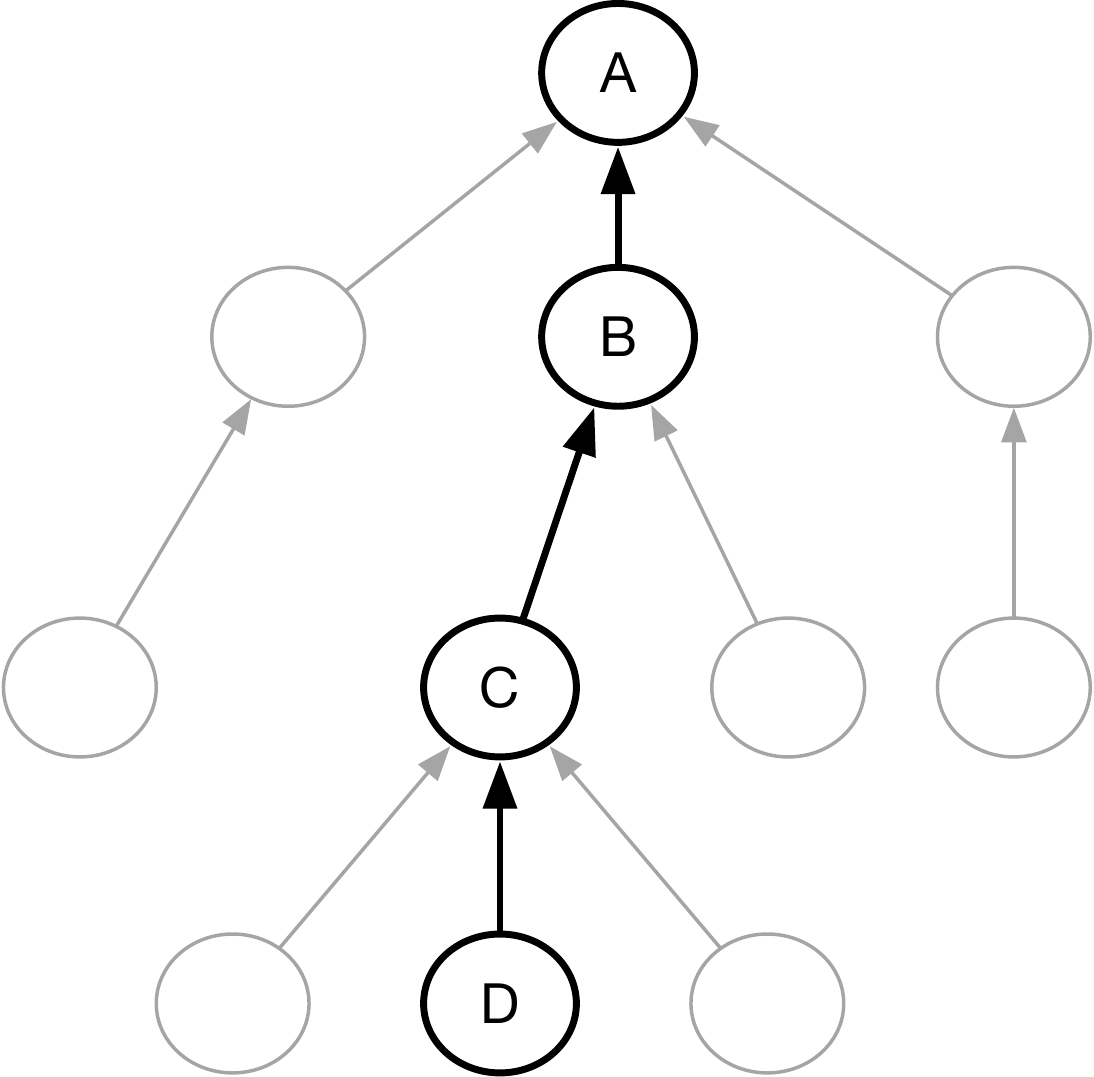}
    \caption{Multi-hop example.}
    \label{fig:multihop-flow}
\end{wrapfigure}

The key insight to building a scalable \evaluator is to require coordinator nodes to \emph{operate independently}.
Consider a multi-hop flow $F_2$ shown in Figure \ref{fig:multihop-flow} whose data is forwarded using the path $\flowpath{2}=\{D, C, B, A\}$.
To forward $F_2$'s data, a policy must include a sequence of \pullops that have the nodes $C$, $B$, and $A$ as coordinators and include $F_2$ as part of their service lists.
A simple approach to ensure that coordinators operate independently is to use an approach similar to the Phase Modification Protocol \cite{bettati1994end}, where a multi-hop flow is divided into single-hop subflows flows and allocate $\delta_2=D_2 / |\flowpath{2}|$ slots for the execution of each flow.
The first subflow $F_{2,1}$ from $D$ to $C$ is released $\phi_{2,1} = \phi_2$ and must complete with $\delta_2$ slots.
The second subflow $F_{2,2}$ from $C$ to $B$ is released at $\phi_{2,2} = \phi_2 + \delta_2$ and it must complete within $\delta_2$ slots.
The remainder of the subflows are set up in a similar fashion.
To ensure that coordinators operate independently, it is essential that each subflow releases a packet regardless of whether the previous subflow delivered it successfully or unsuccessfully to the next hop.
By taking advantage of the independence, we can use the single-hop \evaluator described in the previous section to evaluate the reliability of each subflow.
Then, the end-to-end reliability of the original flow is simply the product of the reliability of each subflow (due to independence).


The drawback of this approach is that each subflow is allocated an equal number of slots which can be problematic when the workload of nodes is not uniform.
To address this issue, we first convert the end-to-end target reliability of $T_i$ into a local reliability target that each subflow must meet: 
\begin{equation}
    T_i^{\frac{1}{|\flowpath{i}|}}
    \label{eq:local-reliability}
\end{equation}
\noindent where $|\flowpath{i}|$ is the length of $F_i$'s path measured in hops.
Each subflow is then allowed to release the earliest slot in which all subflows associated with the previous hops of the original flow have met the local reliability target.
Notably, different subflows may need to be executed a different number of times to meet their local target reliability to handle non-uniform workloads effectively. 



\subsubsection{The Multi-hop Builder}

The optimization problem can be formulated as an Integer Linear Program (ILP).
The ILP includes three types of variables.
For each node $R$ ($R \in \nodes$), the variable $N_R$ ($N_R \in \{0, 1\}$) indicates whether $R$ is the coordinator of a \pullop.
For each released instance $J_i$, the variable $I_i$ ($I_i \in \{0, 1\}$) indicates whether its associated active link will be added to a service list.
Finally, variable $C_{R,ch}$ ($C_{R,ch} \in \{0, 1\}$) indicates whether $R$ will use channel $ch$ to communicate.
The ILP solution is converted into a set of \pullops as follows:
for each node $R$ such that $N_R=1$, we add a \pullop that has $R$ as the coordinator and a service list with all instances $J_i$
where $I_i = 1$ and $R$ is the receiver of the active link of $J_i$.
The \pullop is assigned to the entry in the matrix for the current slot and the channel $ch$ for which $C_{R,ch}=1$.
We let $\mathcal{A}$ be the union of the \activelist of all nodes.

A well-formed policy must ensure that no transmission conflicts will be introduced at run-time.
Consider a \pullop that has $R$ as a coordinator and services instance $J_i$.
Let $(SR)$ be the active link of $J_i$, where $S=src(J_i)$ and $R=dst(J_i)$. 
If $J_i$ will be assigned in the current slot (i.e., $I_i = 1$), then $S$ cannot be a coordinator for any other instance since this would require $S$ to transmit and receive in the same slot.
We enforce this using the following constraint:

\begin{IEEEeqnarray}{c}
N_S \le (1 - I_i) ~~~~\forall I_i \in \mathcal{A}: ~~S=src(J_i) 
\label{eq:constraint-sender}
\end{IEEEeqnarray}

\noindent A similar constraint must also be included for the receiver $R$. 
If node $R$ is not a coordinator (i.e., $N_R = 0$), then $J_i$ cannot be assigned and $I_i = 0$.
Conversely, if $R$ is selected as a coordinator, instance $J_i$ may (or may not) be assigned (i.e., $I_i \le N_R = 1)$ depending on the objective of the optimization, which we will discuss later in this section.
These aspects are captured by the following constraint:

\begin{IEEEeqnarray}{c}
I_i \le N_R   ~~~~\forall I_i \in \mathcal{A}: ~~R=dst(J_i) 
\label{eq:constraint-receiver}
\end{IEEEeqnarray}

\begin{figure}[t!]
    \centering
    \begin{subfigure}[b]{0.18\textwidth}
        \centering
        \includegraphics[width=\textwidth]{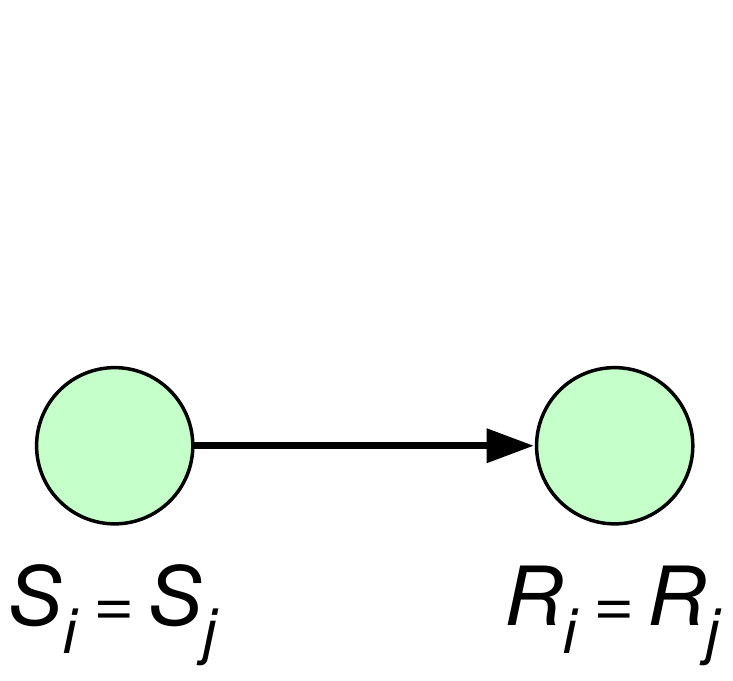}
        \caption{\footnotesize{Same link}}
        \label{fig:ilp-total-common}
    \end{subfigure}
    \begin{subfigure}[b]{0.18\textwidth}  
        \centering 
        \includegraphics[width=\textwidth]{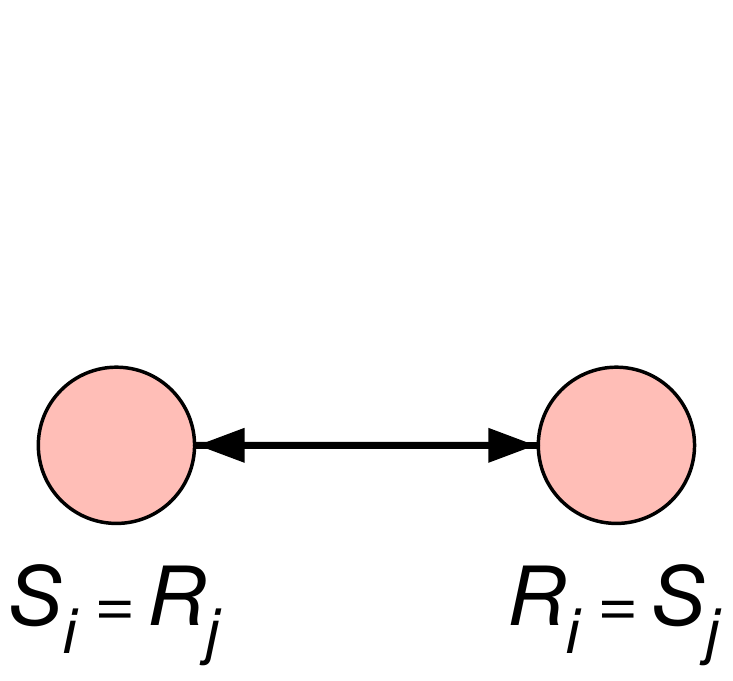}
        \caption{\footnotesize{Opposite link}}
        \label{fig:ilp-reverse}
    \end{subfigure}
    \begin{subfigure}[b]{0.18\textwidth}   
        \centering 
        \includegraphics[width=\textwidth]{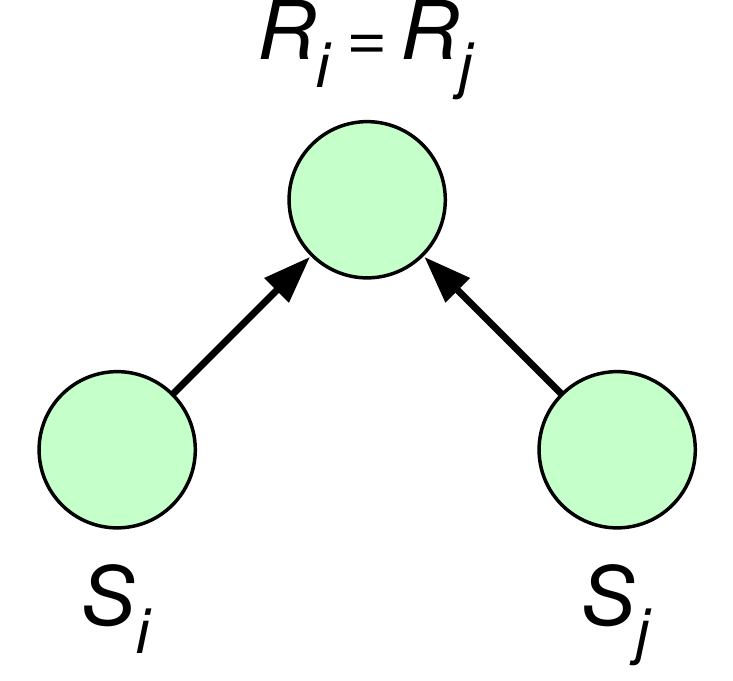}
        \caption{\footnotesize{Common Recv}}
        \label{fig:ilp-common-rec}
    \end{subfigure}
    \begin{subfigure}[b]{0.18\textwidth}   
        \centering 
        \includegraphics[width=\textwidth]{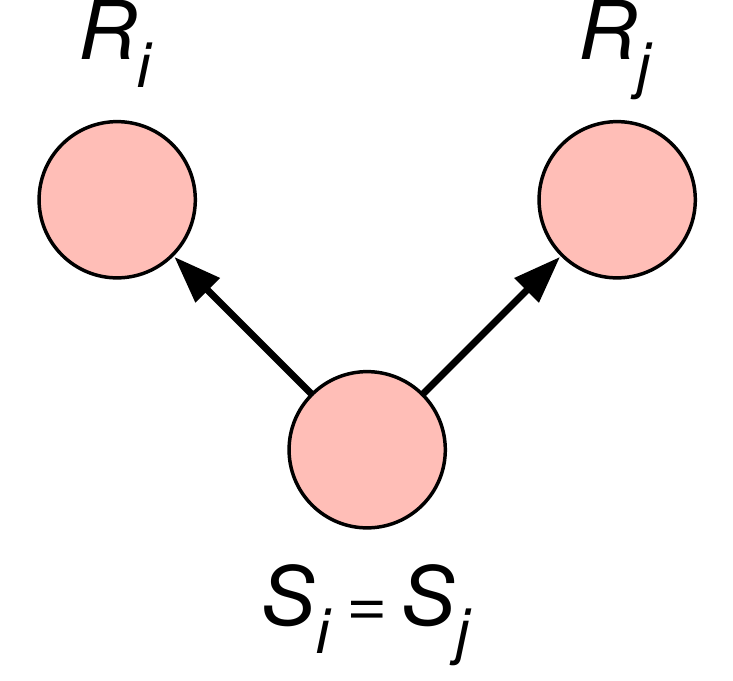}
        \caption{\footnotesize{Common Sender}}
        \label{fig:ilp-common-send}
    \end{subfigure}
   \vskip\baselineskip
    \begin{subfigure}[b]{0.18\textwidth}   
        \centering 
        \includegraphics[width=\textwidth]{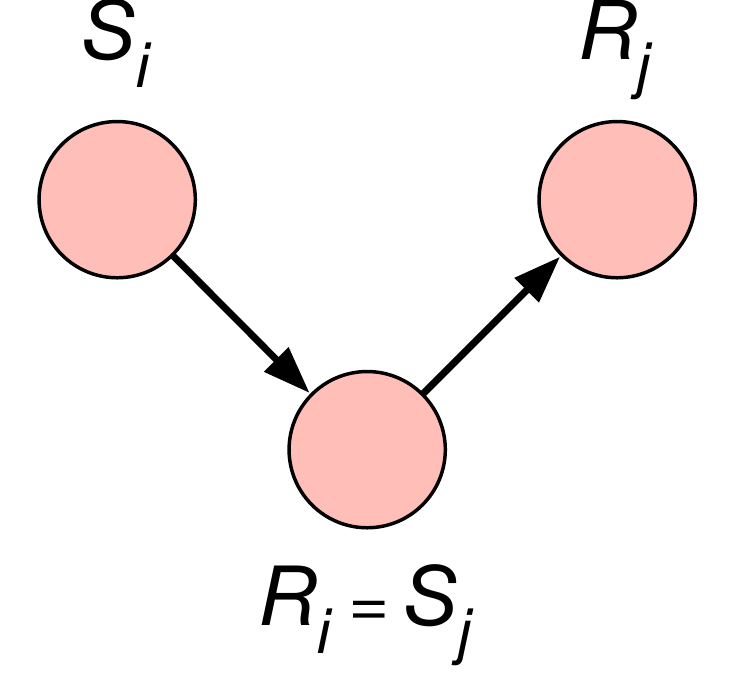}
        \caption{\footnotesize{Recv/sender}}
        \label{fig:ilp-rec-send}
    \end{subfigure}
    \begin{subfigure}[b]{0.18\textwidth}   
        \centering 
        \includegraphics[width=\textwidth]{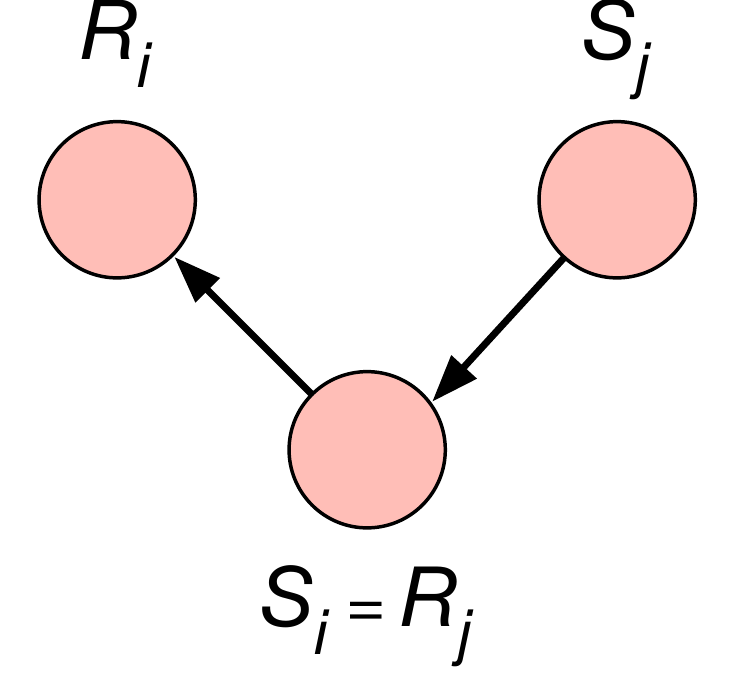}
        \caption{\footnotesize{Sender/Recv}}
        \label{fig:ilp-send-rec}
    \end{subfigure}
    \caption{Possible ways two instances may share at least one node.  Green cases have no transmission conflict while red cases do.} 
    \label{fig:ilp-cases}
\end{figure}

The above constraints avoid all transmission conflicts with one exception.
Consider the case when two instances $J_i$ and $J_j$ share the same sender but have different receivers.
An assignment that respects constraints \ref{eq:constraint-sender} and \ref{eq:constraint-receiver} is for both instances to be assigned in the current slot (i.e., $I_i = I_j = 1$). 
However, this would result in a conflict, since the common sender can only transmit one packet in a slot.
To avoid this situation, we introduce the following constraint:

\begin{IEEEeqnarray}{rCl}
 & I_i + I_j - 1 \le N_{S} & \label{eq:constraint-pairs}\\
 & \forall I_i, I_j \in \mathcal{A}: ~~S = src(J_i) = src(J_j) ~~\& ~~dst(J_i) \ne dst(J_j) & \nonumber
\end{IEEEeqnarray}

\begin{theorem}
    Constraints \ref{eq:constraint-sender}, \ref{eq:constraint-receiver}, and \ref{eq:constraint-pairs} ensure that the execution of \pullops will result in no node receiving or transmitting more than once in a time slot.
    \label{th:partially-wellformed}
\end{theorem}

\begin{proof}
    To prove Theorem \ref{th:partially-wellformed} holds it is sufficient to consider whether two arbitrary flow instances may conflict.
    Accordingly, there are six cases to be considered as depicted in Figure \ref{fig:ilp-cases} where two instances $J_i$ and $J_j$ share at least a node.

    \emph{\textbf{Case 1 --} Same link (see Figure \ref{fig:ilp-total-common}:)}
    If $I_i = I_j = 1$, then $N_{R_i} = N_{R_j} = 1$ due to constraint \ref{eq:constraint-receiver}.
    In this case, both $J_i$ and $J_j$ will be serviced as part of the same \pn operation that is coordinated by node $R = R_i = R_j$. 
    At run-time, the coordinator $R$ will pull either $J_i$ or $J_j$ (but not both) depending on its local state. 
    Note that this is one the cases \pn exploits to adapt and improve performance.
    
    \emph{\textbf{Case 2 --} Opposite link (see Figure \ref{fig:ilp-reverse}:)}
    Executing $J_i$ and $J_j$ in the same slot would result in a conflict since one of the common nodes would have to be both a sender and a receiver.
    We will prove by contradiction that $J_i$ and $J_j$ will not be assigned in the same slot.
    Assume that $I_i = I_j = 1$ and, without loss of generality, that the common node is $N = S_i = R_j$.
    Since 
    Since $I_j = 1$, then $N_{R_j} = 1$ due to constraint \ref{eq:constraint-receiver}. 
    Also, since $I_i = 1$, then $N_{S_i} = 0$ due to constraint \ref{eq:constraint-sender}.
    This is a contradiction since $N_{R_j} = N_{S_i}$ and $R_j$ and $S_i$ refer to the same node.
    The proofs for the cases given in Figures \ref{fig:ilp-rec-send} and \ref{fig:ilp-send-rec} are similar.
    
    \emph{\textbf{Case 3 --} Common receiver (see Figure \ref{fig:ilp-common-rec}):}
    The common receiver case is similar to the same link case with the exception that the senders for both $J_i$ and $J_j$ are different.  
    Note that this is one the cases \pn exploits to adapt and improve performance.

    \emph{\textbf{Case 4 --} Common sender (see Figure \ref{fig:ilp-common-send}):}
    Executing $J_i$ and $J_j$ in the same slot would result in a conflict since $S = S_i = S_j$ would have to transmit two packets in the same slot. We will prove by contradiction that this cannot happen. Assume that $I_i = I_j = 1$. Since $I_i = 1$, then $N_{S_i} = 0$ according to constraint \ref{eq:constraint-sender}. However, $I_i + I_j - 1 = 1 \le N_{S_i}$ due to constraint \ref{eq:constraint-pairs},  which is a contraction.
\end{proof}

The next set of constraints ensures that each \pullop is assigned a unique channel.  
We accomplish this by introducing $C_{R,ch}$ to indicate whether coordinator $R$ uses channel $ch$ ($ch = 1 \dots K$), where $K$ is the number of channels.
The selection of channels is subject to the constraints:

\begin{IEEEeqnarray}{c}
 \sum_{R \in \nodes}{C_{R,ch} } \le 1  ~~~ \forall ch \in 1 \dots K  \\
 \sum_{ch=1}^{K}{C_{R,ch}} = N_{R}
\end{IEEEeqnarray}

A requirement of the TLR model described in Section \ref{sec:reliability-model} is that coordinators must switch channels between \pullops to ensure independence between transmissions. We enforce this property by introducing additional constraints to prevent coordinators from using the same channel.

To enforce the prioritization of instances, we set the optimization objective to be:
\begin{IEEEeqnarray}{c}
\sum_{i = 0}^{i < |\mathcal{A}|}{2^{|\mathcal{A}| - i}  I_i}
\end{IEEEeqnarray}

\noindent The objective function ensures that a flow $F_i$ will be assigned over lower priority flows unless there is a higher priority flow with a conflict with $F_i$.


\section{Experiments}
\label{sec:experiments}

Our experiments demonstrate the efficacy of \pn to support higher performance and agility than traditional scheduling approaches.
We focus on the next generation of smart factories that will use sophisticated sensors that are grid-powered and require higher data rates than current IIoT systems.
Specifically, we are interested in answering the following questions:

\begin{itemize}[topsep=-8pt]
    \item Does \pn improve the real-time capacity in typical IIoT workloads?
    \item Does \pn provide safe reliability guarantees as the quality of links varies significantly?
    \item Can \pn synthesize policies in a timely manner?
\end{itemize}

\subsection{Methodology}
We compare \pn policies against three baselines.
First, we compare two scheduling approaches that do not share entries.
To provide a fair comparison between schedules and policies, we first construct schedules (\textsf{Sched}) using the same ILP formulation as \pn policies but without allowing entries to be shared.
This is accomplished by adding to the ILP an additional constraint that the size of the service list is one .
We also compare against the conflict aware least laxity first scheduler (\textsf{CLLF}) \cite{conflictAwareLeastLaxityFirst}.
\textsf{CLLF} has been shown to produce near-optimal schedules and constitutes the current state-of-the-art scheduler.
Similar to \textsf{Sched}, \textsf{CLLF} also does not share entries.
Second, we compare against the Flow Centric Policy (\textsf{FCP})~\cite{brummet2018flexible}, which allows entry sharing only among the links of \emph{a single flow}, whereas \pn can share entries \emph{across multiple flows}.
\textsf{Sched}, \textsf{CLLF}, and \textsf{FCP} utilize sender-initiated transmissions, while \pn utilizes receiver-initiated pulls.

Unless stated otherwise, we use $\pmin = 70\%$ as suggested by Emerson's guide to deploying WirelessHART networks.
In simulations, we set the probability of a successful transmission to equal \pmin.
The number of retransmissions used by \pn, \textsf{Sched}, \textsf{CLLF}, and \textsf{FCP} is configured to achieve a 99\% end-to-end reliability for all flows.
The period and deadline are equal, and the phases are 0 in all workloads.
Flow priorities are assigned such that flows with shorter deadlines have higher priority.
To break ties, flows with longer routes are assigned a higher priority.
The remaining ties are broken arbitrarily.

We quantify the performance of protocols using \emph{max flows scheduled}, \emph{real-time capacity}, and \emph{response time}.
The max flows scheduled measures the maximum number of flows that can be supported without missing the deadlines or reliability requirements of any flows.
The real-time capacity is the highest rate at which flows can release packets without missing deadlines or reliability constraints.
The response time is the maximum latency of all instances of a flow as measured from the time when an instance is released until it is delivered to its destination.


\subsection{Simulations}
We use a discrete event simulator to control \pmin in the TLR model precisely, which is impractical on a testbed.
The simulator determines the success or failure of transmitting a packet and receiving the acknowledgment over a link by drawing from a Binomial distribution whose change of success can be configured.
Unless stated otherwise, all links are configured to have the same success chance of \pmin.
All simulations are either single-hop or performed on one of the following two topologies: a 41-node, 6-hop diameter topology with an average of 5.5 links per node derived from a testbed deployed at Washington University in St. Louis (WashU topology) \cite{washUTopology} and an 85-node, 6-hop diameter topology with an average of 10.4 links per node derived from the Indriya testbed (Indriya topology) \cite{indriyaTopology}.
In simulations, we used settings consistent with 802.15.4: the number of channels was set to sixteen and we used 10 ms slots sufficiently large to transmit a packet and receive an acknowledgment.

\subsubsection{Star Topology}
We compare \pn and \textsf{Sched} in the practically important case of star topologies.
In star topologies, for the workloads we consider, \textsf{Sched}, \textsf{CLLF}, and \textsf{FCP} perform identically and, therefore, we only report the results of \textsf{Sched}.
In this experiment, we consider workloads consisting of flows that have a period and deadline of 100 slots.
We increase the number of flows until the workload becomes unschedulable under both \pn and \textsf{Sched}.

\textbf{Performance in Star Topologies:}
Figure \ref{fig:one-hop-performance} plots the max response time of all scheduled flows as the number of flows in the workload is increased.
We configure \textsf{Sched} and \pn to have an end-to-end reliability of 99\% for each flow when $\pmin= 60\%$ and $\pmin= 70\%$.
The figure indicates the max response time increased until each protocol reached its real-time capacity, as indicated by the vertical line in the figure.
When $\pmin = 70\%$, \pn supports 63 flows without missing deadlines compared to only 25 flows supported by \textsf{Sched}.
This represents a real-time capacity improvement of 2.52 times at $\pmin = 70\%$ and 3.25 times at $\pmin = 60\%$. 


\begin{figure}
 \centering
\begin{subfigure}[b]{0.45\linewidth}
\vskip 0pt
\centering
\includegraphics[width=\linewidth]{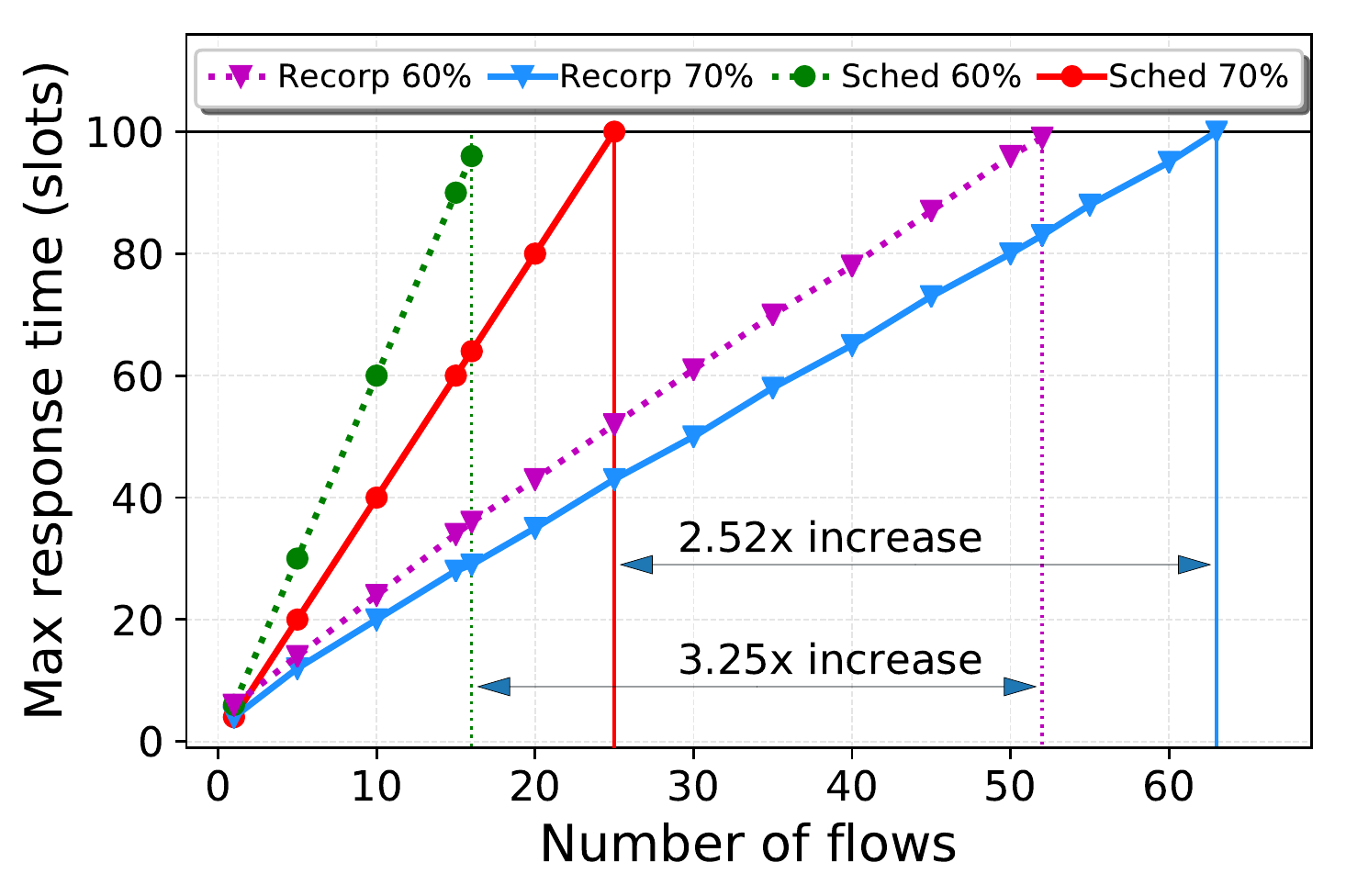}
\caption{Max response time and real-time capacity}
\label{fig:one-hop-performance}
\end{subfigure}
%
%
\begin{subfigure}[b]{0.45\linewidth}
\vskip 0pt
\centering
\includegraphics[width=\linewidth]{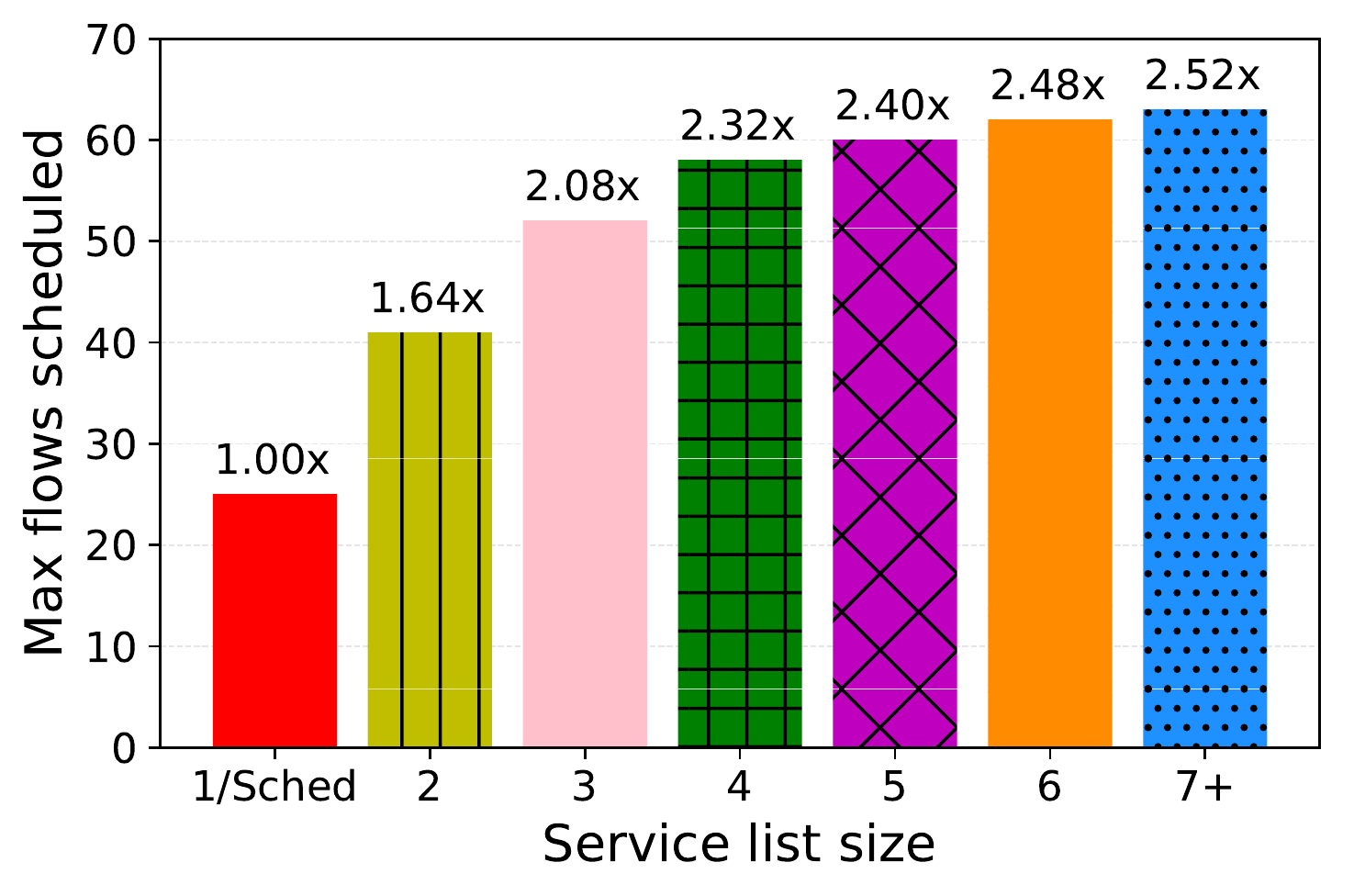}
\caption{Size Impact of service list on schedulability}
\label{fig:servicelist}
\end{subfigure}
\caption{Simulations on star topologies.}
\end{figure}

\textbf{Impact of the Service List Size:}
Schedules and \pn policies differ in how many instances can share an entry, which can be controlled by constraining the size of the service list.
Schedules provide no sharing and are limited to a service list size of one.
In contrast, \pn policies allow multiple flows to be included in the service list to share an entry.
Figure \ref{fig:servicelist} plots the maximum number of flows scheduled as the service list size is varied when $\pmin=70\%$.
When the size of the service list is one, \pn behaves like \textsf{Sched}.
As we allow more flows to potentially share a slot, the number of flows scheduled increases.
However, there are diminishing returns; most of the benefit is observed when the service list is capped at 4 to 6 flows.
No meaningful improvement in the real-time capacity may be observed after increasing the service list size beyond 7 flows.
Based on this result, we set the maximum service list size to 4 for all remaining experiments.
These results indicate that \emph{it is sufficient to share slots across only a few flows to gain most of the benefits of using \pn policies}.

\subsubsection{Multihop Topology}
To provide a comprehensive comparison between \pn, \textsf{Sched}, \textsf{CLLF}, and \textsf{FCP}, we consider four typical workloads: data collection, data dissemination, a mix of data collection and dissemination, and route through the base station. 
The results presented in this section are obtained from 100 simulation runs for each workload type on each multihop topology.
In all runs, the node closest to the center of the target topology is selected as the base station.
In each run, we generate 50 flows whose sources and destinations are picked as follows:
\begin{itemize}
    \item{\textbf{Data Collection (COL)}: Flows are randomly generated from the nodes to the same base station.}
    \item{\textbf{Data Dissemination (DIS)}: Flows are randomly generated from the same base station to nodes.}
    \item{\textbf{Data Collection and Dissemination (MIX)}: Each flow is randomly selected to use either COL or DIS}
    \item{\textbf{Route Through the Base Station (RTB)}: The source and destination of flows are selected at random}, but the routes are constrained to pass through the base station.
\end{itemize}
\noindent
Each flow is assigned at random to one of three flow classes whose periods and deadlines maintain a 1:2:5 ratio.
For example, if \textsf{Class 1} has a period of 100~$ms$, then \textsf{Class 2} has a period of 200~$ms$, and \textsf{Class 3} has a period of 500~$ms$.
We refer to the period of \textsf{Class 1} as the base period. 
In a run, the base period of the flows is decreased until the workload is unschedulable.
The results of a run are obtained for the smallest base period for which the workload is schedulable. 


\begin{figure}
    \centering
    \begin{subfigure}[b]{0.45\linewidth}
    \centering
    \includegraphics[width=\linewidth]{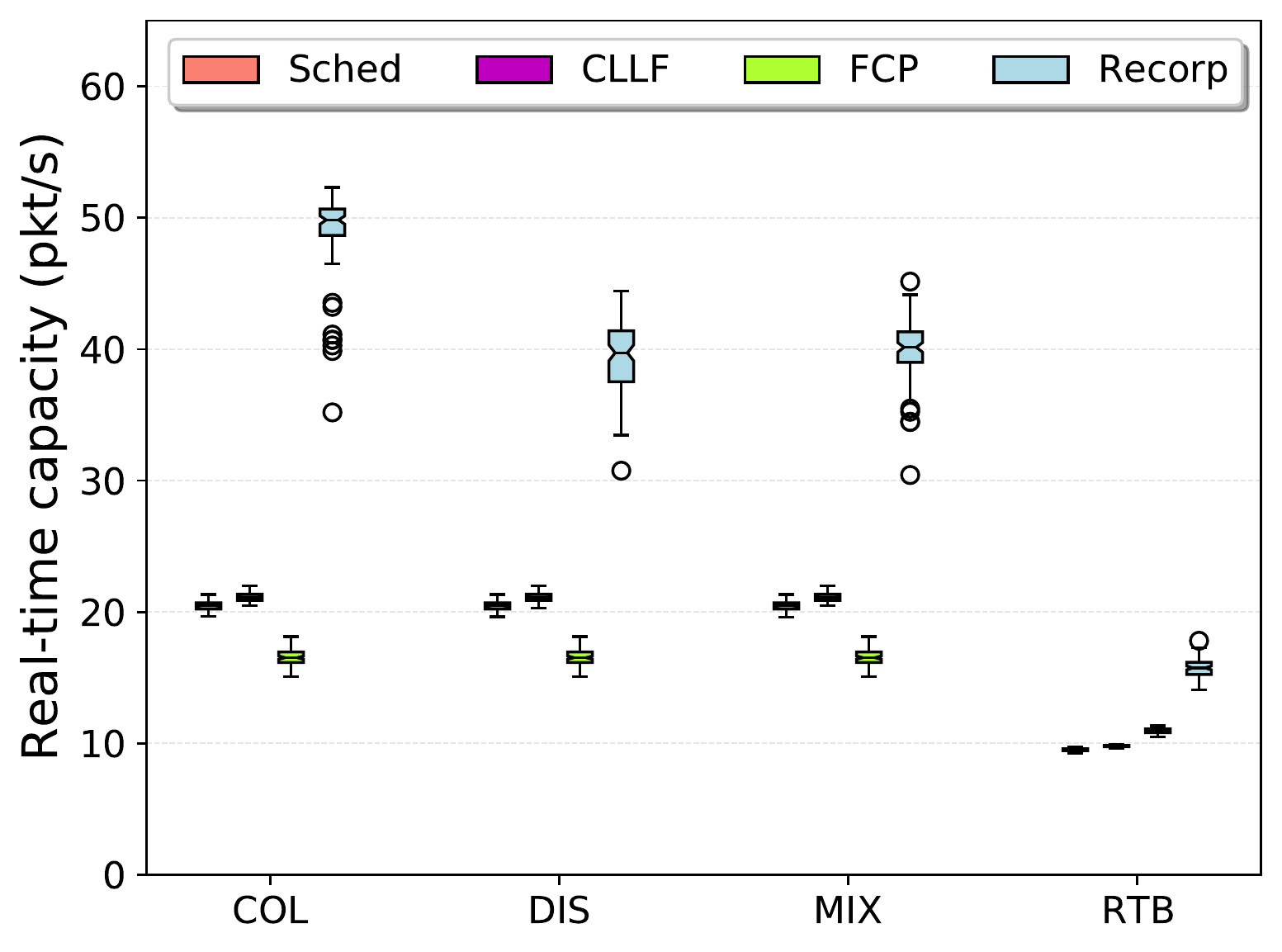}
    \caption{41-node WashU topology}
    \label{fig:multi-hop-capacity-washu}
    \end{subfigure}
    \begin{subfigure}[b]{0.45\linewidth}
    \centering
    \includegraphics[width=\linewidth]{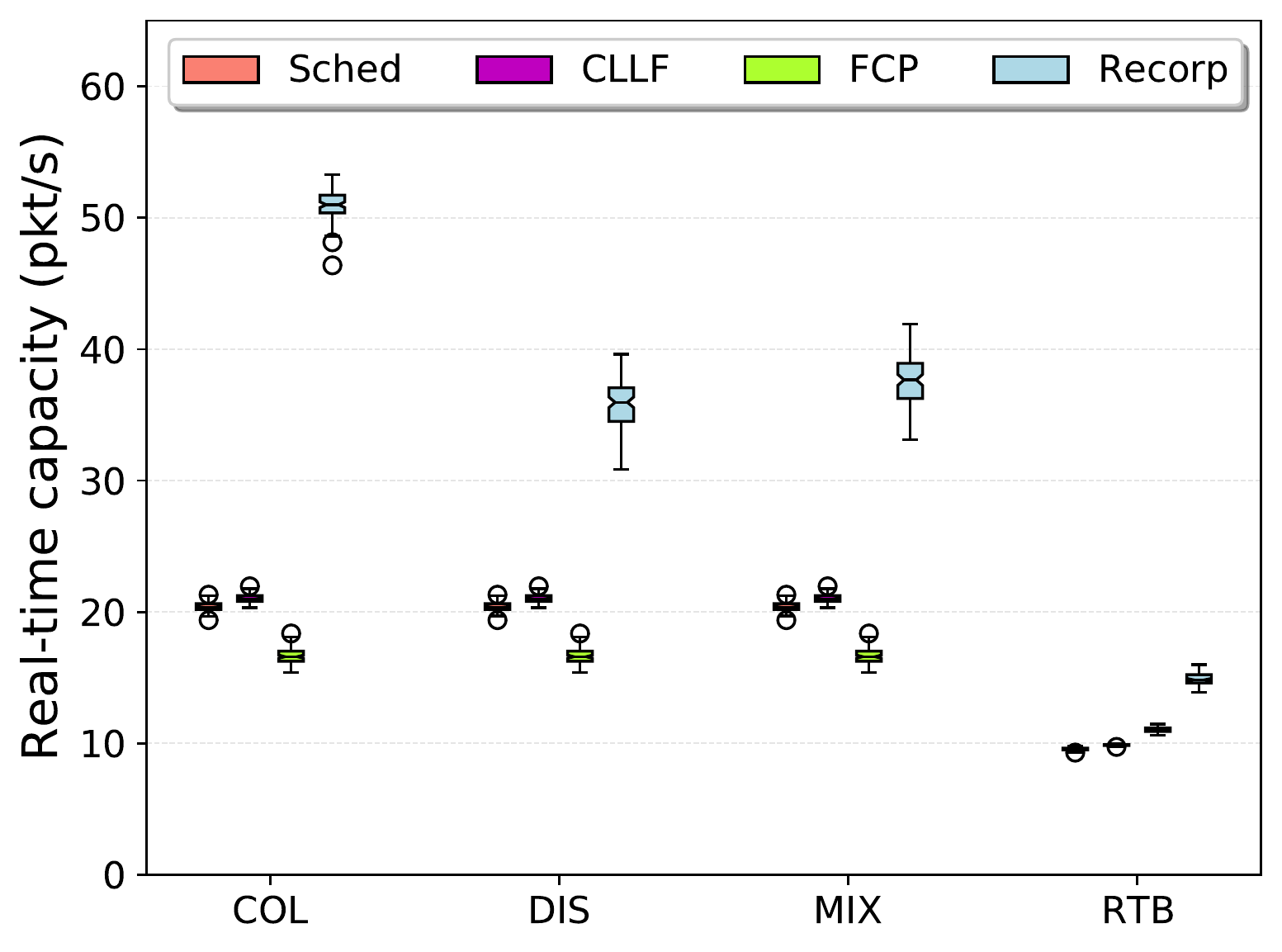}
    \caption{85-node Indriya topology}
    \label{fig:multi-hop-capacity-washu-indriya}
    \end{subfigure}
    \caption{Real-time capacity results.}
\end{figure}

\textbf{Real-time Capacity and Response Time:}
Figures \ref{fig:multi-hop-capacity-washu} and \ref{fig:multi-hop-capacity-washu-indriya} plot the distribution of the observed real-time capacities for the WashU and Indriya topologies, respectively.
\textsf{FCP} provides a median improvement over \textsf{Sched} and \textsf{CLLF} only for the RTB workload.
Moreover, the improvement is minor, with \textsf{FCP} increasing real-time capacity by only 1.15 pkt/s and 1.16 pkt/s over \textsf{CLLF} for the RTB workload in the WashU and Indriya topologies, respectively.
The improvement over \textsf{Sched} was similar.
For the other workloads where the base station is the source/destination, \textsf{FCP} has worse performance since sharing within a flow reduces only the utilization of the intermediary nodes on a flow's path, but not on the source and destination nodes.
In contrast, \pn outperforms all other protocols.
For example, \pn outperforms the overall next best protocol, \textsf{CLLF}, by a median margin of 28.74 pkt/s, 18.63 pkt/s, 19.05 pkt/s, and 5.94 pkt/s in the WashU topology and 30.00 pkt/s, 14.96 pkt/s, 16.67 pkt/s, and 4.97 pkt/s in the Indriya topology for the COL, DIS, MIX, and RTB workloads, respectively.
Together these results correspond to a median increase in real-time capacity over \textsf{CLLF} of between 
50\% and 142\% across each workload and topology.
Moreover, \pn outperforms both \textsf{Sched} and \textsf{FCP} by similar amounts across all workloads and topologies.

Figures \ref{fig:multi-hop-response-washu} and \ref{fig:multi-hop-response-washu-indriya} show the distribution of the response times of each flow class from the previous experiment for the MIX workload (including all runs).
Consistent with the real-time capacity for the MIX workload, \textsf{FCP} underperforms both \textsf{Sched} and \textsf{CLLF} with one exception.
For both topologies, \textsf{FCP} provides a slightly lower median response time than \textsf{CLLF} for \textsf{Class 2}.
The reason for this, and the reason that \textsf{CLLF} has a higher response time than \textsf{Sched} across all workloads and topologies, is due \textsf{CLLF} making scheduling decisions as a function of remaining conflict-aware laxity.
The consequence of this approach is that \textsf{CLLF} occasionally allows lower priority flows to preempt higher priority flows.
In contrast, \pn maintains deadline-monotonic prioritization and reduces the response time for all classes
in both topologies, with particularly good performance for the middle and lowest priority flow classes.
Specifically, \pn decreased the median response time in the WashU topology by 0.11 s, 0.40 s, and 2.50 s and in the Indriya topology by 0.13 s, 0.48 s, and 2.32 s over the next best protocol, \textsf{Sched}, for flow \textsf{Class 1}, \textsf{Class 2}, and \textsf{Class 3}, respectively.
Similar trends and performance differences were observed for the other workloads under all topologies, with one exception. \textsf{FCP} slightly outperformed \textsf{Sched} in the RTB workload across flow classes and topologies.
However, \pn still significantly outperformed \textsf{Sched}, \textsf{CLLF}, and \textsf{FCP}.
These results indicate \emph{\pn policies can significantly improve real-time capacity and response times for common IIoT workloads.}

\begin{figure}
    \centering
    \begin{subfigure}[b]{0.45\linewidth}
    \centering
    \includegraphics[width=\linewidth]{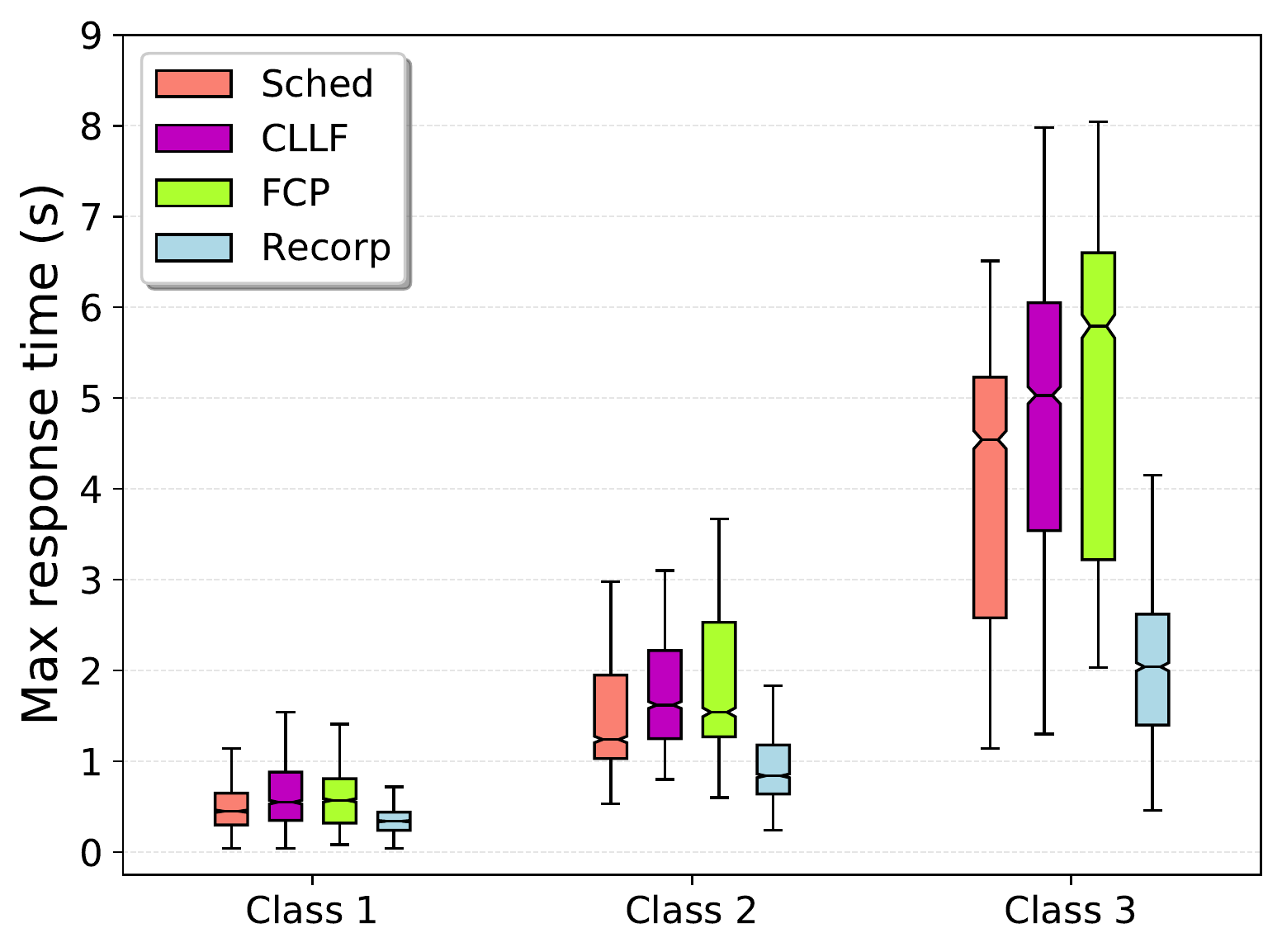}
    \caption{41-node WashU topology}
    \label{fig:multi-hop-response-washu}
    \end{subfigure}
    \begin{subfigure}[b]{0.45\linewidth}
    \centering
    \includegraphics[width=\linewidth]{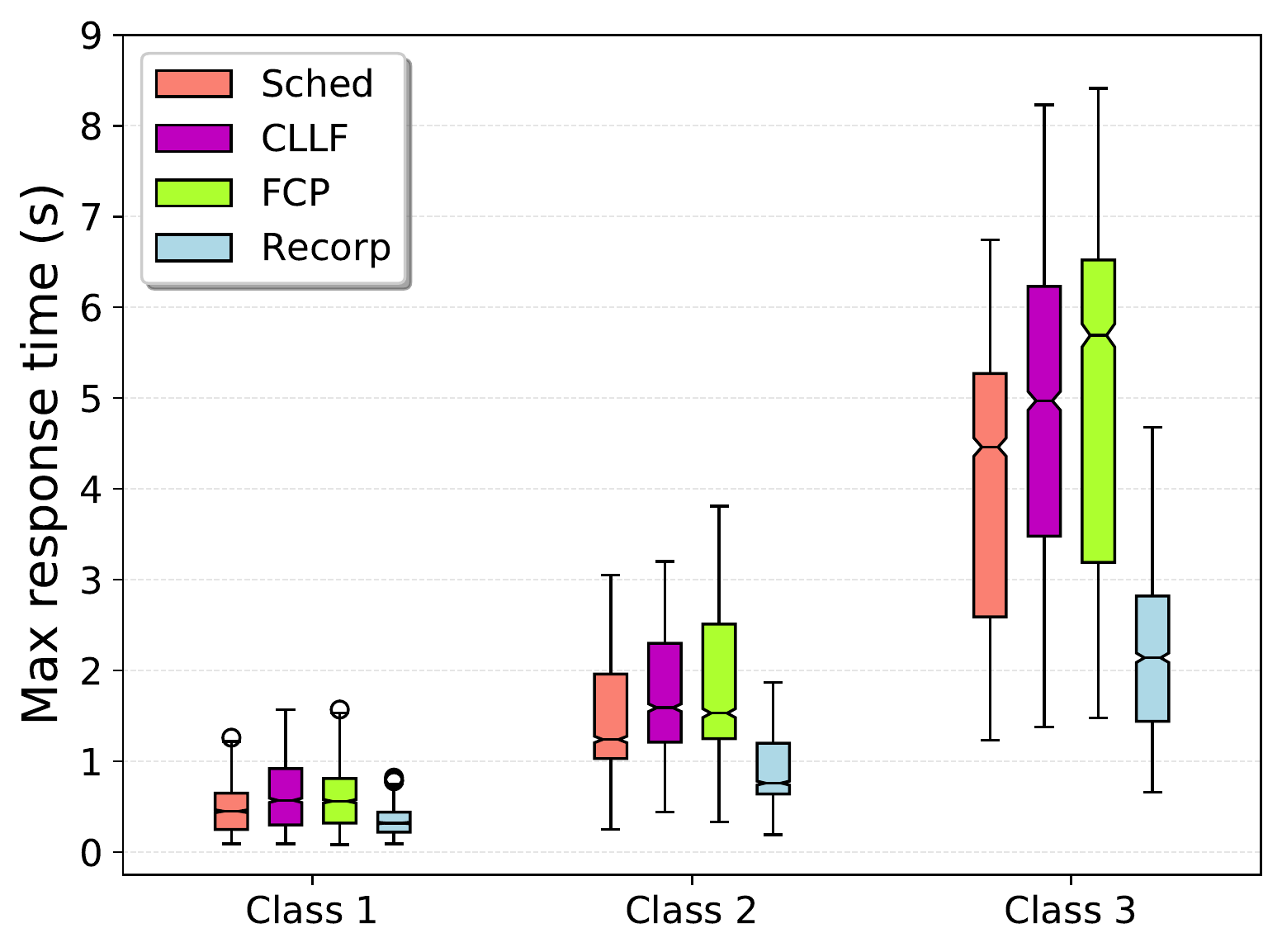}
    \caption{85-node Indriya topology}
    \label{fig:multi-hop-response-washu-indriya}
    \end{subfigure}
    \caption{Response time per flow class.}
\end{figure}

\begin{figure}
    \centering
    \begin{subfigure}[b]{0.45\linewidth}
    \centering
    \includegraphics[width=\linewidth]{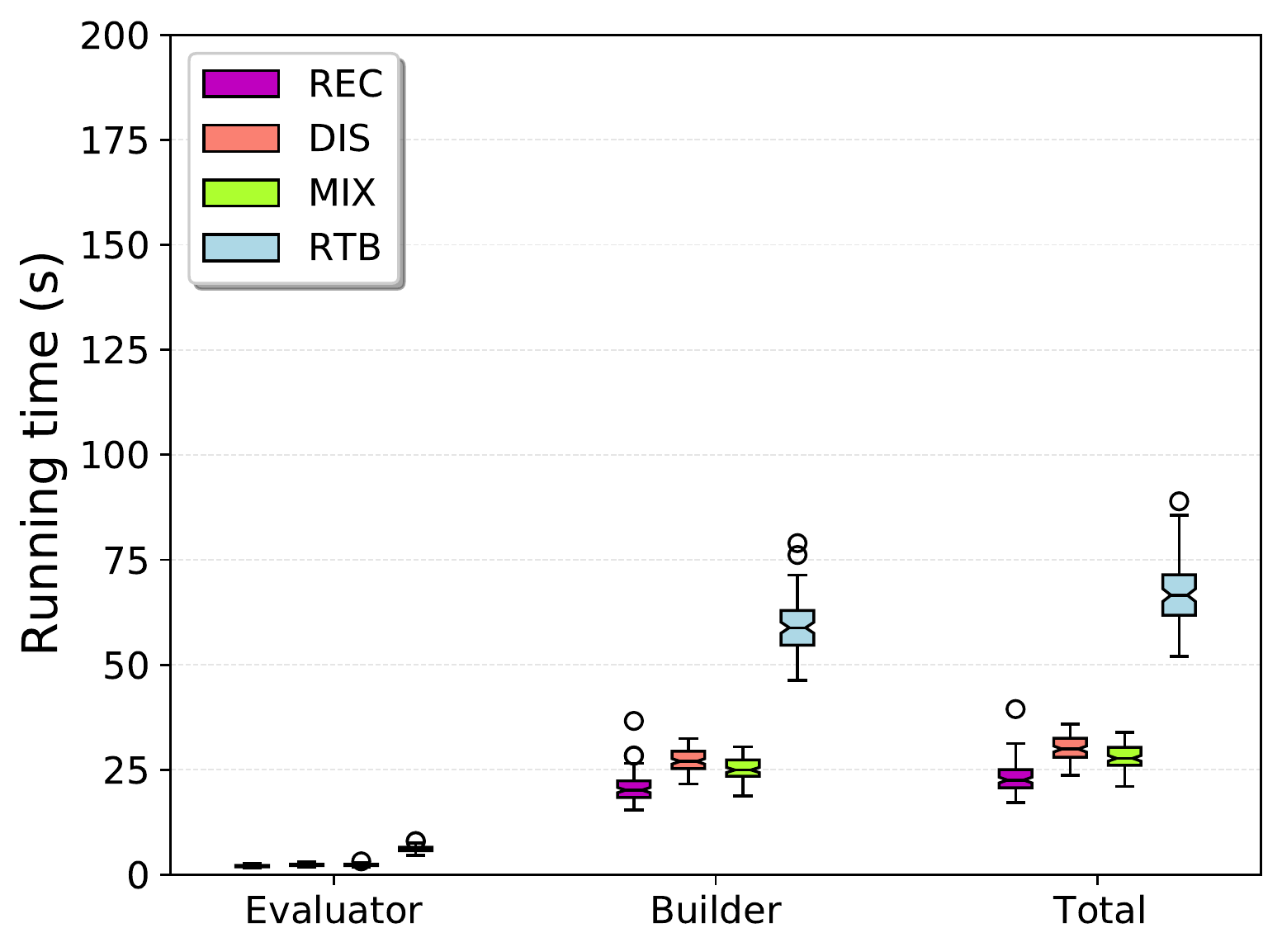}
    \caption{41-node WashU topology}
    \label{fig:multi-hop-synthesis-washu}
    \end{subfigure}
    \begin{subfigure}[b]{0.45\linewidth}
    \centering
    \includegraphics[width=\linewidth]{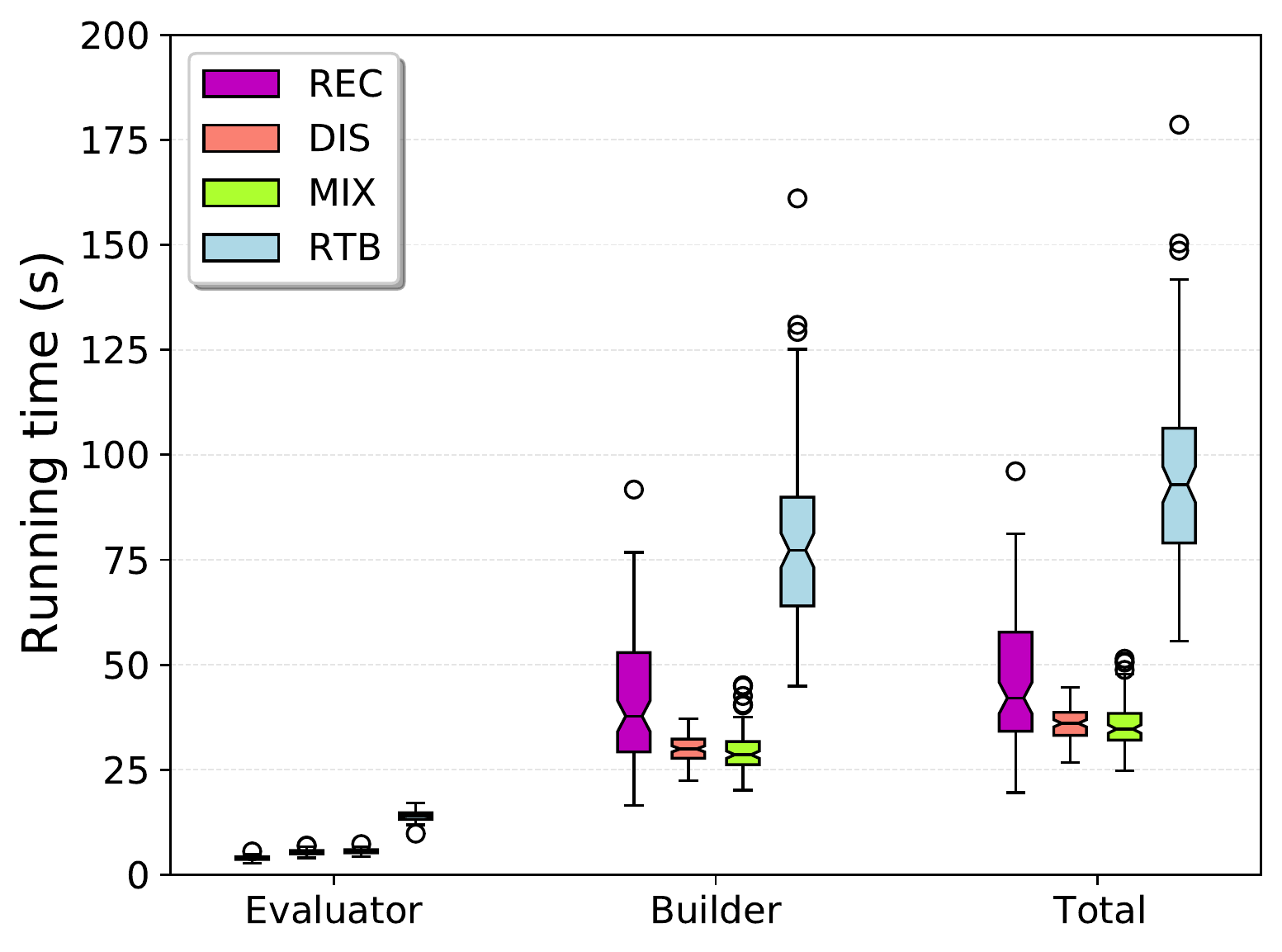}
    \caption{85-node Indriya topology}
    \label{fig:multi-hop-synthesis-washu-indriya}
    \end{subfigure}
    \caption{Synthesis time.}
\end{figure}

\textbf{Synthesis Time:}
Next, we turn our attention to the feasibility of synthesizing policies.
Typical IIoT systems have workloads that are stable for tens of minutes, which justifies synthesizing \pn policies.
We divided the total time to synthesize a policy into two categories:
the time the \evaluator spends managing the system state and the time the \builder spends solving ILPs to determine the \pullops in each slot.
Figures \ref{fig:multi-hop-synthesis-washu} and \ref{fig:multi-hop-synthesis-washu-indriya} plot the distribution of the execution times for each workload under the WashU and Indriya topologies, respectively.
The median total synthesis time is below $93$~$s$ for all workloads and both topologies.
The synthesis time of the route through the base station is significantly higher than the other workloads, as flows tend to have longer paths.
This results in more states to be managed and longer schedules.
The \builder tends to be the most expensive, followed by the \evaluator.
We plan to explore ways to reduce the synthesis time further.  
These results indicate that \emph{it is feasible to synthesize policies within 1--3 minutes for realistic networks.}

\textbf{Threshold Link Reliability Model Evaluation:}
Next, consider \pn's reliability guarantees.
\pn uses a safe lower bound on a flow's end-to-end reliability under the TLR model (i.e., when the link quality of a flow exceeds \pmin) described in Theorem \ref{th:monotonic}.
We are interested in providing simulated and empirical evidence that the lower bound is safe.
Additionally, when the link quality degrades below \pmin, \pn provides no performance guarantees.
However, the end-to-end reliability of flows should degrade gracefully as link quality falls below \pmin.

To this end, we simulated a representative MIX workload \pn policy on the WashU topology with $\pmin = 70\%$.
We varied the link quality from 50\% to 100\% in increments of 5\%. 
For each setting, we simulated 1,000,000 hyperperiods and recorded each flow instance's outcome in each hyperperiod, delivering their data successfully or otherwise. 
For each instance, we computed the probability of delivering its data and plotted the distribution of all instances as ``Simulated'' in Figure \ref{fig:washUMixSimulated}.
Additionally, we used the \evaluator to compute the lower bound on each instance's reliability.
We plotted the worst-case reliability across all instances in the same figure as ``Predicted worst case''.
The worst-case bounds computed by the \evaluator are smaller than those predicted through simulations for all test link qualities indicating that they are safe (i.e., Theorem \ref{th:monotonic} holds).
As expected, when link quality exceeds $\pmin=70\%$, all instances had reliability above their target end-to-end reliability of 99\%. 
When the link quality is below $\pmin = 70\%$, \pn provides no guarantees regarding the reliability of flows.
Nevertheless, the results indicate that the reliability of flows degrades gracefully as link quality deteriorates.
In the next section, we additionally validate the safety of the TLR model on a real testbed.

\begin{figure}[ht]
    \centering
    \includegraphics[width=.45\linewidth]{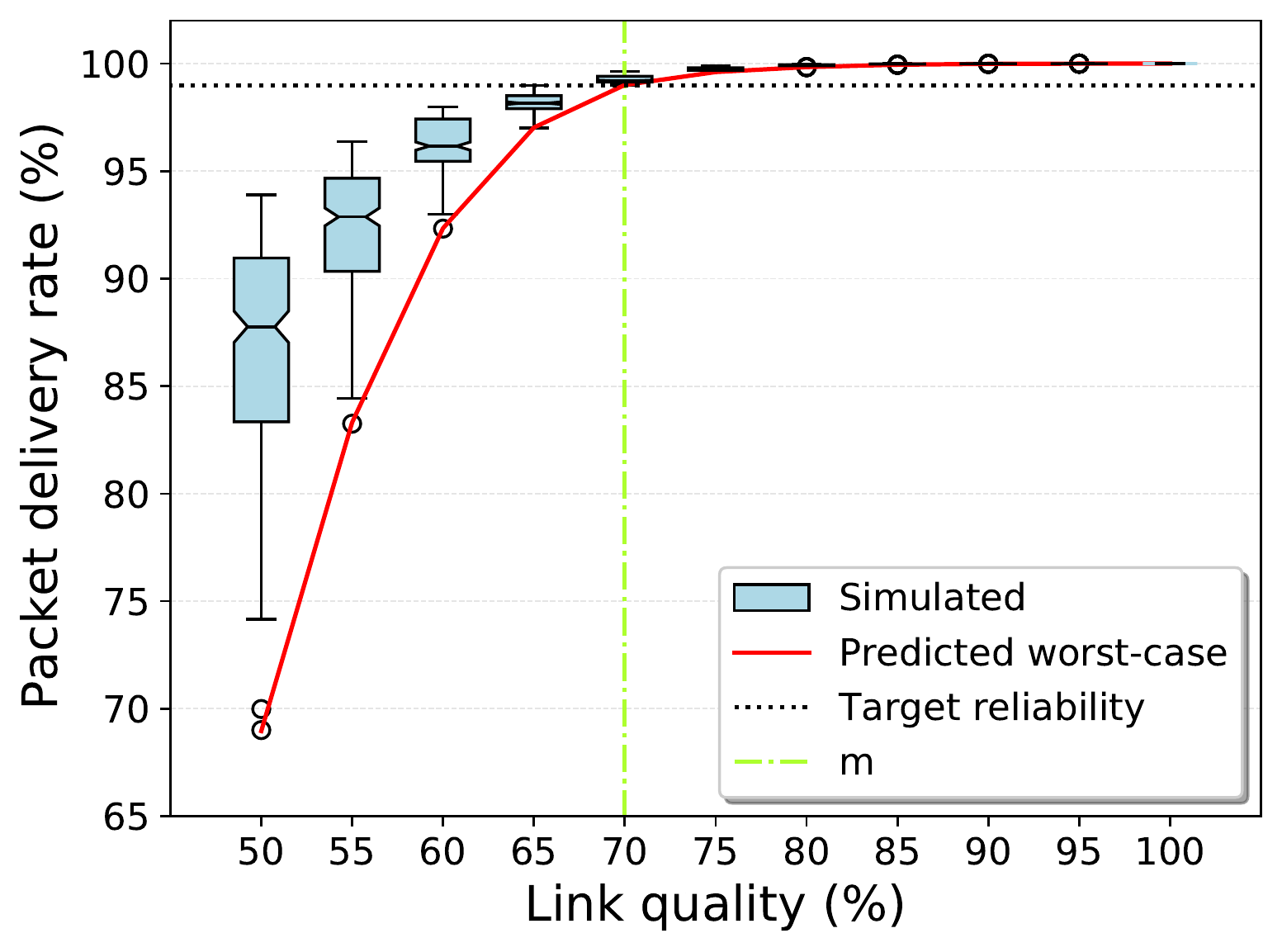}
    \caption{PDR at different link qualities for a representative MIX workload in the WashU topology for $\pmin = 70\%$}
    \label{fig:washUMixSimulated}
\end{figure}



\subsection{Testbed Results}
\label{sec:testbed-resuls}

We evaluated \pn and the baselines on a testbed of 16 TelosB motes deployed at the Univesity of Iowa (See Figure \ref{fig:realTopology}).
At the beginning of each protocol's hyperperiod 3 slots are reserved for a broadcast graph that is used to control traffic and time synchronization.
When a parent broadcasts a packet, it includes its current time in the packet. The children detect the start-of-frame delimiter upon receiving the packet and adjust their clocks to match their parent’s.
We consider a data collection workload that involves ten flows with equal periods whose routes are included in Figure \ref{fig:realWorkload}.
We configured \pn, \textsf{Sched}, and \textsf{FCP} to provide an end-to-end reliability of $99\%$ when $\pmin = 70\%$.
We did not consider \textsf{CLLF} in this experiment since \textsf{CLLF} provided nearly identical performance to \textsf{Sched}.
The experiments use 802.15.4 channels 11, 12, 13, and 14, which overlap with the 802.11g WiFi network co-located in the building.
We have evaluated the performance of \pn with and without additional interference generated by a laptop near the base station, which transmitted ping packets at a rate of 1.5 Mbps.
When no interference was present, all flows met their end-to-end reliability, and the quality of the links exceeded $\pmin=70\%$.
In the following, we will focus on when interference was present to evaluate \pn's ability to adapt in an environment with significant link quality variation.
We organized our experiments into multiple runs, each run 
consisting of running the schedule/policy of each protocol for one hyperperiod and storing the outcome of each transmission to flash at the end of the run.
The reported results were obtained from releasing 10,000 packets for each protocol (i.e., 10,000 runs) over approximately 6 hours.

\begin{figure}
    \centering
    \begin{subfigure}[b]{.144\textwidth}
        \centering
        \includegraphics[width=\textwidth]{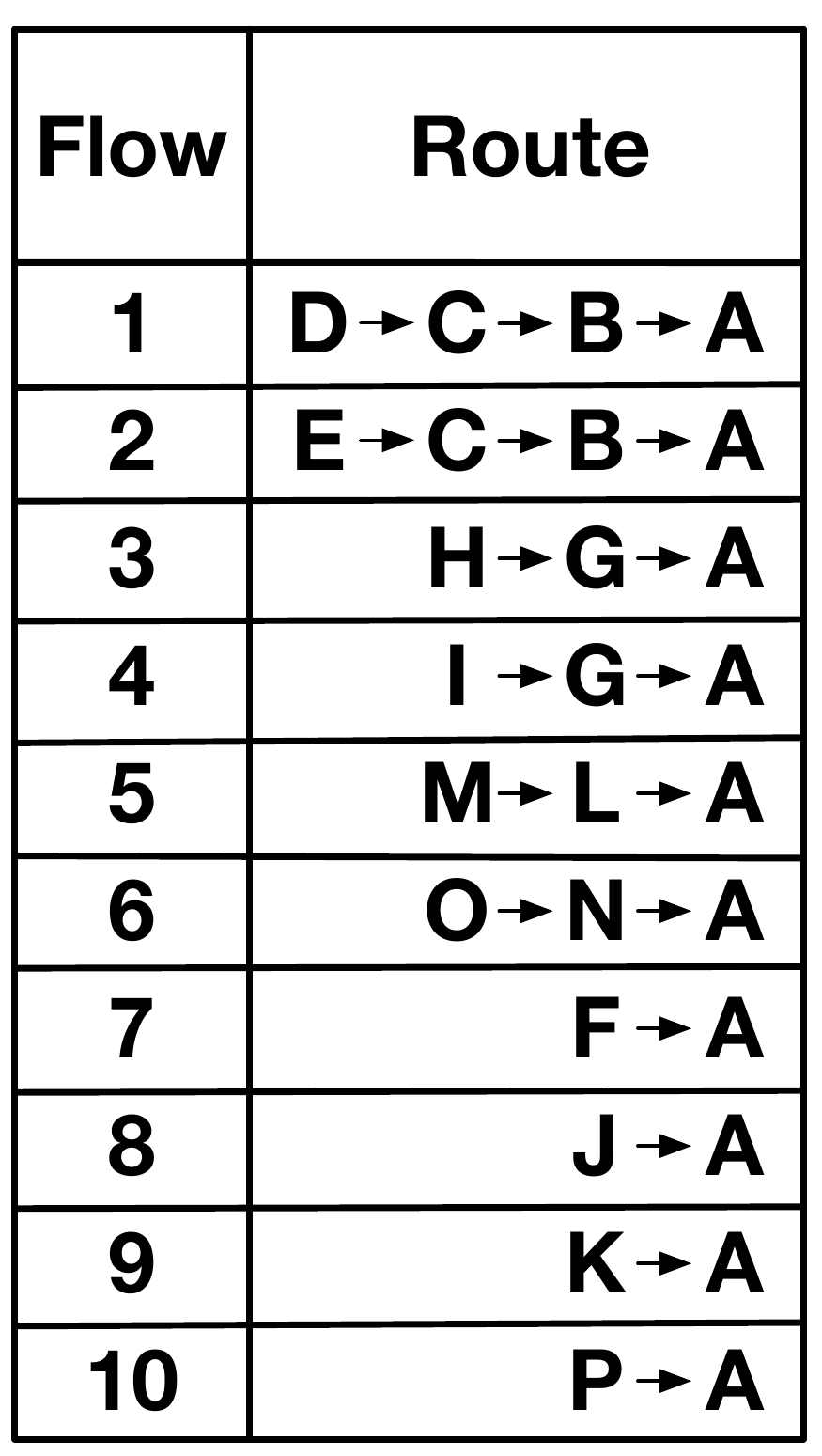}
        \caption[Network2]%
        {{\small Flows}}    
        \label{fig:realWorkload}
    \end{subfigure}
    \begin{subfigure}[b]{.59\textwidth}  
        \centering 
        \includegraphics[width=\textwidth]{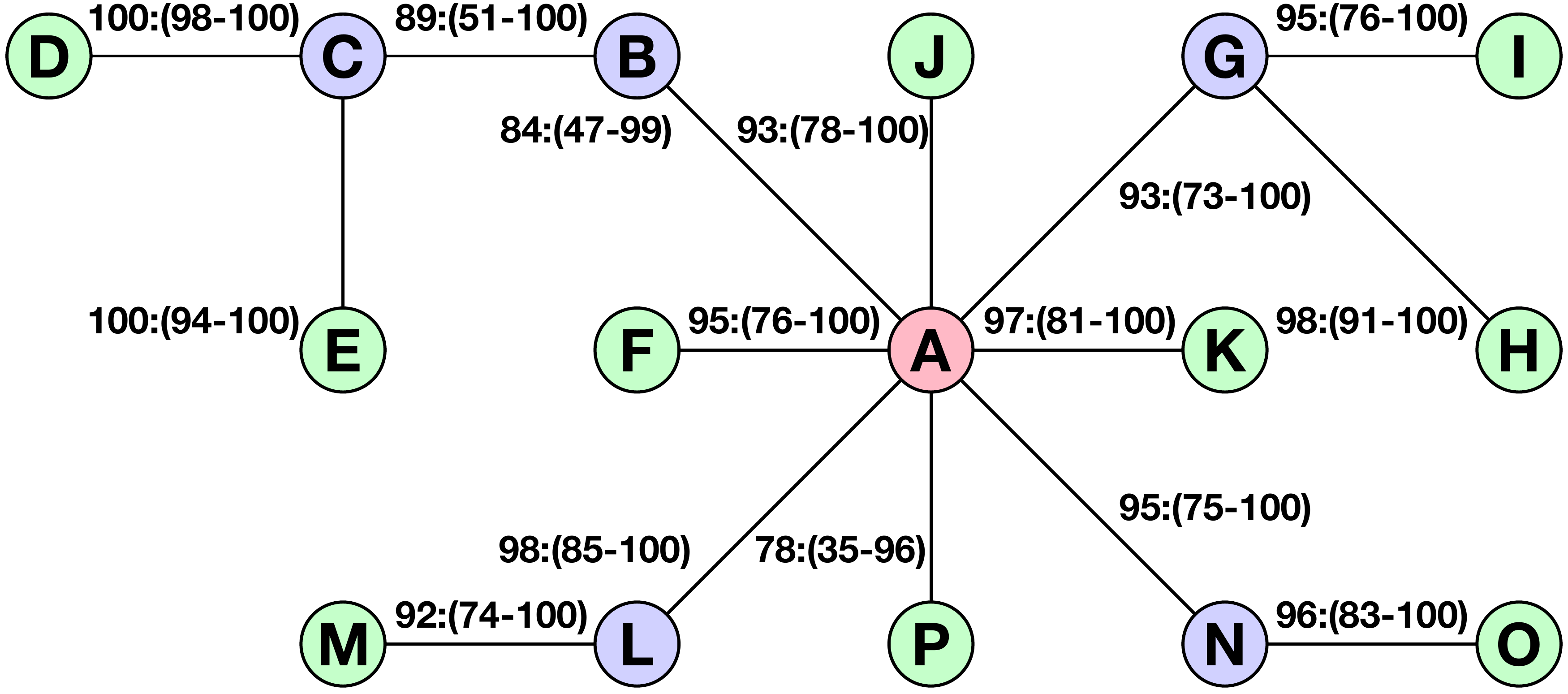}
        \caption[]%
        {{\small Topology}}    
        \label{fig:realTopology}
    \end{subfigure}
    \caption[]
    {\small Testbed topology and flow routes.  The green, purple, and red nodes indicate the flow sources, intermediary nodes, and the base station, respectively. Link quality (with interference) was calculated over a sliding interval of 100 runs (about 200s).  The min, median, and max observed link quality over all intervals for each link is given as median:(min-max).} 
    \label{fig:realTopologyWorkload}
\end{figure}

\textbf{Real-time Capacity and Reliability:}
We determined the maximum rates of the ten data collection flows that can be supported using \pn, \textsf{Sched}, and \textsf{FCP}.
\pn provides a real-time capacity of 38.46 pkt/s compared to 19.6 pkt/s and 18.2 pkt/s provided by \textsf{Sched} and \textsf{FCP}, respectively.
The real-time capacity of \pn is 96\% higher than \textsf{Sched}.
This result is consistent with the multihop experiments where \pn significantly outperforms the baselines. 
Next, we will evaluate whether the improved capacity comes at the cost of lower reliability.

\begin{figure}
 \centering
\begin{subfigure}[b]{.45\linewidth}
\vskip 0pt
\centering
\includegraphics[width=\linewidth]{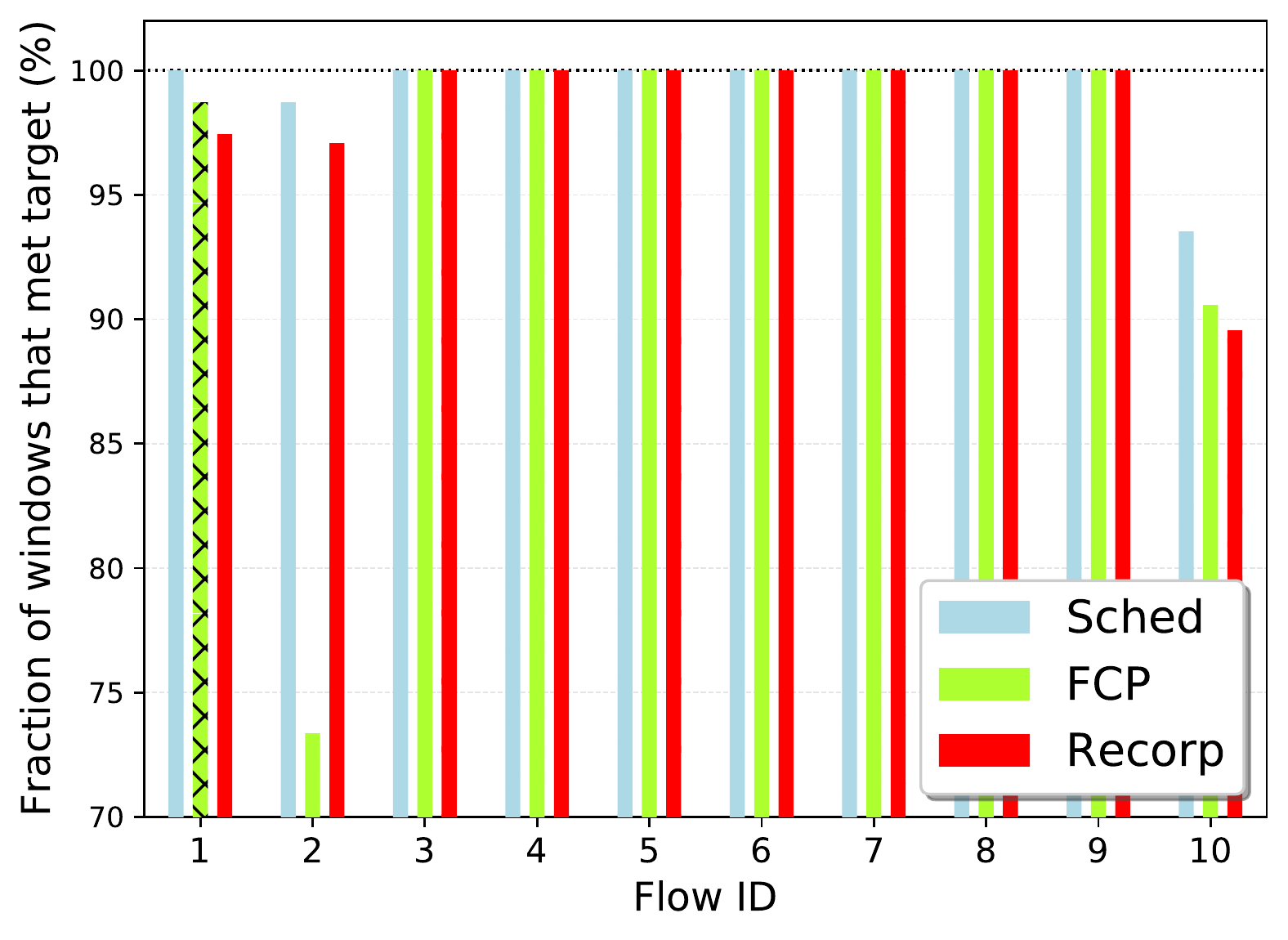}
\caption{End-to-end flow reliability}
\label{fig:totalE2E}
\end{subfigure}
\begin{subfigure}[b]{.45\linewidth}
\vskip 0pt
\centering
\includegraphics[width=\linewidth]{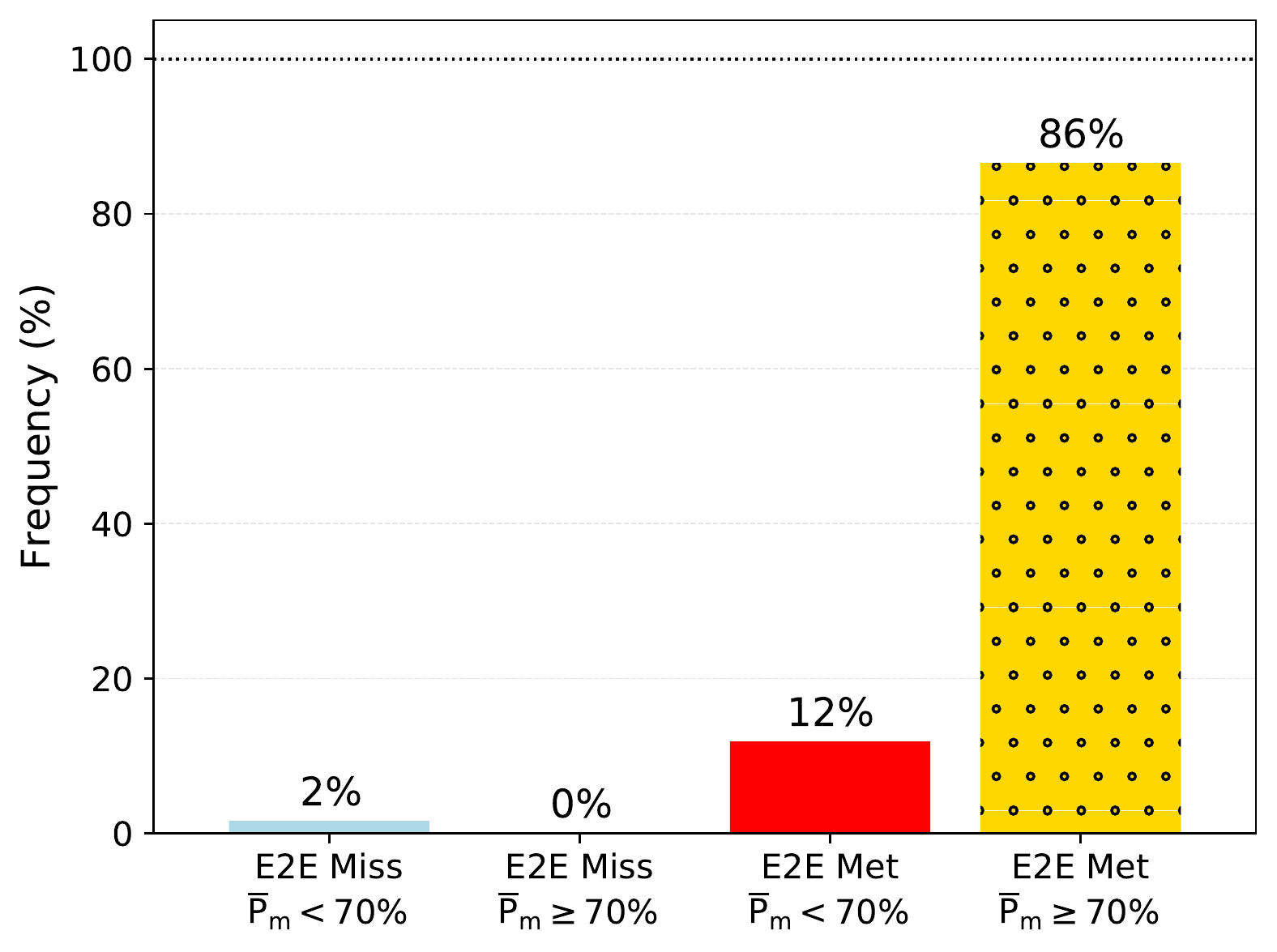}
\caption{Safety of reliability bound}
\label{fig:mprrSafety}
\end{subfigure}
\caption{Evaluating the safety of the reliability bound.}
\end{figure}

We computed the packet delivery rate (PDR) over sliding windows of 100 runs.
In Figure \ref{fig:totalE2E}, we plot the fraction of windows that met the end-to-end reliability target of 99\% for each protocol.
The lowest reliability was observed for \textsf{FCP}'s flow 2.
We found that the root cause behind the lower performance of \textsf{FCP} is the contention-based mechanism used to arbitrate access to the entries shared by the links of a flow.
\textsf{FCP} prioritizes the transmission of nodes closer to the flow's destination by having them transmit at the beginning of the slot while the other nodes only transmit after clear channel assessment (CCA) indicates the slot is not used.
In the presence of WiFi interference, CCA was not a robust indicator of transmissions.
This experience highlights the potential advantage of using receiver-initiated pulls over contention-based approaches that rely on CCA.


\pn policies guarantee probabilistically that the end-to-end reliability constraints are met as long as the quality of all used links exceeds a minimum packet reception rate $\pmin$.
When the quality of the links falls below $\pmin$, we provide no guarantees on the end-to-end reliability of flows.
We evaluate whether our guarantee holds as follows.
Based on the trace of successes and failures observed during the experiment, we fit a Bernoulli $\bar{P}_m$ random variable to lower bound the observed failure distributions.
Accordingly, \pn's analytical bounds on flow reliability hold only if $\bar{P}_m \ge P_m = 70\%$.
Figure \ref{fig:mprrSafety} classifies each window of 100 runs into the following cases:
\begin{enumerate}

    \item Case $\bar{P}_m \ge 70\%$, E2E Met: For 86\% of the windows, the minimum link quality met or exceeded 70\% (i.e.,  $70\% = \pmin \le \bar{P}_m$). Over all these windows, \pn policies indeed guaranteed that the end-to-end reliability of all flows exceeded the 99\% target.
    \item Case $\bar{P}_m \ge 70\%$, E2E Miss: There are no cases where the minimum link quality exceeds 70\%, and the flows do not meet the target 99\% reliability. These first two cases demonstrate that the TLR model is safe since no flows miss their end-to-end reliability targets when the minimum link quality is met.
    \item Case $\bar{P}_m < 70\%$: When the actual link quality falls below the minimum link quality of $\pmin=70\%$, we provide no guarantees on the flow's reliability. Out of the 14\% of windows where $\bar{P}_m < 70\%$, in 12\%, the end-to-end reliability is met, while for the other 2\%, it is not.



\end{enumerate}
\noindent
These experiments show that \emph{\pn policies can significantly improve real-time capacity while meeting the end-to-end reliability of flows as the quality of links fluctuates above the minimum link quality \pmin}. 

\textbf{Effective Adaptation:}
To analyze \pn's ability to adapt to variations in link quality, we consider the trace of \textsf{Sched} and \pn for flow 10, which exhibits the lowest link reliability and highest variability in our experiments.
%
%
Figure \ref{fig:realAdaptability} plots the end-to-end reliability (after retransmissions), 
the parameter $\bar{P}_m$ of a Bernoulli distribution that is fitted to account for the burst of failures observed empirically in each window, and the maximum number transmissions used by \textsf{Sched} and \pn over a trace of 4000~$s$.
Notably, the end-to-end reliability of \textsf{Sched} and \pn is similar during this time frame (Figures \ref{fig:lcE2E} and \ref{fig:ncE2E}).
\pn achieves a similar level of end-to-end reliability by performing more retransmissions, as it is clear from comparing Figures \ref{fig:lcTx} and \ref{fig:ncTx}.
\textsf{Sched} uses 3 -- 4 maximum retransmissions over the course of the hour but notably still briefly missed the end-to-end PDR target.
In contrast, \pn uses between 3 -- 7 retransmissions to combat a slightly lower link quality it experienced and did not miss the end-to-end PDR target over the interval.
Remarkably, \pn can (almost) double the number of retransmissions that may be used for flow 10 over \textsf{Sched} without degrading the performance of other flows.
These results indicate \emph{that \pn can provide higher agility than schedules by using its lightweight and local run-time adaptation mechanism to reallocate retransmissions in response to variations in link quality.}

\begin{figure}
    \centering
    \begin{subfigure}[b]{0.226\textwidth}
        \centering
        \includegraphics[width=\textwidth]{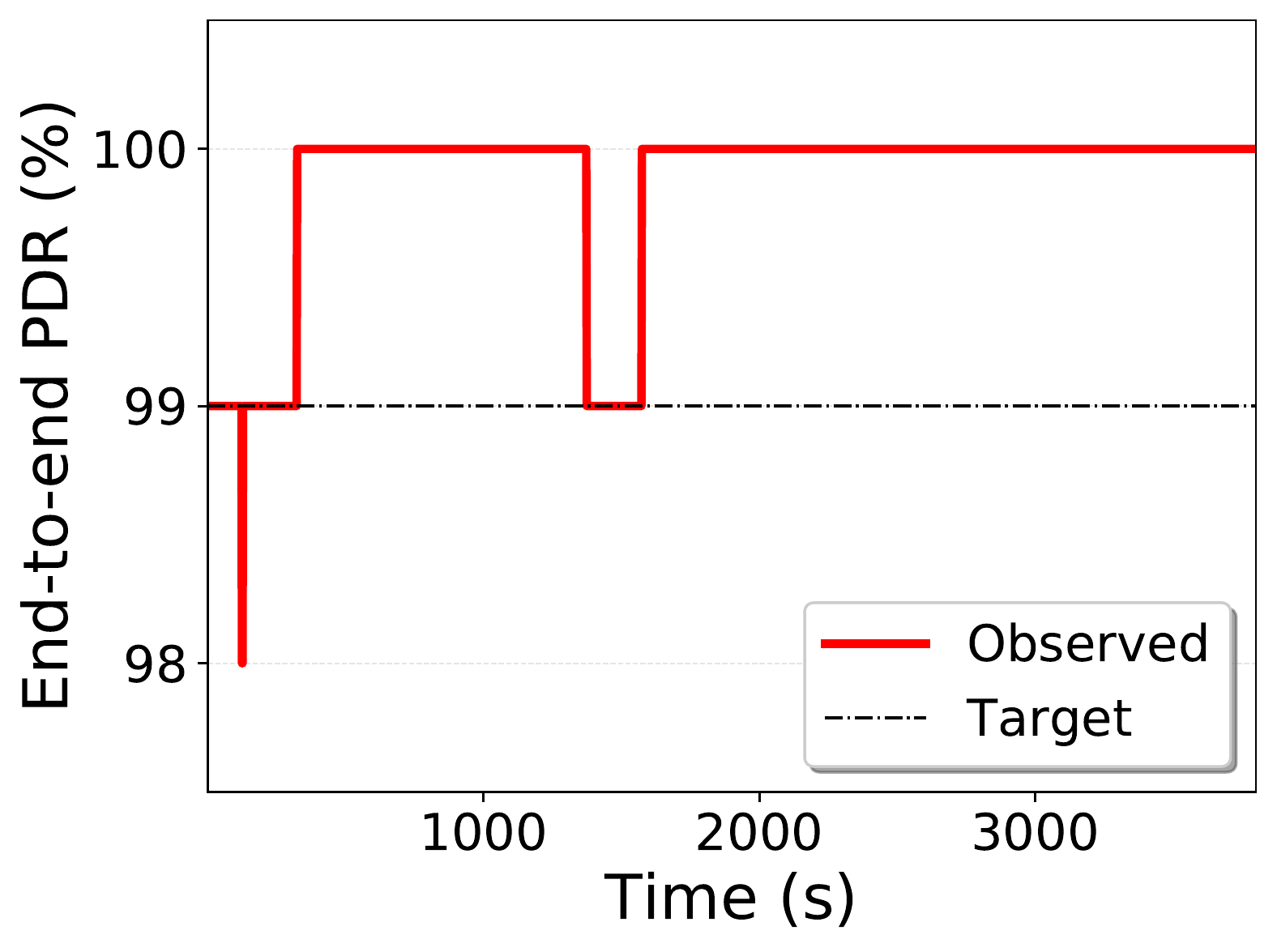}
        \caption[Network2]%
        {{\small \textsf{Sched} reliability }}    
        \label{fig:lcE2E}
    \end{subfigure}
    \begin{subfigure}[b]{0.226\textwidth}  
        \centering 
        \includegraphics[width=\textwidth]{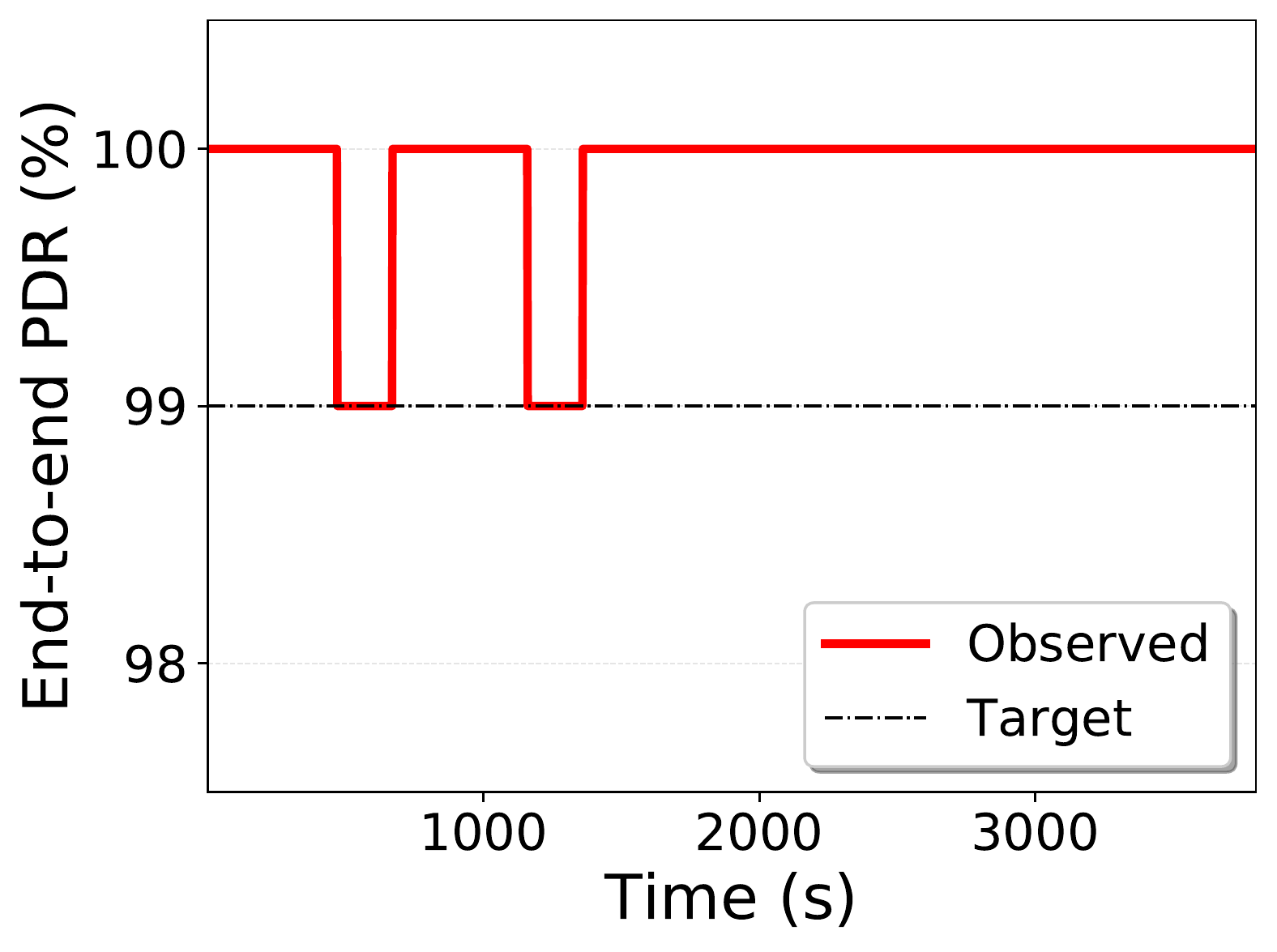}
        \caption[]%
        {{\small \pn reliability}}    
        \label{fig:ncE2E}
    \end{subfigure}
    \begin{subfigure}[b]{0.226\textwidth}   
        \centering 
        \includegraphics[width=\textwidth]{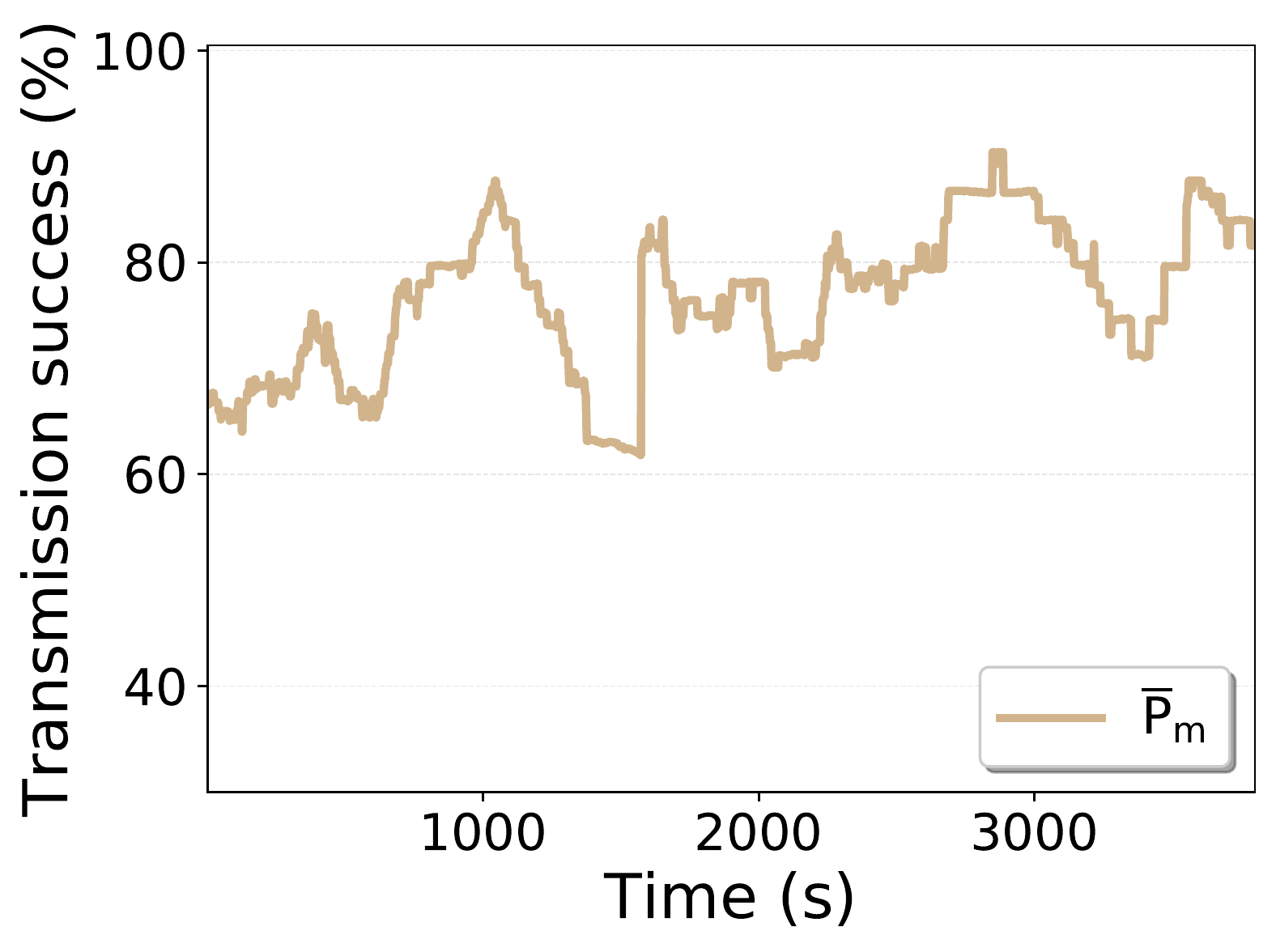}
        \caption[]%
        {{\small Link $\bar{P}_m$ for \textsf{Sched}}}    
        \label{fig:lcPRR}
    \end{subfigure}
    \begin{subfigure}[b]{0.226\textwidth}   
        \centering 
        \includegraphics[width=\textwidth]{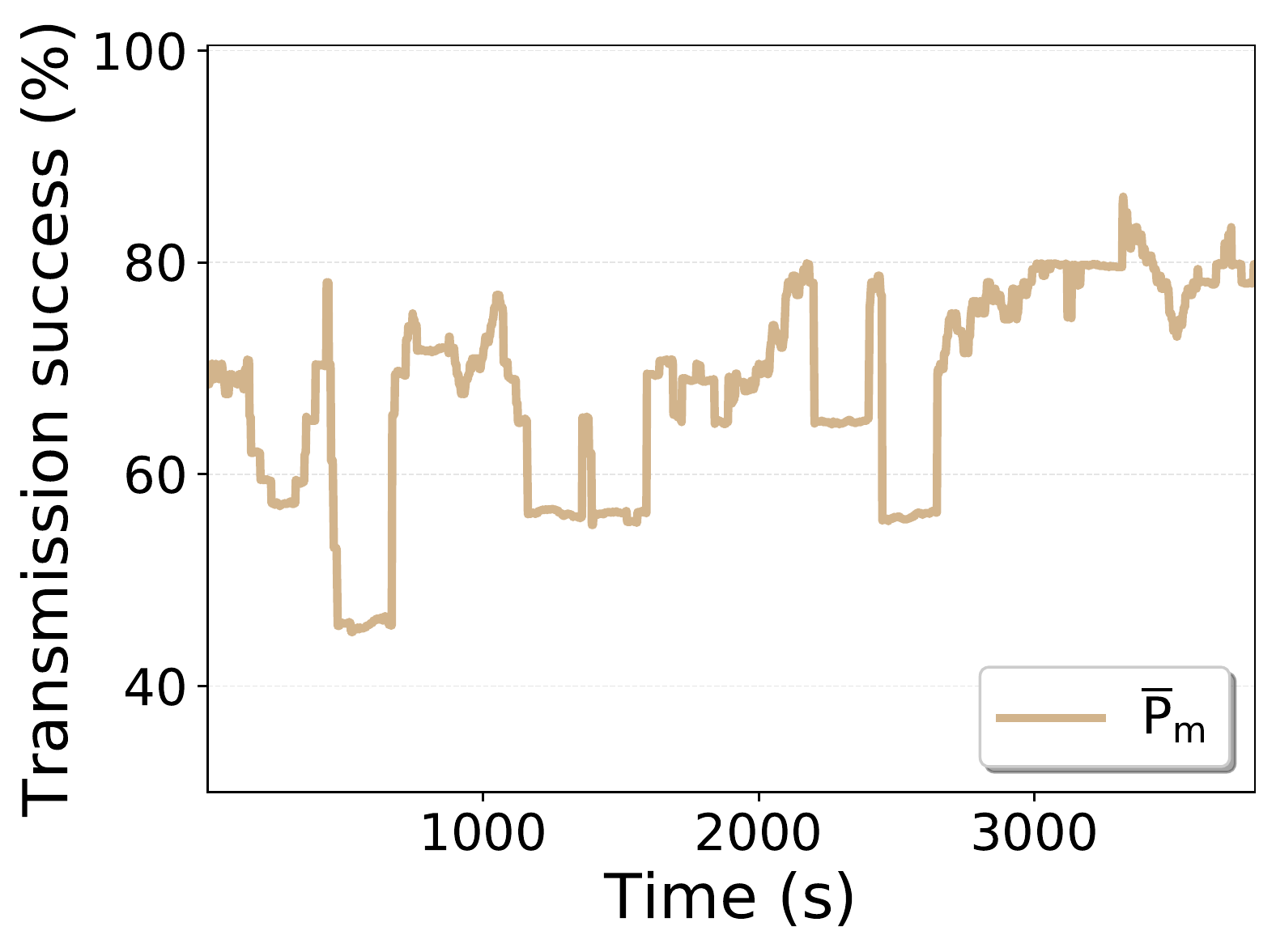}
        \caption[]%
        {{\small Link $\bar{P}_m$ for \pn}}    
        \label{fig:PPRR}
    \end{subfigure}
    \begin{subfigure}[b]{0.226\textwidth}   
        \centering 
        \includegraphics[width=\textwidth]{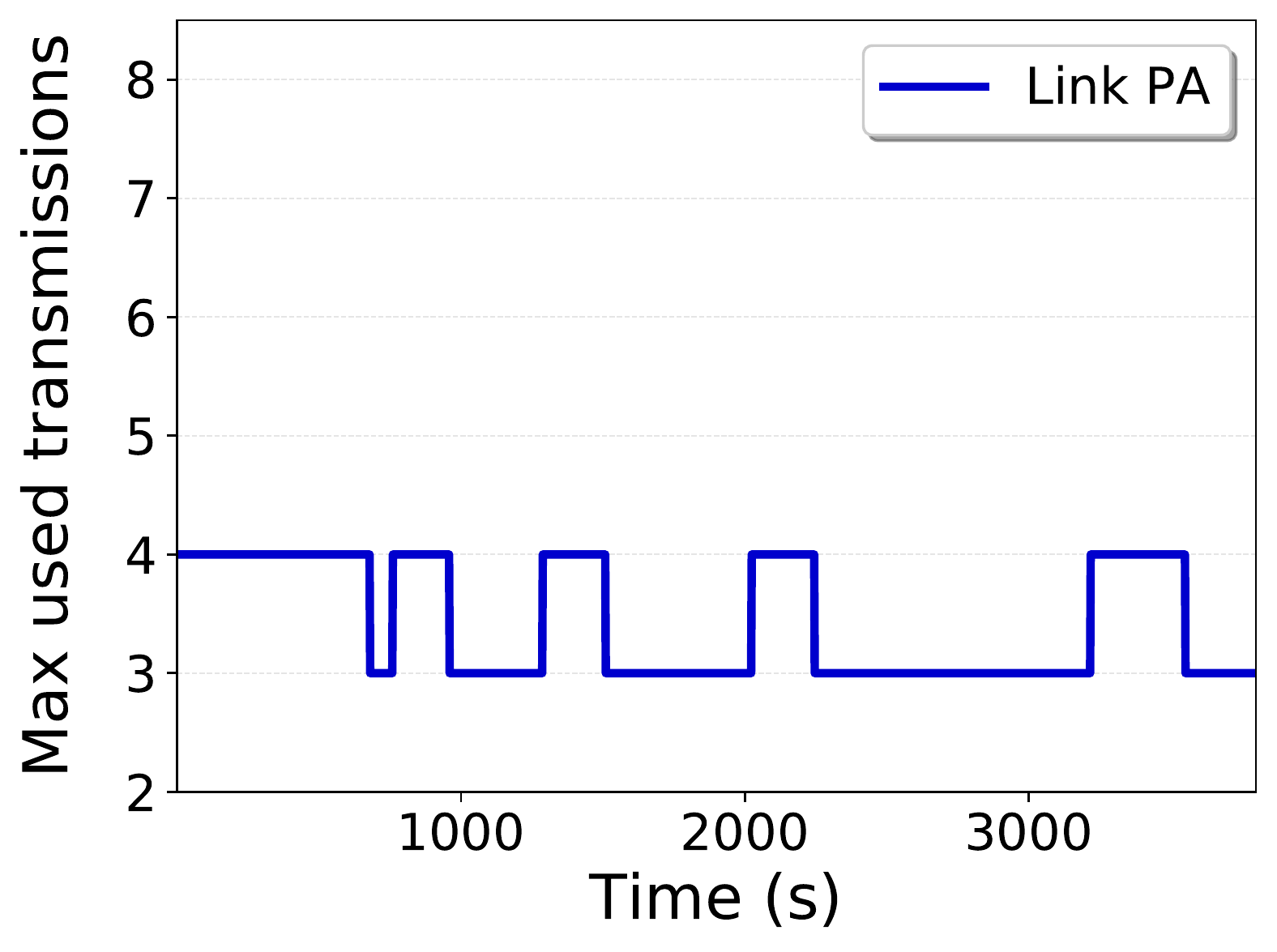}
        \caption[]%
        {{\small \textsf{Sched} max tx. }}    
        \label{fig:lcTx}
    \end{subfigure}
    \begin{subfigure}[b]{0.226\textwidth}   
        \centering 
        \includegraphics[width=\textwidth]{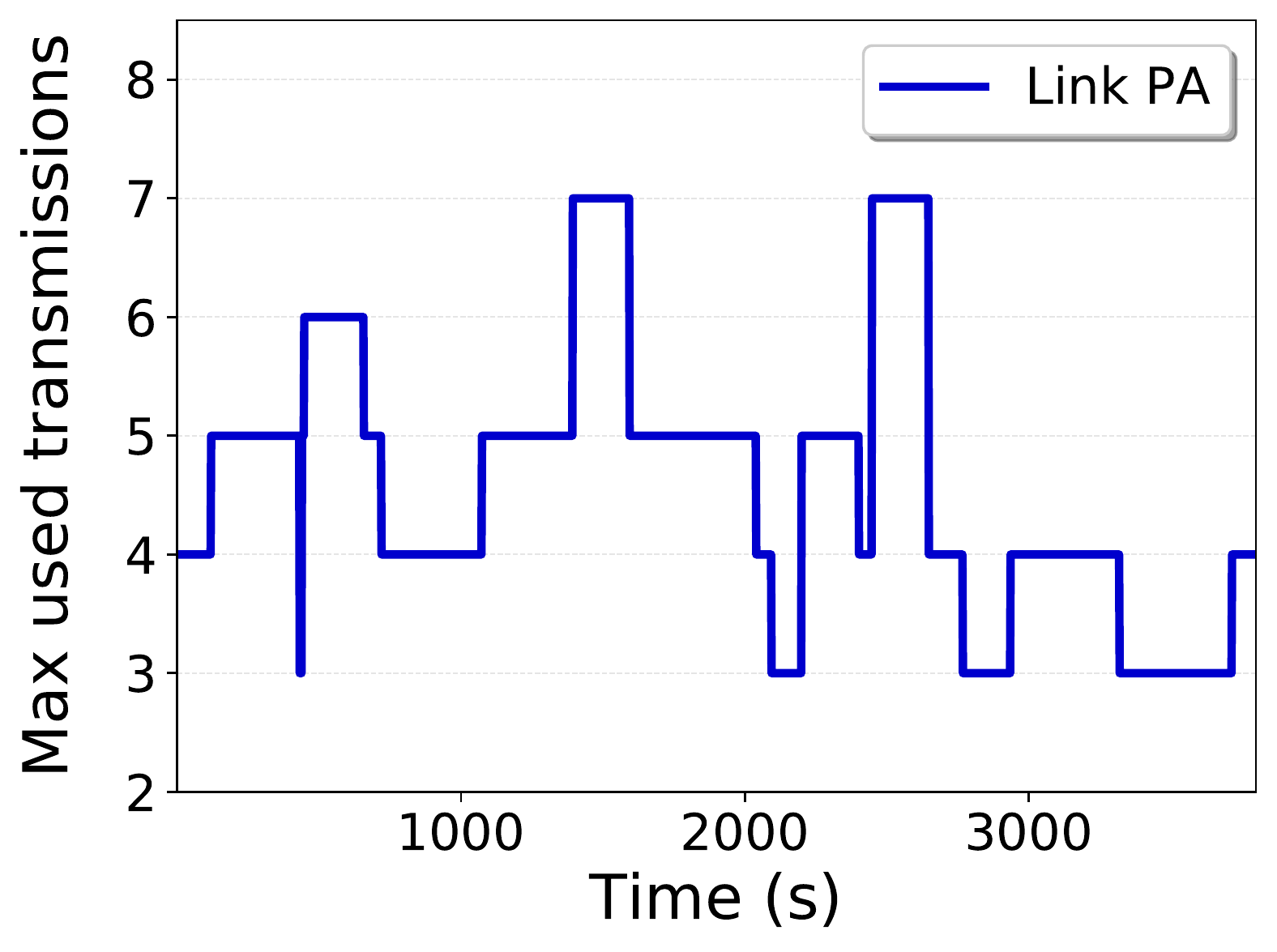}
        \caption[]%
        {{\small \pn max tx.}}    
        \label{fig:ncTx}
    \end{subfigure}
    \caption[ The average and standard deviation of critical parameters ]
    {\small Performance of \pn and \textsf{Sched} on flow 10 over time.} 
    \label{fig:realAdaptability}
\end{figure}

\section{Discussion}
\label{sec:discussion}

\subsection{Deployment}
Wireless networks that support IIoT applications require careful planning and deployment.
The deployment process usually involves profiling the quality of links and the interference on all 16 channels.
The collected statistics are used to ensure that there are redundant routes that connect each node to the base station whose link quality exceeds the \pmin threshold of the TLR model.
Consistent with Emerson's guide for deploying WirelessHART networks, the value of \pmin is usually set to 60\% -- 70\%.
Additionally, the channels that have consistently poor reliability are usually blacklisted~\cite{channel_blacklisting}.

\subsection{Handling Network Dynamics}
\pn's design focuses on supporting the communication needs of IIoT applications with long-running real-time flows.
The network manager uses the current set of flows to build a policy that meets a flow's end-to-end reliability as long as the link quality exceeds \pmin.
This approach makes it feasible to run the same policy for prolonged periods of time without modification.
However, the industrial environment may change and lead to node and link failures.
The primary mechanism used by \pn to adapt to topology changes and node failures is to synthesize new policies.
However, the frequency with which new policies need to be synthesized can be reduced by integrating \pn with multi-hop routing techniques (e.g.,~\cite{dcqs-journal,saifullah2015schedulability}) to allow \pn to tolerate some link or node failures without having to reconstruct policies. 
In the following, we describe a centralized management and control plane that can detect and adapt to node failures and topology changes using an approach similar to WirelessHART.

The network manager uses three different types of specialized flows to implement the control plane.
These specialized flows are periodic, but unlike regular flows, they have a higher priority the regular flows and require different mechanisms to allocate their slots.
A \emph{dissemination flow} is used to disseminate policies to all nodes after its synthesis.
A dissemination flow reserves a single slot during which the base station floods a packet to all the nodes. 
An efficient approach to handle this type of communication is to use GLOSSY floods~\cite{ferrari2011efficient}.
A \emph{join flow} reserves a single slot during which all nodes wait to listen to a fixed channel for nodes to request joining the network.
A node wanting to join will use CSMA techniques to broadcasts its request to join the network.
All nodes receiving the request will forward it to the base station using a \emph{report flow}.
A \emph{report flow} is used to inform the network manager about the status of nodes and links.
A report flow is set up from each leaf node in the upstream graph to the base station.
As the packet of a report flow is forwarded, nodes along the paths may append to the payload node and link health information to be delivered to the network manager. 
Each node along the path is provided a fixed number of bytes that they may use.
As described next, node and link failure reports are prioritized while the remaining space is used for the quality of links that are not currently part of the routing tree.
Node and link health reports can also be piggy-backed onto periodic traffic to improve network agility.

Each node collects statistics about the bursts of packet losses within a window of slots to assess the quality of links currently in use. 
As described in Section \ref{sec:testbed-resuls}, a node uses this information to fit a Binomial distribution whose chance of success $P_{s}$ is sufficiently large to account for the observed burst.  
On the one hand, if $P_{s} < \pmin$, the TLR assumptions are violated, and the network manager is notified immediately about the link failure.
Accordingly, this information is included in the report flow.
On the other hand, if $P_{s} \ge \pmin$, the TLR assumptions are not violated.
This information is not as urgent to the network's operation and is included in a control packet only if there is room available.
This approach can also be extended to collect statistics of links that are currently not in use.
However, a node should only use the reception of the data packets (ignoring the pull requests) to estimate the one-way link quality between itself and the packet's sender.
This information should be included infrequently in control packets to allow the network manager to update the upstream and downstream graphs.

\pn allows nodes to join or leave the network dynamically.
A node wanting to join the network first listens to a fixed channel until it receives the packet of a report flow that allows it to synchronize with the network and learn the join flows' parameters. 
When the join flow is released, the node broadcasts a request to join the network, which is routed to the network manager using the next available report flow.
Upon receiving a join request, the network manager updates the upstream and downstream graphs to include the new node and starts the synthesis of a new policy.
When the synthesis is complete, the policy is disseminated to all nodes, including the new node. 
A node leaving the network uses a report flow to send its request to the base station.

\subsection{Handling Other Traffic}
\pn is optimized for improving the performance of real-time flows that are expected to carry the bulk of the traffic in IIoT applications.
However, other types of traffic may also exist.
For example, IIoT applications may benefit from supporting even-triggered emergency communication in response to unsafe situations or failures. 
\pn can support emergency communication using techniques proposed in \cite{emergency_alarms}:
each slot is modified such that emergency traffic is transmitted at the beginning of the slot. 
In contrast, regular traffic is transmitted after a short delay. 
Other examples of traffic include aperiodic communication.
The simplest solution for handling these transmissions is to dedicate slots for their transmission periodically. 
Transmissions during these slots are done using typical CSMA/CA techniques. This approach reserves a portion of the bandwidth for other types of traffic. Moreover, it is straightforward to account for these additional slots in our analysis.

\subsection{High Data-Rates and Energy Efficiency}
\pn is designed for IIoT applications that require high data rates and usually use grid-powered nodes.
Our simulation and testbed results are performed using the IEEE 802.15.4 physical layer.
The standard supports a maximum packet size of 128 bytes and a maximum data rate of 250 kbps.
Under these settings, the maximum real-time capacity from the testbed experiments equates to only approximately 39 Kbps.
While this data rate may be able to meet the real-time and reliability requirements of some high data rate sensors such as torque and temperature sensors with update rates on the order of 10 -- 500 ms \cite{highDataRateExamples}, it is unlikely to be sufficient for microphones and cameras.
For this type of application, \pn can be used unmodified with IEEE 802.15.4a UWB.
IEEE 802.15.4a provides a significantly higher data rate of 27.24 Mbps and some UWB radios (e.g., DWM1001) support packets as large as 1024 bytes.
In future work, we will explore using other physical layers to extend \pn's applicability further.

One of the limitations of \pn is its potentially high energy usage.
Indeed, an entry in the scheduling matrix using the \textsf{Sched} protocol will involve at most two nodes using their radios.
In contrast, a \pn \pullop may involve as many as five nodes for a service list of size four.
The nodes in the service list size must turn on their radio for a short duration to determine whether the coordinator will request information from them.
As a result, the real-time capacity and response time improvements offered by \pn come with the cost of additional energy consumption.
However, industrial applications that require higher data rates usually use grid-powered sensors.
Additionally, many applications use a powered backbone to carry high data-rate traffic while including some battery-powered nodes (e.g., \cite{burns2018airtight,agarwal2011duty,chipara2010reliable}).
\pn can be configured in such scenarios to use service lists size of size one on battery-operated devices to achieve the same energy consumption level as existing scheduling approaches.
Moreover, \pn could use larger service list sizes for the powered nodes to provide higher throughput and lower latency.


\section{Related Work}
\label{sec:related-work}


Due to its predictability, TDMA has become the de facto standard for IIoT systems.
There are many scheduling algorithms to construct TDMA schedules (e.g., \cite{deadlineScheduling,conflictAwareLeastLaxityFirst,Pottner2014,tschSchedulerReviewerWanted,TELESHERMETO201784}).
However, a common weakness of TDMA protocols is their lack of adaptability to network dynamics. 
To address this limitation, various techniques to handle variations in link quality, topology changes, and fluctuations in workloads have been proposed (e.g. \cite{bmax, rewimo, rdpas}).
In this paper, we focus on handling variations in link quality, as they are common in harsh industrial environments~\cite{candell2017industrial, emInterference2}.
Our work is complementary to and may be integrated with techniques designed to handle other types of network dynamics.

Researchers have considered various approaches to combining CSMA and TDMA into hybrid protocols, ultimately sacrificing either flexibility or predictability.
A common approach to combine CSMA and TDMA is to have each protocol run in different slots.
This approach is adopted in industrial standards such as WirelessHART \cite{wirelesshart} and ISA100.11a \cite{isa100}. 
However, predictable performance cannot be provided for the traffic carried in CSMA slots.
Another alternative is to dynamically reuse slots (e.g., \cite{rhee2008z}) or transmit high-priority traffic (e.g., \cite{li2015incorporating}) by selecting primary and secondary slot owners.  
In this approach, slot owners are given preference to transmit and send data using a short initial back-off.
If a slot owner does not have any data to transmit, other nodes may contend for its use after some additional delay.
A generalization of this scheme is prioritized MACs that divide a slot into sub-slots to provide different levels of priority \cite{shen2013prioritymac}. 
However, none of these protocols provide analytical bounds on their performance.
In contrast to the above approaches that involve carrier sensing, our policies rely on receiver-initiated polling and the local state of nodes to adapt. 
We expect policies to be less brittle in practice than solutions that use carrier sense as they do not require tight time synchronization for adaptation.

Several distributed protocols for constructing TSCH schedules that support best-effort  \cite{duquennoy2015orchestra, tinka2010decentralized} and real-time \cite{FD-PaS} traffic have been proposed.
Our work is complementary since these works focus primarily on handling workload changes while we focus on adapting to variations in link quality over short time scales.
These protocols can't adapt at the time scales required to handle link quality variations due to their communication overheads.
Our approach combines offline policy synthesis with local adaptation performed at run-time.
This approach can effectively handle changes over short time scales as the adaptation process is local and lightweight.


Transient link failures are common in wireless networks~\cite{betafactor, temporalprops_cerpa} and even more prevalent in harsh industrial environments~\cite{candell2017industrial, emInterference2}.
The state-of-the-art is to  schedule a fixed number of retransmissions for each link, potentially using different channels.
Little consideration is usually given to selecting the correct number of retransmissions based on link quality.
Recently, some work has been done to tune the number of retransmissions based on the burstiness of links~\cite{bmax, yang2017reliable}.
While this is a step in the right direction, the fundamental problem is that links are treated in isolation and provisioned to handle worst-case behavior in a fixed manner.
As a result, retransmissions cannot be redistributed across links as needed at run-time.
A notable exception is our prior work \cite{brummet2018flexible}, which proposes a technique to share transmissions among the links of a flow at run-time.
However, this technique's performance benefits are sensitive to the length of flows, with the most benefit occurring in large multi-hop networks uncommon in practice. 
Our experiments show that this approach is only effective when flows are routed through the base station and not for the more common data collection and dissemination scenarios.
By enabling entries to be \emph{shared across flows}, we can significantly reduce the number of slots needed by flows to meet their end-to-end reliability, resulting in significant performance improvements.

\section{Conclusions}
\label{sec:conclusions}

\pn is a practical and effective solution for IIoT applications that require predictable, real-time, and reliable communication in dynamic wireless environments.
We leverage the stability of IIoT workloads and the improving resources of wireless nodes to build a solution that combines offline policy construction and run-time adaptation.
A \pn policy assigns a \pn operation to each slot and channel, which specifies a coordinator that will arbitrate channel access and a list of flows that may be serviced.
At run-time, the coordinator dynamically executes the flows in the service list from which it has not received a packet.
The advantage of \pn is that nodes can locally reallocate the retransmissions of flows in response to variations in link quality and, as a result, provide higher performance than scheduling approaches.

The synthesis of policies required us to address two key challenges: handling the state explosion problem and providing predictable performance as the quality of links varies.
We developed a practical approach to synthesize policies iteratively.
In each slot, the \builder employs an ILP program to determine the \pn operations that will be performed in the current slot.
Based on the selected operations, the \evaluator determines a lower-bound on the end-to-end reliability of each flow to determine if it met its target end-to-end reliability. 
A key advantage of \pn is that it provides guarantees when slots are shared under a realistic model of wireless communication.
Specifically, we guarantee that a constructed \pn policy will meet a user-specified reliability and deadline constraint for each flow as long as the quality of all (used) links exceeds a minimum link quality.


We have extensively evaluated the performance of \pn through both simulations and testbed experiments.
Our results indicate that due to their increased agility, \pn policies can significantly improve real-time capacity (median 50\% -- 142\%) and reduce worst-case response time (median 27\% -- 70\%) while meeting a specified end-to-end reliability.
These trends hold across typical IIoT workloads, including data collection, data dissemination, and route through the base station.
Additionally, we showed empirically that our theoretical guarantees of real-time performance and reliability hold even in the presence of significant interference.

\section*{Acknowledgement}
This work is funded in part by NSF under CNS-1750155.

\bibliographystyle{ACM-Reference-Format}
\bibliography{main}

\newpage
\section{Proof of Theorem \ref{th:monotonic}}
\label{sec:proof}

In this section, we prove Theorem \ref{th:monotonic}.
Before proving the theorem though, we will introduce some definitions and lemmas.
We will illustrate their use using a single-hop scenario with two flows $F_0$ and $F_1$ ($\flows = \{F_0, F_1\}$) that relay data to the base station (see Figure \ref{fig:topology-example}).
In the following, we let $N = |\flows|$.
We consider the execution of two generic instances -- $J_0$ and $J_1$ -- of these flows .

Under the considered example, the state of the system is represented as a vector where the $i^{th}$ entry indicates whether the currently released instance of flow $i$  was received successfully (\success) or not (\failure) by the base station.
Accordingly, the states of our example are \texttt{FF}, \texttt{SF}, \texttt{FS}, and \texttt{SS}.
There are four possible \pullops that may be performed in a slot $t$: \pull{A}{$J_0$}, \pull{A}{$J_1$}, \pull{A}{$J_0, J_1$}, and \pull{A}{$J_1, J_0$}.
Note that the \builder described in Section \ref{sec:single-hop} would never assign \pull{A}{$J_1, J_0$} as it strictly enforces prioritization among flows.
Nevertheless, the theorem and lemmas presented in this section apply to a broader class of builders that allow priority inversions and may assign \pull{A}{$J_1, J_0$}. 
For each \pullop, we construct an associated transition matrix according to Algorithm \ref{algo:matrix}:
\begin{itemize}
    \item \tmat{0} -- the transition matrix associated with \pull{A}{$J_0$}
    \item \tmat{1} -- the transition matrix associated with \pull{A}{$J_1$}
    \item \tmat{0,1} -- the transition matrix associated with \pull{A}{$J_0, J_1$}
    \item \tmat{1,0} -- the transition matrix associated with \pull{A}{$J_1, J_0$}
\end{itemize}
\noindent Each of the matrices for the considered example are included in Figure \ref{fig:example-transition-matrices}.
Note that all of the transition matrices depend on the quality of the links \linkqt{0} and \linkqt{1} at time $t$.

\begin{figure}[h!]
\tiny
\centering
\begin{subfigure}[b]{0.45\textwidth}
\begin{displaymath}
    \begin{blockarray}{ccccc}
    \texttt{FF} & \texttt{SF} & \texttt{FS} & \texttt{SS} \\
    \begin{block}{(cccc)c}
      1 - \linkqt{0} & \linkqt{0} & 0          & 0 & \texttt{FF} \\
      0           & 1       & 0          & 0 & \texttt{SF} \\
      0           & 0       & 1 - LQ_{0}(t) & LQ_{0}(t) & \texttt{FS} \\
      0           & 0       & 0          & 1 & \texttt{SS} \\
    \end{block}
\end{blockarray}
\end{displaymath}
\caption{$\mathcal{M}_{0}$}
\label{fig:proof:m0}
\end{subfigure}
\quad
\begin{subfigure}[b]{0.45\textwidth}
\begin{displaymath}
    \begin{blockarray}{ccccc}
    \texttt{FF} & \texttt{SF} & \texttt{FS} & \texttt{SS} \\
    \begin{block}{(cccc)c}
      1 - \linkqt{1} & 0   & \linkqt{1}          & 0 & \texttt{FF} \\
      0           & 1 - \linkqt{1} & 0          & \linkqt{1} & \texttt{SF} \\
      0           & 0         & 1          & 0 & \texttt{FS} \\
      0           & 0         & 0          & 1 & \texttt{SS} \\
    \end{block}
\end{blockarray}
\end{displaymath}
\caption{$\mathcal{M}_{1}$}
\label{fig:proof:m1}
\end{subfigure}
\quad
\begin{subfigure}[b]{0.45\textwidth}
\begin{displaymath}
    \begin{blockarray}{ccccc}
    \texttt{FF} & \texttt{SF} & \texttt{FS} & \texttt{SS} \\
    \begin{block}{(cccc)c}
      1 - \linkqt{0} & \linkqt{0}     & 0          & 0 & \texttt{FF} \\
      0           & 1 - \linkqt{1} & 0          & \linkqt{1} & \texttt{SF} \\
      0           & 0           & 1 - \linkqt{0} & \linkqt{0} & \texttt{FS} \\
      0           & 0           & 0          & 1 & \texttt{SS} \\
    \end{block}
\end{blockarray}
\end{displaymath}
\caption{$\mathcal{M}_{0,1}$}
\label{fig:proof:m01}
\end{subfigure}
\quad
\begin{subfigure}[b]{0.45\textwidth}
\begin{displaymath}
    \begin{blockarray}{ccccc}
    \texttt{FF} & \texttt{SF} & \texttt{FS} & \texttt{SS} \\
    \begin{block}{(cccc)c}
      1 - \linkqt{1} & 0     & \linkqt{1}          & 0 & \texttt{FF} \\
      0           & 1 - \linkqt{1} & 0          & \linkqt{1} & \texttt{SF} \\
      0           & 0           & 1 - \linkqt{0} & \linkqt{0} & \texttt{FS} \\
      0           & 0           & 0          & 1 & \texttt{SS} \\
    \end{block}
\end{blockarray}
\end{displaymath}
\caption{$\mathcal{M}_{1,0}$}
\label{fig:proof:m10}
\end{subfigure}
\caption{Possible transition matrices when two flows are active}
\label{fig:example-transition-matrices}
\end{figure}

According to Equation \ref{eq:state-transition}, the network state after executing $t$ \pullops is:
\begin{displaymath}
    \pstate{t} = s_0^{T} \tmat{srv(0)} \tmat{srv(1)} \cdots \tmat{srv(t)} 
\end{displaymath}
\noindent where $s_0$ is an initial state and $\tmat{srv(t')}$ is the transition matrix associated with the \pullop performed in slot $t'$, $0 \leq t' \leq t$, and in our example is equal to either $\tmat{0}$, $\tmat{1}$, $\tmat{0, 1}$ or $\tmat{1, 0}$. 
This equation describes the state evolution of a Markov Chain (MC) over time.
Note that unlike traditional MCs, the transition matrix of this MC is parametric and the value of those parameters change over time.

The transition matrices have a special structure which we will characterize next. 
We impose a partial order on the states that reflects how the network changes its state in response to a successful \pullops (see procedure \onSuccess{} of Algorithm \ref{algo:matrix}).

\begin{definition}
We say the states $s_1$ and $s_2$ are partially ordered, $s_1 \preceq s_2$, if and only if the following is true:
\begin{displaymath}
 s_1[k] = \success \Rightarrow s_2[k] = \success~~~~~\forall k \in [0, N) 
\end{displaymath}
\end{definition}

\noindent
The partial order induced by $\preceq$ in our example is: 
$\texttt{FF} \preceq \texttt{SF} \preceq \texttt{SS}$ and $\texttt{FF} \preceq \texttt{FS} \preceq \texttt{SS}$.
The states \texttt{SF} and \texttt{FS} are not comparable. 
Relating $\preceq$ to the \onSuccess{} method, the ordering $\texttt{FF} \preceq \texttt{SF}$ implies that there is a service list $srv$ (e.g., $srv=\{J_0\}$ or $srv=\{J_0, J_1\}$) such that \onSuccess{\texttt{FF}, $J_0$} = \texttt{SF}.
We make two observations of this partial order:

\begin{lemma}
$s_1 \preceq \onSuccess{$s_1, J_k$}$ for all instances $J_k$.
\label{lemma:next-order}
\end{lemma}
\begin{proof}
\onSuccess can change only the $k^{th}$ entry in $s_1$ to \texttt{S}.
If $s_1[k] =$ \texttt{S} then the partial order holds as the state will not change (i.e. $s_1 = \onSuccess(s_1,J_k)$).
If $s_1[k] =$ \texttt{F}, then the $k^{th}$ entry in $s_1$ will change to \texttt{S} and all other entries will stay the same.
This also does not violate the partial order.
\end{proof}

\begin{lemma}
If $s_1 \preceq s_2$, then $\onSuccess{$s_1, {J_k}$} \preceq \onSuccess{$s_2, J_k$}$, for all instances $J_{k}$. 
\label{lemma:next-state}
\end{lemma}
\begin{proof}
\onSuccess can change only the $k^{th}$ entry of a state so there are four possibilities.
(1) If $s_1[k] =$ \texttt{S} and $s_2[k] =$ \texttt{S} then $s_1 = \onSuccess(s_1, J_k)$ and $s_2 = \onSuccess(s_2, J_k)$.
Therefore, $\onSuccess{$s_1, {J_k}$} \preceq \onSuccess{$s_2, J_k$}$.
(2) If $s_1[k] =$ \texttt{F} and $s_2[k] =$ \texttt{F} the $k^{th}$ entry of $s_1$ and $s_2$ will change to \texttt{S} and all other entries will stay the same.  
Therefore, $\onSuccess{$s_1, {J_k}$} \preceq \onSuccess{$s_2, J_k$}$.
(3) If $s_1[k] =$ \texttt{S} and $s_2[k] =$ \texttt{F} the assumed partial ordering is violated and therefore the lemma is not violated.
(4) If $s_1[k] =$ \texttt{F} and $s_2[k] =$ \texttt{S} then the  $k^{th}$ entry of $s_1$ will change to \texttt{S} with all other entries staying the same and $s_2 = \onSuccess(s_2, J_k)$.
Since $s_2[k] =$ \texttt{S}, $\onSuccess{$s_1, {J_k}$} \preceq \onSuccess{$s_2, J_k$}$.
\end{proof}

We will use the notation $\mathcal{M}_{srv(t)}[i,j]$ to refer to the $i,j$ element of the matrix and $\mathcal{M}_{srv(t)}[i, :]$ to refer to the i$^{th}$ row.
The values of $\mathcal{M}_{srv(t)}[i,:]$ include the likelihood of transitioning from $s_i$ to another state in \allstates.
The values of a row follow one of two patterns:
(1) 
If the current state is $s_i$, $J_k$ is an instance in the current service list to be executed such that $s_i[k] =$ \texttt{F}, and $s_j = \onSuccess(s_i, J_k)$, then all entries in
$\mathcal{M}_{srv(t)}[i,:]$ are zero except for $\mathcal{M}_{srv(t)}[i,i] = 1 - LQ_k(t)$ and $\mathcal{M}_{srv(t)}[i, j] = LQ_k(t)$.
(2) Otherwise if the current state is $s_i$ there is only one non-zero entry in $\mathcal{M}_{srv(t)}[i,:]$ and it is $\mathcal{M}_{srv(t)}[i,i] = 1$.
Based on these observations, we can rewrite $\mathcal{M}_{srv(t)}$ as:

\begin{equation}
    \mathcal{M}_{srv(t)} = \boldsymbol{I} + LQ_0(t) E_{0} + LQ_1(t) E_{1} + \dots + LQ_{N}(t)E_{N} = \boldsymbol{I} + \sum_{i=0}^{N}{LQ_i(t) E_i}
\end{equation}

\noindent where $\boldsymbol{I}$ is the identity matrix and matrix $E_i(t)$ has the following properties:
(1) $E_i(t)$ is upper-triangular,
(2) the entries of $E_i(t)$ are in $\{-1, 0, 1\}$ and
(3) in each row, $E_i(t)[i,:]$, there is either exactly one +1 entry off the diagonal and one -1 entry on the diagonal or all the entries of the row are zero.
As an example, the transition matrix $\tmat{0,1}$ may be rewritten as:

\begin{align*}
    \tmat{0,1} &= \boldsymbol{I} + \linkq{0}(t)E_0(t) + \linkq{1}(t)E_1(t) \\
               &=
               \begin{pmatrix}
 1 & 0 & 0 & 0 \\
 0 & 1 & 0 & 0 \\
 0 & 0 & 1 & 0 \\
 0 & 0 & 0 & 1
 \end{pmatrix}
 +
 \linkq{0}(t)
 \begin{pmatrix}
 -1 & 1  &  0 & 0 \\
  0 & 0  &  0 & 0 \\
  0 & 0  & -1 & 1 \\
  0 & 0  &  0 & 0 \\
 \end{pmatrix}
 +
 \linkq{1}(t)
 \begin{pmatrix}
 0  & 0  & 0  & 0  \\
 0  & -1 & 0  & 1  \\
 0  & 0  & 0  & 0  \\
 0  & 0  & 0  & 0  \\
 \end{pmatrix}
\end{align*}

We now create the following definition to relate the partial ordering to the actual state probabilities and make the following two observations.

\begin{definition}
A vector $\boldsymbol{f}$ given the partial order induced by $\preceq$, if $s_i \preceq s_j$ implies $\boldsymbol{f}[i] \le \boldsymbol{f}[j]$.
\label{definition:increasing-vector}
\end{definition}

\begin{lemma}
If $\boldsymbol{f}^T$ is an increasing vector and $\mathcal{M}_{srv(t)}$ is a transition matrix , then $\boldsymbol{g}^T = \boldsymbol{f}^T \mathcal{M}_{srv(t)}^T$ is also an increasing vector.
\label{lemma:increasing}
\end{lemma}
\begin{proof}
Consider an arbitrary instance $J_k$ and let $s_i \preceq s_j$, $s_a = \onSuccess(s_i, J_k)$, and $s_b = \onSuccess(s_j, J_k)$. 
Consider the $i^{th}$ and $j^{th}$ entries of $\boldsymbol{g}^T$:

\begin{align*}
    \boldsymbol{g}^T[i] &=
    \boldsymbol{f}^T[i](1 - LQ_k(t)) + \boldsymbol{f}^T[a]LQ_k(t) \\
    \boldsymbol{g}^T[j] &= \boldsymbol{f}^T[j](1 - LQ_k(t)) + \boldsymbol{f}^T[b]LQ_k(t)
\end{align*}

\noindent
Notice that $\boldsymbol{f}^T[i] \leq \boldsymbol{f}^T[j]$ by definition since $s_i \preceq s_j$ and $\boldsymbol{f}^T[a] \leq \boldsymbol{f}^T[b]$ by Lemma \ref{lemma:next-state}.
As a result, we can conclude $\boldsymbol{g}^T[i] \leq \boldsymbol{g}^T[j]$.
Since $\boldsymbol{g}^T[i] \leq \boldsymbol{g}^T[j]$ holds for an arbitrary instance $J_k$, $\boldsymbol{g}^T$ must be an increasing vector.
\end{proof}

\begin{lemma}
If $\boldsymbol{f}^T$ is an increasing vector, $\boldsymbol{g}^{T} = \boldsymbol{f}^T \mathcal{M}_{srv(t)}^T$, and $\boldsymbol{g'}^T = \boldsymbol{f}^{T} \mathcal{\widehat{M}}_{srv(t)}^T$ with 
$\mathcal{\widehat{M}}_{srv(t)}= (\boldsymbol{I} + \sum_{i=0}^{N}{\pmin E_i(t)})$ and $LQ_i(t) \ge \pmin$, then $\boldsymbol{g}^{T} \geq \boldsymbol{g'}^{T}$ component-wise.
\label{lemma:lowerbound}
\end{lemma}
\begin{proof}
Consider $\boldsymbol{g}^{T} - \boldsymbol{g'}^{T}$:

\begin{align*}
    \boldsymbol{g}^{T} - \boldsymbol{g'}^{T} &= \boldsymbol{f}^{T} \mathcal{M}_{srv(t)}^T - \boldsymbol{f}^{T} \mathcal{\widehat{M}}_{srv(t)}^T\\
    &= \boldsymbol{f}^{T} \left( \boldsymbol{I} + \sum_{i=0}^{N}{LQ_i(t) E_i(t)}\right)^T - \boldsymbol{f}^{T} \left( \boldsymbol{I} + \sum_{i=0}^{N}{\pmin E_i(t)}\right)^T \\
    &= \left(\sum_{i=0}^{N}{(LQ_i(t) - m)E_i(t)}\right)\boldsymbol{f}
\end{align*}

\noindent
Consider now an arbitrary instance $J_k$ and state $s_i$ such that $s_a = \onSuccess(s_i, J_k)$.
By Lemma \ref{lemma:next-order},
$s_i \preceq s_a$.
Since $\boldsymbol{f}$ is an increasing vector (because $\boldsymbol{f}^T$ is an increasing vector), $\boldsymbol{f}[i] \leq \boldsymbol{f}[a] \implies 0 \leq \boldsymbol{f}[a] - \boldsymbol{f}[i]$.
Notice that either $E_i(t)[i,i] = E_i(t)[i,a] = 0$ or $E_i(t)[i,i] = -1$ and $E_i(t)[i,a] = 1$.

\noindent
If $E_i(t)[i,i] = E_i(t)[i,a] = 0$, then

\begin{align*}
    \left(\sum_{i=0}^{N}{(LQ_i(t) - m)E_i(t)}\right)[i,:]\boldsymbol{f} = 0
\end{align*}

\noindent
If instead $E_i(t)[i,i] = -1$ and $E_i(t)[i,a] = 1$ then

\begin{align*}
    \left(\sum_{i=0}^{N}{(LQ_i(t) - m)E_i(t)}\right)[i,:]\boldsymbol{f} &=
    (LQ_i(t) - m)\boldsymbol{f}[a] - (LQ_i(t) - m)\boldsymbol{f}[i] \\
    &\geq 0
\end{align*}

\noindent
Since this result holds for an arbitrary instance $J_k$, $\boldsymbol{g}^{T} \geq \boldsymbol{g'}^{T}$ component-wise.
\end{proof}

We are now prepared to prove Theorem \ref{th:monotonic} which we reproduce below for convenience.

\begin{customthm}{2}
Consider a star topology that has node $A$ as a base station and a set of flows $\flows = \{F_0, F_1, \dots F_N\}$ that have $A$ as destination.
Let \linkqt{0}, \linkqt{1}, \dots \linkqt{N} be the quality of the links used by each flow in slot $t$ such that $\pmin \le LQ_i(t) \le 1$ for all flows $F_i$ ($F_i \in \flows$) and all slots $t$ ($t \in \mathbb{N}$).
Under these assumptions, the reliability $\reliability{i,t}$ of an instance $J_i$ after executing $t$ \pullops of the \pn policy $\pi$ is lower bounded by $\reliabilitylb{i,t}.$
\end{customthm}
\begin{proof}
The end-to-end reliability of flow instance $J_i$ after $t$ slots is:
\begin{displaymath}
    \reliability{i,t} =  \pstate{t} \boldsymbol{\chi}_{i} =  s_0^{T} \tmat{srv(0)} \tmat{srv(1)} \cdots \tmat{srv(t)} \boldsymbol{\chi}_{i}
\end{displaymath}

\noindent
Since $R_{i,t}$ is a number, we can apply the transpose to obtain:
\begin{align*}
    R_{i,t} &= (R_{i,t})^T  \nonumber \\
            &= \boldsymbol{\chi}_{i}^{T} \mathcal{M}_{srv(0)}^{T} \mathcal{M}_{srv(1)}^{T} \cdots \mathcal{M}_{srv(t)}^{T} s_0     ~~~~~~~~~~~~~
\end{align*}

\noindent
We observe that $\boldsymbol{\chi}_{i}$ is an increasing vector by construction, and by extension, $\boldsymbol{\chi}_{i}^T$.
By Lemma \ref{lemma:lowerbound} the following must be true as a result:

\begin{align*}
    R_{i,t} &= \boldsymbol{\chi}_{i}^{T} \mathcal{M}_{srv(0)}^{T} \mathcal{M}_{srv(1)}^{T} \cdots \mathcal{M}_{srv(t)}^{T} s_0\\
    &\geq
    \boldsymbol{\chi}_{i}^{T} \mathcal{\widehat{M}}_{srv(0)}^{T} \mathcal{M}_{srv(1)}^{T} \cdots \mathcal{M}_{srv(t)}^{T} s_0
\end{align*}

\noindent
As a consequence of Lemma \ref{lemma:increasing}, $\boldsymbol{\chi}_{i}^{T} \mathcal{\widehat{M}}_{srv(0)}^{T}$ is an increasing vector and therefore we can again apply Lemma \ref{lemma:lowerbound} to get the following:

\begin{align*}
    R_{i,t} &\geq
    \boldsymbol{\chi}_{i}^{T} \mathcal{\widehat{M}}_{srv(0)}^{T} \mathcal{\widehat{M}}_{srv(1)}^{T} \cdots \mathcal{M}_{srv(t)}^{T} s_0
\end{align*}

\noindent
Continuing in this way gives the desired result

\begin{align*}
    R_{i,t} &\geq
    \boldsymbol{\chi}_{i}^{T} \mathcal{\widehat{M}}_{srv(0)}^{T} \mathcal{\widehat{M}}_{srv(1)}^{T} \cdots \mathcal{\widehat{M}}_{srv(t)}^{T} s_0 \\
    &= \widehat{R}_{i,t}
\end{align*}
\end{proof}

\end{document}